\documentclass[prb,aps,twocolumn,superscriptaddress,longbibliography,amsmath,amsmath,amssymb]{revtex4-1}

\bibpunct{[}{]}{,}{n}{}{}

\usepackage{graphicx}
\usepackage{orcidlink}
\usepackage{amsmath,amsthm,amssymb,amsfonts,bm,color,mathtools}
\usepackage{float}
\usepackage{microtype}
\usepackage{epstopdf}
\usepackage{color}
\usepackage{physics}
\usepackage{url}
\usepackage{mathbbol}
\usepackage{newtxtext,newtxmath}
\usepackage{dsfont}   
\usepackage{pifont}  
\usepackage{placeins}
\usepackage{braket}
\usepackage{comment} 
\usepackage{stackrel}
\graphicspath{{figures/}}
\usepackage{hyperref}
\hypersetup{colorlinks=true,linkcolor=magenta,citecolor=cyan,urlcolor=blue,filecolor=magenta}
\usepackage[capitalize]{cleveref}

\newcommand{\bk}{\bm{k}}
\newcommand{\bq}{\bm{q}}

\newcommand{\bU}{\bm{U}}

\newcommand{\lee}{{\{8{,}8\}}}

\newcommand{\HTG}{\Gamma} 
\newcommand{\GM}[1]{{\mathbbm{\Gamma}}_{#1}} 
\newcommand{\htgel}{\eta} 
\newcommand{\HTGuc}{\HTG^{(1)}} 
\newcommand{\HTGsc}{\HTG^{(2)}}
\newcommand{\HTGnc}{\HTG^{(n)}}
\newcommand{\htg}[1]{\gamma_{#1}} 
\newcommand{\HPG}{G} 
\newcommand{\hpgel}[1]{g} 
\newcommand{\nsg}{\tilde{\Gamma}} 
\newcommand{\TG}{\Delta}                    
\newcommand{\gpres}[2]{\left.\left\langle #1\vphantom{#2} \right| #2 \right\rangle} 
\newcommand{\genus}{\mathfrak{g}}           
\newcommand{\tgquot}[2]{\mathrm{T}#1.#2}
\newcommand{\nsubg}{\vartriangleleft}
\newcommand{\gbusn}{\vartriangleright}       

\usepackage[normalem]{ulem}

\newtheorem{theorem}{Theorem}[section]
\newtheorem{lemma}[theorem]{Lemma}

\DeclareMathOperator{\lcm}{lcm}
\newcommand{\nn}{\nonumber}
\DeclareSymbolFontAlphabet{\mathbbm}{bbold}

\allowdisplaybreaks
\makeatletter
\def\maketitle{
\@author@finish
\title@column\titleblock@produce
\suppressfloats[t]}
\makeatother

\makeatletter
\def\bibsection{%
   \par
   \begingroup
    \baselineskip26\p@
    \bib@device{\hsize}{72\p@}%
   \endgroup
   \nobreak\@nobreaktrue
   \addvspace{19\p@}%
  }%
\makeatother

\makeatletter
\def\clearfmfn{\let\@FMN@list\@empty}   
\makeatother

\begin{document}

\title{Hyperbolic Non-Abelian Semimetal}

\author{Tarun Tummuru\,\orcidlink{0009-0001-7196-2763}}
\thanks{These authors contributed equally to this work.}
\affiliation{Department of Physics, University of Zurich, Winterthurerstrasse 190, 8057 Zurich, Switzerland}

\author{Anffany Chen\,\orcidlink{0000-0002-0926-5801}}
\thanks{These authors contributed equally to this work.}
\affiliation{Department of Physics \& Theoretical Physics Institute, University of Alberta, Edmonton, Alberta T6G 2E1, Canada}

\author{Patrick M. Lenggenhager\,\orcidlink{0000-0001-6746-1387}}
\thanks{These authors contributed equally to this work.}
\affiliation{Department of Physics, University of Zurich, Winterthurerstrasse 190, 8057 Zurich, Switzerland}
\affiliation{Condensed Matter Theory Group, Paul Scherrer Institute, 5232 Villigen PSI, Switzerland} 
\affiliation{Institute for Theoretical Physics, ETH Zurich, 8093 Zurich, Switzerland}
\affiliation{Max Planck Institute for the Physics of Complex Systems, Nöthnitzer Str. 38, 01187 Dresden, Germany}

\author{Titus Neupert\,\orcidlink{0000-0003-0604-041X}}
\affiliation{Department of Physics, University of Zurich, Winterthurerstrasse 190, 8057 Zurich, Switzerland}

\author{Joseph Maciejko\,\orcidlink{0000-0002-6946-1492}}
\affiliation{Department of Physics \& Theoretical Physics Institute, University of Alberta, Edmonton, Alberta T6G 2E1, Canada}
\affiliation{Quantum Horizons Alberta, University of Alberta, Edmonton, Alberta T6G 2E1, Canada}

\author{Tom\'{a}\v{s} Bzdu\v{s}ek\,\orcidlink{0000-0001-6904-5264}}
\email{tomas.bzdusek@uzh.ch}
\affiliation{Department of Physics, University of Zurich, Winterthurerstrasse 190, 8057 Zurich, Switzerland}
\affiliation{Condensed Matter Theory Group, Paul Scherrer Institute, 5232 Villigen PSI, Switzerland}



\begin{abstract}
We extend the notion of topologically protected semi-metallic band crossings to hyperbolic lattices in a negatively curved plane. Because of their distinct translation group structure, such lattices are associated with a high-dimensional reciprocal space. In addition, they support non-Abelian Bloch states which, unlike conventional Bloch states, acquire a matrix-valued Bloch factor under lattice translations. Combining diverse numerical and analytical approaches, we uncover an unconventional scaling in the density of states at low energies, and illuminate a nodal manifold of codimension five in the reciprocal space. The nodal manifold is topologically protected by a nonzero second Chern number, reminiscent of the characterization of Weyl nodes by the first Chern number.
\end{abstract}

\maketitle


\emph{Introduction.---}Topological semimetals are gapless electronic phases with stable band-touching points at the Fermi energy. 
As solid-state analogs of relativistic fermions, low-energy excitations in semimetals are often imbued with nontrivial topological properties. 
Dirac and Weyl semimetals in three dimensions (3D), for instance, exhibit a range of intriguing properties such as protected surface states, anomalous Hall effect, chiral anomaly, and unusual magneto-resistance~\cite{Wan:2011,Zyuzin:2012,Hosur:2013,Burkov:2014,Vazifeh:2013,vafek2014, Armitage2018,ong2021}. 
Broadly, semimetals can be categorized based on spatial dimension, origin of band crossings, and properties of the nodal manifold (e.g., its degeneracy and codimension)~\cite{Gao2019}. 
All existing models and material realizations, however, are typically embedded in geometrically flat Euclidean space.

\begin{figure}[t]
    \includegraphics[width=\columnwidth]{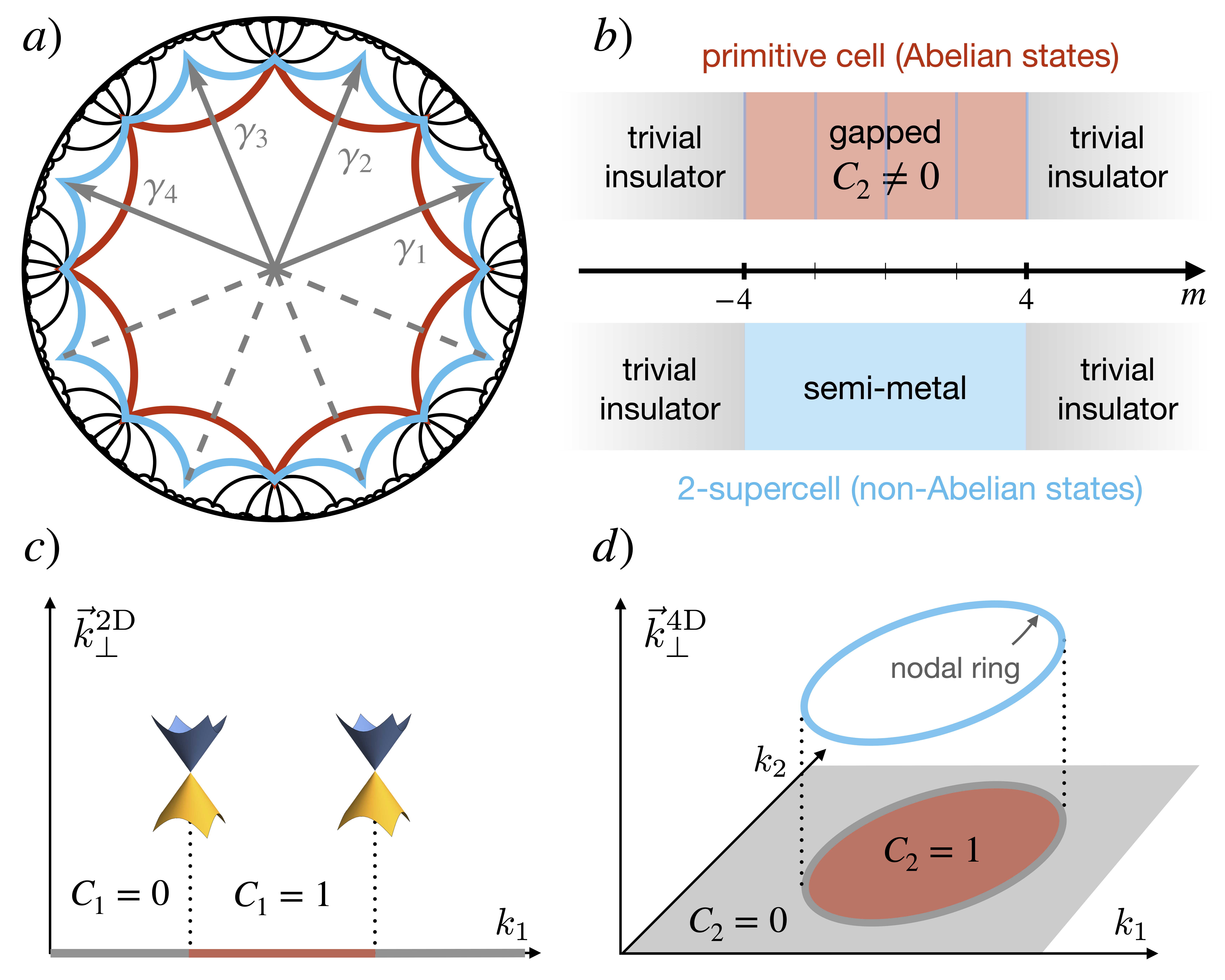}
    \caption{Real- and reciprocal-space structure of the hyperbolic non-Abelian semimetal. 
    (a)~The model is defined on the $\lee$ lattice, displayed in the Poincar\'{e} disk model, wherein the single- and two-site unit cells are marked in red and blue, respectively. 
    Translation generators of the primitive cell are denoted by $\gamma_j$. 
    (b)~Unit cell size affects the accessible reciprocal space and, consequently,
    the phase diagram for $|{m}|\,{<}\,4$: the primitive cell suggests topological insulator phases with Chern numbers $|C_2|\,{=}\,1$ or $3$ whereas the $2$-supercell description exhibits a topological semimetal. 
    (c)~Comparison to Weyl semimetals, where 2D planes (along $\vec{k}_{\perp}^{\rm 2D}$) between the Weyl nodes have a finite first Chern number. (d)~Similarly, in the hyperbolic semimetal, when $|m| \,{\in}\, (3,4)$, a nodal ring separates 4D orthogonal subspaces (along $\vec{k}_{\perp}^{\rm 4D}$) with a nontrivial second Chern number.}
    \label{fig1} 
\end{figure}

Hyperbolic lattices---regular tessellations of 2D hyperbolic space with constant negative curvature~\cite{Coxeter:1957}---provide the means to study quantum matter in non-Euclidean geometry.
Recent realizations in coplanar waveguide resonators~\cite{Kollr2019} and topoelectrical circuits~\cite{Chen2023a,Lenggenhager2022,Zhang:2022,Zhang_2023} have catalyzed the research of physical properties of hyperbolic lattices. 
Theoretical studies of hyperbolic analogs of topological insulators~\cite{Yu2020,urwyler2022hyperbolic, Liu2022,Zhang:2022,Zhang_2023,Pei:2023,Chen2023b} show that bulk topological invariants and protected edge modes persist in hyperbolic space. 
Other celebrated condensed-matter phenomena such as Hofstadter spectra~\cite{ikeda2021,stegmaier2022}, flat bands~\cite{kollar2020,saa2021,Bzdusek2022,mosseri2022}, higher-order topology~\cite{liu2023,tao2023,sun2023},  strong correlations~\cite{zhu2021,bienias2022,gluscevich2023,gluscevich2023b}, and fractons~\cite{yan2019,yan2022} have also been explored in the hyperbolic~context.

Hyperbolic lattices have unique features that unlock new physics beyond the Euclidean paradigm. 
Hyperbolic translation groups, known as Fuchsians, are non-Abelian~\cite{Boettcher:2022}. 
Accordingly, their unitary irreducible representations (IRs) can have a dimension $d$ larger than one. 
Via a generalized Bloch theorem, hyperbolic band theory (HBT)~\cite{maciejko2021,Maciejko2022} provides a reciprocal-space description of hyperbolic lattices and their eigenstates. Hyperbolic Bloch states are classified as Abelian [$\textrm{U}(1)$]~\cite{maciejko2021} or non-Abelian [$\textrm{U}(d)$ for $d\,{>}\,1$]~\cite{Maciejko2022,Cheng2022} depending on the rank $d$ of the matrix-valued Bloch factor acquired under lattice translations. 
Whereas $\textrm{U}(1)$ eigenstates are labeled by a vector analogous to crystal momentum, additional quantum numbers are required to distinguish $\textrm{U}(d)$ Bloch states. 
The problem of accessing and characterizing non-Abelian Bloch states constitutes an actively investigated topic~\cite{Lux:2022, Lux2023,Mosseri2023,Lenggenhager2023}. 

In this work, we present a hyperbolic lattice model whose semimetallic topology is inherently rooted in non-Abelian Bloch states. 
While the Abelian spectrum is gapped, we demonstrate the presence of a topologically protected nodal manifold in the non-Abelian $\textrm{U}(2)$ Bloch bands, which results in a gapless density of states (DOS). 
Summarized in \cref{fig1}, this nodal manifold in reciprocal space is associated with a finite second Chern number $C_2$. 
There is no Euclidean analog of such a phase because 2D Euclidean lattices admit neither non-Abelian Bloch states nor a second Chern number.
We term this phase a \emph{hyperbolic non-Abelian semimetal}. 


\emph{Model.---}HBT suggests a way to translate higher-dimensional Euclidean models to the hyperbolic plane. 
The mapping is especially lucid for 4D Euclidean Hamiltonians because of a $\textrm{one-to-one}$ correspondence between the Euclidean BZ and the simplest [i.e.~the $\textrm{U}(1)$/Abelian] hyperbolic BZ of the $\lee$ lattice---the hyperbolic lattice where eight regular octagons meet at each vertex. 
Motivated by this connection, we study the hyperbolic counterpart of the 4D quantum Hall insulator (QHI). 

Conventionally, the 4D Euclidean QHI is defined by a Dirac tight-binding model on the hypercubic lattice \cite{zhang2001, Qi2008}:
\begin{align}
    \mathcal{H} = \sum_{r, j} \left[ \psi_r^\dagger \frac{\GM{5} - \mathrm{i} \GM{j}}{2} \psi_{r+j} + {\rm h.c.} \right] + m \sum_r \psi_r^\dagger \GM{5} \psi_r,
    \label{eq:Hr}
\end{align}
where $\{\GM{\mu} \}_{\mu\,{=}\,1}^{5}$ denote the Dirac matrices, $\{ \GM{\mu}, \GM{\nu} \} = 2 \delta_{\mu \nu} \mathbbm{1}$. 
In terms of the Pauli matrices $\sigma_{i}$, the choice $\{\GM{\mu} \}_{\mu\,{=}\,1}^{5} = \{\sigma_{1} \sigma_{0}, \sigma_{2} \sigma_{0}, \sigma_{3} \sigma_{1}, \sigma_{3} \sigma_{2}, \sigma_{3} \sigma_{3} \}$ ensures that $\GM{\mu}$ is real (imaginary) for $\mu$ odd (even). 
The four components of the spinor $\psi_r$ correspond to internal degrees of freedom at site $r\,{\in}\,\mathbb{Z}^4$, and $j\,{=}\,1,\ldots,4$ label orthogonal directions. 

The $\lee$ hyperbolic lattice, illustrated in \cref{fig1}(a), provides a convenient setting to reinterpret the above model~\cite{Zhang_2023}. 
Formally, the lattice is defined by its Fuchsian group $\HTG$, an infinite translation group with four noncommuting generators $\gamma_j$ [arrows in \cref{fig1}(a)] that obey a single constraint~\cite{maciejko2021}:
\begin{equation}
    \HTG = \left<\gamma_1,\gamma_2,\gamma_3,\gamma_4\,|\,\gamma_1 \gamma_2^{-1} \gamma_3 \gamma_4^{-1} \gamma_1^{-1} \gamma_2 \gamma_3^{-1} \gamma_4 \,{=}\,1\right>.
    \label{eq:constraint}
\end{equation}
The `hyperbolized' 4D QHI Hamiltonian is realized with the following reinterpretation: $r$ is the unit cell position on the hyperbolic plane and $r\,{+}\,j$ denotes a cell shifted by the translation operator $\gamma_j$. 

To obtain a reciprocal-space perspective, we proceed in two steps. 
First, as in the Euclidean case, the lattice needs to be compactified, i.e., periodic boundary conditions (PBC) are imposed and edges of the unit cell related by translation symmetry are identified, producing a surface of genus $\genus \,{\geq}\, 2$~\cite{maciejko2021}. 
Unlike the Euclidean case, choosing an enlarged $n\textrm{-site}$ unit cell ($n\,{>}\,1$) increases the genus~\cite{Lenggenhager2023}. 
Second, one has to choose a $d$-dimensional IR $D_{\lambda}$ of the translation group, where $\lambda$ labels a point in a $[2d^2(\genus \,{-}\,1)\,{+}\,2]$-dimensional space of $d$-dimensional IRs of $\Gamma$~\cite{Maciejko2022}.
A generic Bloch Hamiltonian $H_{\lambda}^{(n,d)}$ thus obtained determines states in the Brillouin zone BZ$^{(n,d)}$. 
The integers $n$ and $d$ provide two handles to access different sectors of the hyperbolic reciprocal space. 

On the $\lee$ lattice, the smallest unit cell ($n\,{=}\,1$ or primitive cell) is given by the red octagon in \cref{fig1}(a). 
It compactifies to a $\genus\,{=}\,2$ surface when the opposite edges are identified. 
Further, in $\textrm{U}(1)$ representations, $\lambda$ are the crystal momenta whose components correspond to the four noncontractible cycles of the $\textrm{genus-}2$ surface. 
The mapping $D_{\lambda}(\gamma_j) \,{=}\, D_{\bk}(\gamma_j) \,{=}\, \mathrm{e}^{\mathrm{i} k_j}$ satisfies the constraint in \cref{eq:constraint} and defines the $\textrm{U}(1)$ Bloch Hamiltonian 
\begin{align}
    H^{(1,1)}_{\bk} = \sum_{\mu} d_{\bk}^{\mu} \GM{\mu},
    \label{eq:4D-QHI:Bloch-Ham}
\end{align}
with $d_{\bk}^{\mu} = (\sin k_1, \sin k_2, \sin k_3, \sin k_4, m_{\bk})$ and $m_{\bk} \,{=}\, m+\sum_j \cos k_j$. 
Evidently, these expressions are equivalent to the Fourier transform of the Hamiltonian in \cref{eq:Hr}, such that $\textrm{U}(1)$-HBT reproduces results from the 4D Euclidean case: at half-filling, the model undergoes topological phase transitions at $|m|\,{=}\,0,2,4$, but is gapped otherwise with a nonvanishing $C_2$~\cite{Qi2008}, as depicted in \cref{fig1}(b).
The corresponding gapped DOS for $m\,{=}\,3$ is plotted with the red curve in Fig.~\ref{fig3}(a).

However, as foreshadowed earlier, one should also account for $d\,{>}\,1$ IRs of $\HTG$. 
We first demonstrate the existence of `in-gap' states by real-space diagonalization on PBC clusters, and then show that these states originate from non-Abelian reciprocal space using $\textrm{U}(2)$-HBT and the supercell method (\cref{fig2}).

\begin{figure}[t]
    \includegraphics[width=\columnwidth]{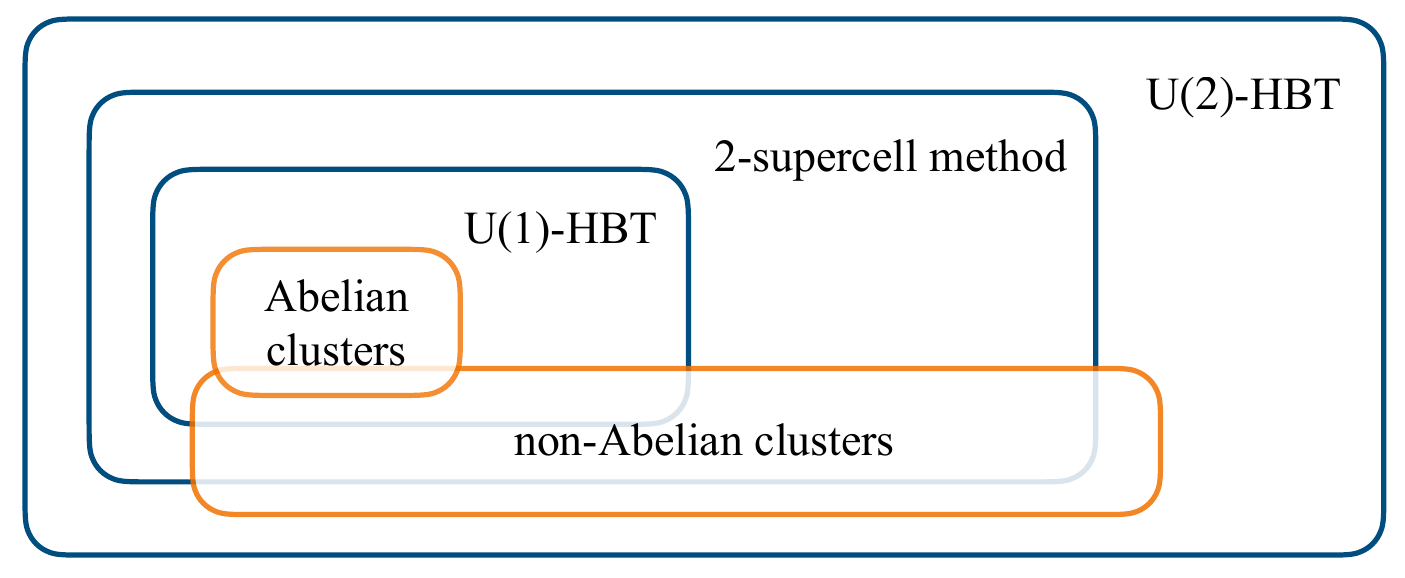}
    \caption{ 
    Schematic summary of techniques used in this work to access eigenstates of hyperbolic lattices. Orange (blue) indicates real- (reciprocal-)space methods. 
    }
    \label{fig2}
\end{figure}


\emph{PBC clusters.---}Closed hyperbolic lattices formed by pairing open edges of a finite-sized lattice are termed PBC clusters~\cite{sausset2007}.
Efficient construction of large PBC clusters can be achieved through computational methods in geometric group theory~\cite{SM}.
Depending on the particular edge pairing, PBC clusters can be classified into two categories: Abelian or non-Abelian~\cite{Maciejko2022}.
In Abelian clusters, the non-Abelian hyperbolic translation group $\HTG$ modulo the edge pairing becomes Abelian.
In other words, the translation generators effectively commute on Abelian clusters, so the wavefunctions of a translationally invariant Hamiltonian must be 1D IRs of the translation group in \cref{eq:constraint}.
Here, we exclusively consider non-Abelian PBC clusters as we are interested in higher-dimensional IRs arising from the noncommutativity of the generators, which are necessary to approximate the thermodynamic limit~\cite{Lux:2022, Lux2023, Mosseri2023, Lenggenhager2023}. 
We numerically diagonalize the hyperbolized Hamiltonian (\ref{eq:Hr}) on 13~non-Abelian $\lee$ PBC clusters, with the number of primitive cells ranging from 300 to 1800.
The resulting DOS averaged over these clusters for $m\,{=}\,3$ is depicted by the black curve in \cref{fig3}(a). Our first main finding is the discovery of states within the gap of the Abelian Bloch states. 

\begin{figure}[ht]
    \includegraphics[width=\columnwidth]{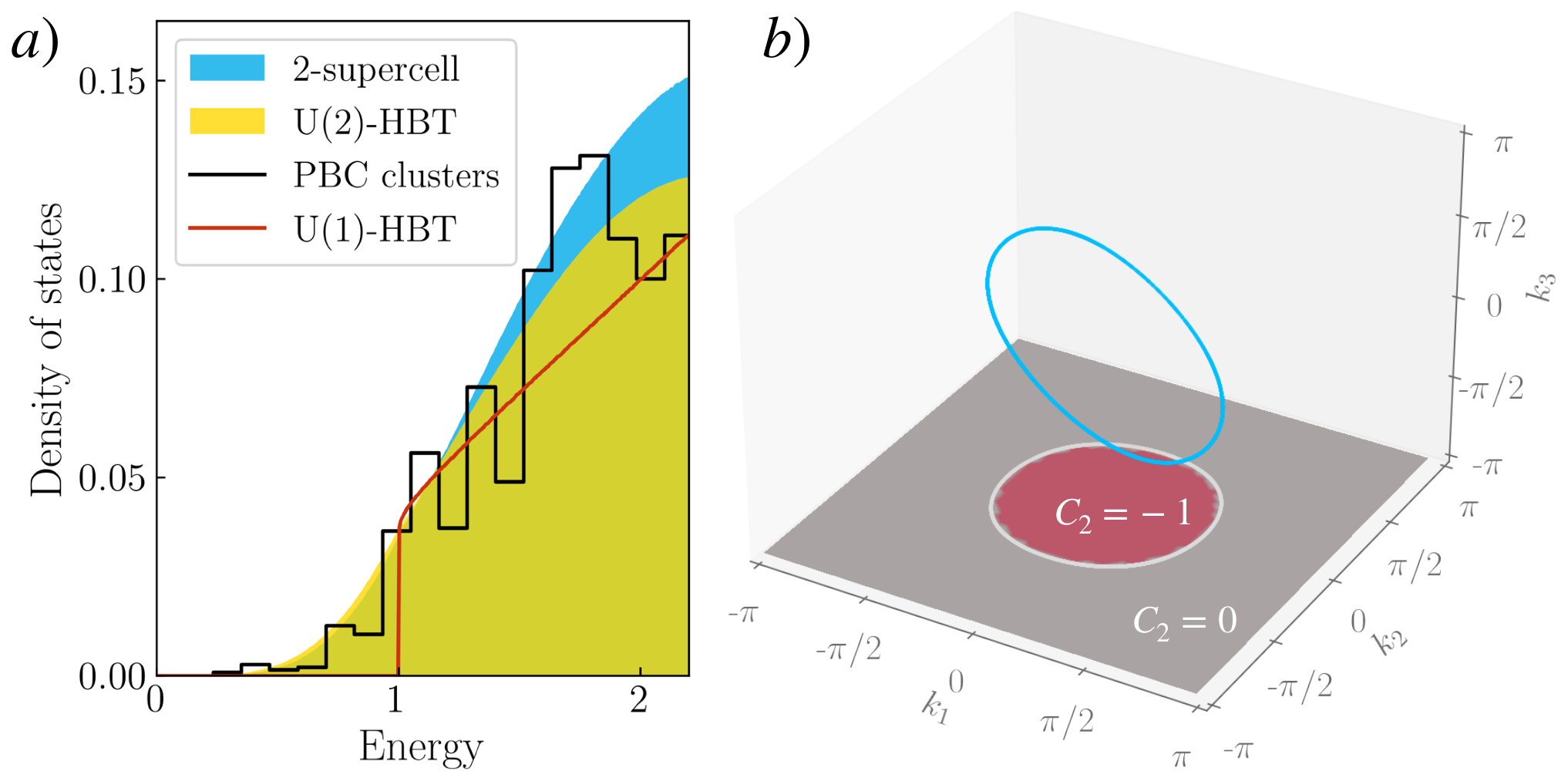}
    \caption{
    Hyperbolic non-Abelian semimetal with Dirac mass $m\,{=}\,3$. 
    (a)~Energy spectra computed with various numerical approaches (cf.~\cref{fig2}).
    In contrast to the hard gap predicted by $\textrm{U}(1)$-HBT, the other approaches indicate a semimetallic DOS symmetric about zero energy.
    (b)~Low-energy excitations are manifested by a nodal ring in the $2\textrm{-supercell}$ description. 
    Numerically obtained nodes in the 6D BZ$^{(2,1)}$ have been projected to the first three momentum components. Projection onto the first two coordinates appears as the white curve. 
    Computing the second Chern number $C_2$ as a function of $(k_1{,}k_2)$ shows that 4D subspaces passing through the ring are topological.
    }
    \label{fig3}
\end{figure}


\emph{$\textrm{U}(d)$-HBT.---}To demonstrate that these in-gap states can be attributed to non-Abelian Bloch physics, we first describe the corresponding BZs~\cite{shankar2023}. 
Here, the Bloch states transform in representations $D_{\lambda}(\gamma_j) \,{=}\, U_j\,\mathrm{e}^{\mathrm{i} k_j}$, where $\bU \,{=}\, \{U_j\}_{j\,{=}\,1}^{4} \,{\in}\, \textrm{SU}(d)$ are special unitary $(d\times d)$ matrices that satisfy \cref{eq:constraint}. 
The specification of $\lambda$ requires $2d^2\,{+}\,2$ parameters, of which four may be understood as hyperbolic momenta $\bk \,{=}\, \{k_j\}_{j\,{=}\,1}^4$ and the remaining $2d^2\,{-}\,2$ characterize $\bU$.
Focusing on a single primitive cell ($n\,{=}\,1$), the states in $\mathrm{BZ}^{(1,d)}$ can be described by the non-Abelian Bloch Hamiltonian 
\begin{equation}
    \!H_{\lambda}^{(1,d)} \!=\! \sum_j\! \left[\!\left( \frac{\GM{5} - \mathrm{i} \GM{j}}{2}\! \right) \!\otimes\! D_{\lambda}(\gamma_{j}) + {\rm h.c.} \!\right] \!+\! m \left( \GM{5} \otimes \mathbbm{1}_d  \right)\!.\! 
    \label{eq:bloch_2D}
\end{equation}
Although an explicit parameterization of $\mathrm{BZ}^{(1,d)}$ is currently lacking, one may randomly sample the subspace of $\textrm{SU}(d)$ matrices that obey \cref{eq:constraint} using a procedure described in~\cite{SM}. 
For different choices of $\bU$ and $\bk$, a numerical diagonalization of $H_{\lambda}^{(1,d)}$ allows for a reconstruction of the $\textrm{U}(d)$ spectra. 
We here focus on the simplest nontrivial case $d\,{=}\,2$, since spectra of small systems that could be realized in experiments are dominated by low-dimensional representations~\cite{Maciejko2022}.
The DOS thus obtained is plotted in yellow in \cref{fig3}(a). 
Consistent with our exact diagonalization results on PBC clusters, non-Abelian states fill the band gap of the Abelian spectrum. 


\emph{Supercell method.---}An alternative way to access non-Abelian Bloch states is the supercell method~\cite{Lenggenhager2023}. 
Instead of tessellating the infinite lattice with copies of the primitive cell [red octagon in \cref{fig1}(a)], the lattice is subdivided into collections of $n\,{>}\,1$ primitive cells, called $n$-\emph{supercells} [blue polygon in \cref{fig1}(a) depicts the $2$-supercell].
Copies of the primitive cell covering the lattice are related by elements of the full translation group of the lattice $\HTGuc\,{=}\,\HTG$, while copies of $n$-supercell are related only by elements of an $\textrm{index-}n$ normal subgroup $\HTGnc\,{\triangleleft}\,\HTGuc$.

In either case, translation symmetry allows a single-cell description of the infinite system by going to hyperbolic reciprocal space as described previously.
In Euclidean lattices, enlarging the unit cell introduces band folding and a reduction in BZ size while the total number of states remains constant.
In contrast, in hyperbolic lattices, the decomposition of reciprocal space into the generalized BZs is changed:
The BZs for the primitive cell $\mathrm{BZ}^{(1,d)}$ are different from the ones for $n$-supercells, $\mathrm{BZ}^{(n,d)}$.
Nevertheless, their union (over IR dimensions $d$) is expected to remain the same, such that both are descriptions of the same infinite system.
Specifically, as a consequence of the negative curvature, the genus $\genus$ of the compactified manifold (and thereby the dimension of $\mathrm{BZ}^{(n,1)}$ which is a $2\genus{}$-dimensional torus), grows linearly with the supercell size. 
In particular, $\mathrm{BZ}^{(1,1)}$ is 4D, while $\mathrm{BZ}^{(2,1)}$ is 6D, with \textrm{BZ}$^{(1,1)}$ passing through the center of \textrm{BZ}$^{(2,1)}$~\cite{SM}.

The advantage of working with supercells comes from the fact that, in contrast to the non-Abelian BZs, the $\textrm{U}(1)$-BZs are well-understood and an explicit parametrization in terms of Bloch phase factors is known.
Due to $\HTGnc$ being an index-$n$ subgroup of $\HTGuc$, 
$\textrm{U}(1)$ representations of $\HTGnc$ induce $\textrm{U}(n)$ representations of $\HTGuc$, which are either reducible or irreducible~\cite{Lenggenhager2023}. 
Due to the growing dimension of \textrm{BZ}$^{(n,1)}$ with supercell size $n$, most of the $\textrm{U}(n)$ representations of $\HTGuc$ are non-Abelian, i.e., they do not split into a product of $\textrm{U}(1)$ representations.
Therefore, by applying $\textrm{U}(1)$-HBT to larger supercells, more non-Abelian states are captured.

We focus again on the simplest nontrivial case of $2$-supercells, thus gaining access to certain non-Abelian $\textrm{U}(2)$ Bloch states.
We utilize the \textsc{HyperCells} package for GAP~\cite{HyperCells2023,Conder2007,GAP4} and the \textsc{HyperBloch} package for Mathematica~\cite{HyperBloch}, to sample states in $\mathrm{BZ}^{(2,1)}$.
The result, plotted in blue in \cref{fig3}(a), confirms our earlier findings: a semimetallic DOS within the gap of the Abelian Bloch states.


\emph{Nodal ring.---}In unison, the three methods paint a clear picture: While the $\textrm{U}(1)$ spectrum is gapped, non-Abelian Bloch states arise inside the gap. 
Since the three DOS curves in \cref{fig3}(b) match each other near zero energy, we use the $2$-supercell approach to derive an approximate low-energy theory by studying the Bloch Hamiltonian $H^{(2,1)}_{\bk}$ analytically for $|m|\,{\gtrsim}\, 3$. 
We find that its zero-energy eigenstates form a \emph{nodal ring} inside $\mathrm{BZ}^{(2,1)}$ whose projection onto a 3D subspace is shown in \cref{fig3}(b). 
This nodal ring is protected by space-time-inversion symmetry~\cite{SM}. 
Through a basis change of the Hamiltonian, linear transformation of the momentum space from $\bk$ to $\bq$, and a Taylor expansion around the origin~\cite{SM}, we find that the spectral gap vanishes when 
\begin{equation}
    q_1^2+q_2^2 = 6(4-|m|)\quad\textrm{and}\quad q_{3}=q_{4}=q_{5}=q_{6}=0.
    \label{eq:circle}
\end{equation}
As a circle with mass-dependent radius, the nodal manifold shrinks to a point at the topological phase transitions $|m|\,{=}\,4$.

To understand the nodal topology, it is useful to draw parallels with Weyl semimetals in 3D Euclidean space. 
Weyl nodes can be understood as sources and sinks of Berry curvature in momentum space. 
As a consequence of the dipolar configuration, a 2D plane situated between two nodes of opposite charge experiences a net Berry flux and features a finite Chern number, cf.~\cref{fig1}(c).
In a similar spirit, one may expect the nodal ring in 6D hyperbolic reciprocal space to encode nontrivial topology. 
Since our starting point was the 4D QHI, the second Chern number $C_2$~\cite{MocholGrzelak2018} is the relevant topological quantity. 
Indeed, integration over the four momenta $k_{3,\ldots,6}$ shows that the nodal ring projection to the $(k_1{,}k_2)\textrm{-plane}$ encloses a region where $C_2\,{=}\,-1$, as shown in \cref{fig3}(b).

While the DOS of 3D Weyl semimetals near $E=0$ scales as $\rho(E)\,{\propto}\,E^2$~\cite{Zyuzin:2012}, we expect that for the hyperbolic non-Abelian semimetal it scales as $\rho(E)\,{\propto}\, E^4$. 
To see this, first observe that 4D QHI Hamiltonian $\mathcal{H}$ carries a spinful space-time-inversion symmetry, $\mathcal{PT}\,{=}\,\GM{2}\GM{4}\mathcal{K}$ with $(\mathcal{PT})^2\,{=}\,{{-}\mathbbm{1}}$ and complex conjugation $\mathcal{K}$.
We find~\cite{SM} that the $\mathcal{PT}$ symmetry is inherited by the $2$-supercell Bloch Hamiltonian $H_{\bk}^{(2,1)}$. 
It follows from the von Neumann-Wigner theorem~\cite{vonNeumann:1929,Bzdusek:2017} that $H_{\bk}^{(2,1)}$ generically exhibits band nodes of codimension $\mathfrak{d}\,{=}\,5$, consistent with the observed nodal ring of dimension $\dim[\textrm{BZ}^{(2,1)}]\,{-}\,\mathfrak{d}\,{=}\,(6\,{-}\,5)\,{=}\,1$. 
Assuming the most generic scenario where the bands disperse linearly in the five directions perpendicular ($\perp$) to the nodal ring, quartic DOS scaling $\rho(E) \,{=}\, \int\mathrm{d}^5k_\perp \delta(E \,{-}\, v |k_\perp|) \,{\propto}\, E^4$ follows. 
Through numerical analysis~\cite{SM}, we find that within the $2$-supercell approximation for $\left|{m}\right|\,{\in}\,(0,4)$ the DOS scales as $\rho(E)\,{\propto}\,E^\alpha$ with $\alpha\,{\in}\,(3.3,4.1)$. 
Specifically for $m\,{\in}\,(3,4)$, where a single nodal ring is present, we confirm the predicted value $\alpha \,{=}\, 4$ within the numerical precision.


\begin{figure}[t]
    \includegraphics[width=\columnwidth]{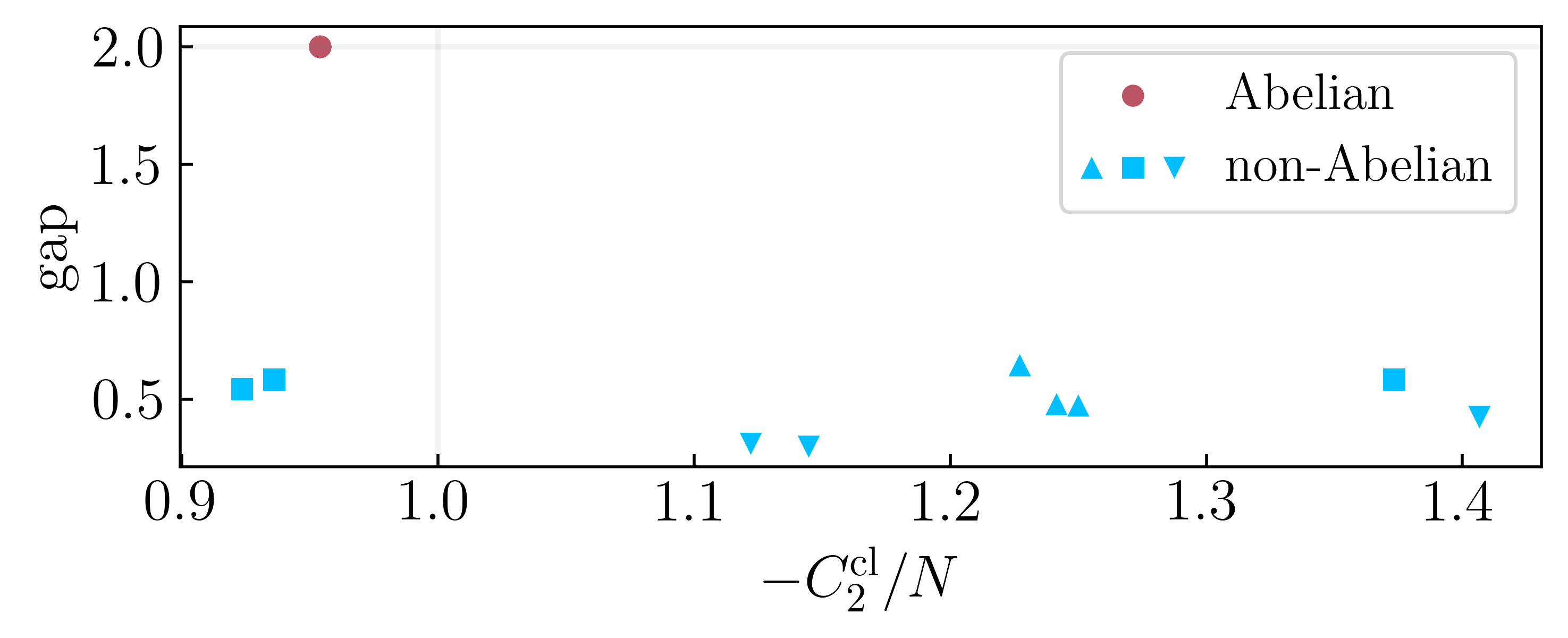}
    \caption{Flux insertion in PBC clusters. We study Abelian and non-Abelian clusters with $N=16, 18, 20$ ($\blacktriangle,\blacksquare,\blacktriangledown$)  primitive cells and use a phase space grid of $20^4$ to compute the cluster-specific second Chern number $C_2^{\textrm{cl}}$ with $m\,{=}\,3$. 
    While the Abelian clusters agree with U(1)-HBT, non-Abelian clusters show smaller gaps owing to the semimetallic states and $C_2^{\textrm{cl}}/N$ has a larger variance.}
    \label{fig4} 
\end{figure}

\emph{Experimental signatures.---}A priori, any physical realization of a hyperbolic lattice has states that transform under higher-dimensional IRs. 
To experimentally probe non-Abelian states, however, one needs to isolate them from the $\textrm{U}(1)$ states. 
The hyperbolic non-Abelian semimetal provides an energy window to \emph{exclusively} access these states. 
In coplanar waveguide resonators~\cite{Kollr2019}, for instance, a gap would appear as a dip in the transmission spectra, whereas a finite DOS at low energies would support transmission. 

To understand the physical implications of nodal topology, recall the case of Weyl semimetals. 
There, each 2D slice between the Weyl nodes, carrying nonzero value of $C_1$ [cf.~Fig.~\ref{fig1}(c)], contributes to anomalous Hall conductance $\sigma_{xy}\,{=}\,(e^2/2\pi h)\,\Delta k$ where $\Delta k$ is the momentum node separation~\cite{Zyuzin:2012,Hosur:2013,Burkov:2014}.
Similarly, the response of the hyperbolic non-Abelian semimetal to external fields should be contingent on the nodal manifold geometry.
We therefore study the second Chern number $C_2^{\textrm{cl}}$ in small PBC clusters (consisting of $N$ primitive cells), which is generated by filled states under flux insertion $(\GM{5}\,{-}\,i\GM{j})\mapsto (\GM{5}\,{-}\,i\GM{j})e^{i\phi_j}$ with $\boldsymbol{\phi}\in[0,2\pi)^{\times 4}$ in Eq.~(\ref{eq:Hr}) \cite{SM}.
Practically, such phase manipulation can be implemented with tunable complex-phase elements \cite{Chen2023a} in topoelectrical circuits \cite{Lenggenhager2022,Zhang:2022,Zhang_2023}, where the eigenstates necessary for computing $C_2^{\textrm{cl}}$ can be accessed with simple oscilloscope measurements.
For Abelian states, the flux insertion induces a momentum shift $\boldsymbol{k}\,{\mapsto}\, \boldsymbol{k} \,{+}\, \boldsymbol{\phi}$; therefore, all eigenstates contribute equally to the integration, implying $C_2^{\textrm{cl}}/N \,{=}\, C_2$.
In contrast, the trajectory of non-Abelian states under flux insertion may fall outside the nodal-ring manifold [cf.~Fig.~\ref{fig3}(b)], or pass through an additional nodal manifold not represented in the $\textrm{6D}\;\textrm{BZ}^{(2,1)}$.
In both cases, one anticipates deviation from the linear scaling $C_2^{\textrm{cl}}\,{=}\,N C_2$. 
\Cref{fig4} (supplemented with further data in \cite{SM}) shows the result of such integration for randomly selected Abelian (red) and non-Abelian (blue) clusters, revealing agreement with our theoretical arguments.
We leave the study of the convergence of $C_2^{\textrm{cl}}/N$ in the thermodynamic limit~\cite{Lux:2022,Lux2023,Mosseri2023,Lenggenhager2023} to future works.


\emph{Conclusion and outlooks.---}Semimetals are fertile platforms that emulate particle physics models on one hand and actualize solid-state notions of topology on the other. 
Conventionally, they have been studied in Euclidean space. 
Extending the notion of topological band nodes to negatively curved space, we investigated the hyperbolized 4D QHI Hamiltonian.
The model exhibits striking features that have no counterpart in Euclidean crystals, such as band nodes stabilized by a second Chern number and low-energy excitations transforming exclusively in non-Abelian IRs of the hyperbolic translation group. 

Our findings motivate broader and more systematic studies of band topology in hyperbolic lattices.
In particular, the semimetallic nature of the hyperbolized QHI model inspires the search for \emph{fully gapped} hyperbolic topological insulators in which the region with a finite second Chern number [red in Fig.~\ref{fig3}(b)] spans the entire higher-dimensional BZ~\cite{Chen2023b}.
Furthermore, while the $\textrm{U}(1)$ and $\textrm{U}(2)$ representations considered here dominate the spectra of small systems that can be realized in experiments~\cite{Maciejko2022}, higher-dimensional $\textrm{U}(d)$ representations are necessary to describe the system in the thermodynamic limit~\cite{Lux:2022,Lux2023,Mosseri2023,Lenggenhager2023}. 
Our preliminary data obtained for larger $n$-supercells~\cite{SM} suggest not only a change in the DOS scaling exponent $\alpha$ but also a gradual evolution of the non-Abelian semimetal below (above) a critical value $|{m_\textrm{c}}|\simeq 2.5$ to a non-Abelian metal (insulator) phase.

On the experimental front, to probe the transport associated with $C_2$~\cite{zhang2001}, one needs to reconcile a dimensionality mismatch: the hyperbolic plane is two-dimensional, whereas four orthogonal directions enter (via the four-component Levi-Civita symbol) the nonlinear response to applied fields~\cite{Qi2008,Price:2015}. 
Experimental realizations would also benefit from a generalization of the hyperbolic non-Abelian semimetal to $\{p{,}q\}$ lattices with a smaller curvature per site than the $\lee$ lattice assumed here.
Furthermore, to enhance the versatility of experimental realizations, it is desirable to seek implementations of hyperbolic lattices in other analog simulators including silicon photonics~\cite{Huang:2024}, as well as in quantum platforms such as optical tweezer arrays~\cite{Kaufman2021,Spar2022} and trapped ions~\cite{Monroe:2021}.
In a broader context, exciting open questions arise for hyperbolic topological models in relation to the bulk-boundary correspondence~\cite{Yu2020,urwyler2022hyperbolic,Liu2022} and the holographic principle~\cite{Witten:1998,Zaanen:2015,Asaduzzaman:2020,Chen:2023:Ads/CFT}.


\let\oldaddcontentsline\addcontentsline     
\renewcommand{\addcontentsline}[3]{}        

\emph{Acknowledgements.---} The code and the generated data used to arrive at the conclusions presented in this work are publicly available in Ref.~\cite{Tummuru:2023:SDC}. 
We would like to thank I.~Boettcher and R.~Thomale for valuable discussions and the referees for their insightful questions. 
T.~T.~and T.~N.~acknowledge funding from the European Research Council (ERC) under the European Union’s Horizon 2020 research and innovation programm (ERC-StG-Neupert-757867-PARATOP). 
A.~C.~was~supported by the University of Alberta startup fund UOFAB Startup Boettcher and the Avadh Bhatia Fellowship. 
J.~M.~was supported by NSERC Discovery Grants \#RGPIN-2020-06999 and \#RGPAS-2020-00064; the Canada Research Chair (CRC) Program; the Government of Alberta's Major Innovation Fund (MIF); the Tri-Agency New Frontiers in Research Fund (NFRF, Exploration Stream); and the Pacific Institute for the Mathematical Sciences (PIMS) Collaborative Research Group program.
P.~M.~L.~and T.~B.~were supported by the Ambizione grant No.~185806 by the Swiss National Science Foundation (SNSF).
T.~B.~was supported by the Starting Grant No.~211310 by SNSF.
P.~M.~L. acknowledges funding by the European Union (ERC, QuSimCtrl, 101113633). Views and opinions expressed are however those of the authors only and do not necessarily reflect those of the European Union or the European Research Council Executive Agency. Neither the European Union nor the granting authority can be held responsible for them.


\pdfbookmark[1]{References}{references}
\bibliography{ref}

\begin{thebibliography}{73}%
\makeatletter
\providecommand \@ifxundefined [1]{%
 \@ifx{#1\undefined}
}%
\providecommand \@ifnum [1]{%
 \ifnum #1\expandafter \@firstoftwo
 \else \expandafter \@secondoftwo
 \fi
}%
\providecommand \@ifx [1]{%
 \ifx #1\expandafter \@firstoftwo
 \else \expandafter \@secondoftwo
 \fi
}%
\providecommand \natexlab [1]{#1}%
\providecommand \enquote  [1]{``#1''}%
\providecommand \bibnamefont  [1]{#1}%
\providecommand \bibfnamefont [1]{#1}%
\providecommand \citenamefont [1]{#1}%
\providecommand \href@noop [0]{\@secondoftwo}%
\providecommand \href [0]{\begingroup \@sanitize@url \@href}%
\providecommand \@href[1]{\@@startlink{#1}\@@href}%
\providecommand \@@href[1]{\endgroup#1\@@endlink}%
\providecommand \@sanitize@url [0]{\catcode `\\12\catcode `\$12\catcode
  `\&12\catcode `\#12\catcode `\^12\catcode `\_12\catcode `\%12\relax}%
\providecommand \@@startlink[1]{}%
\providecommand \@@endlink[0]{}%
\providecommand \url  [0]{\begingroup\@sanitize@url \@url }%
\providecommand \@url [1]{\endgroup\@href {#1}{\urlprefix }}%
\providecommand \urlprefix  [0]{URL }%
\providecommand \Eprint [0]{\href }%
\providecommand \doibase [0]{http://dx.doi.org/}%
\providecommand \selectlanguage [0]{\@gobble}%
\providecommand \bibinfo  [0]{\@secondoftwo}%
\providecommand \bibfield  [0]{\@secondoftwo}%
\providecommand \translation [1]{[#1]}%
\providecommand \BibitemOpen [0]{}%
\providecommand \bibitemStop [0]{}%
\providecommand \bibitemNoStop [0]{.\EOS\space}%
\providecommand \EOS [0]{\spacefactor3000\relax}%
\providecommand \BibitemShut  [1]{\csname bibitem#1\endcsname}%
\let\auto@bib@innerbib\@empty
\bibitem [{\citenamefont {Wan}\ \emph {et~al.}(2011)\citenamefont {Wan},
  \citenamefont {Turner}, \citenamefont {Vishwanath},\ and\ \citenamefont
  {Savrasov}}]{Wan:2011}%
  \BibitemOpen
  \bibfield  {author} {\bibinfo {author} {\bibfnamefont {Xiangang}\
  \bibnamefont {Wan}}, \bibinfo {author} {\bibfnamefont {Ari~M.}\ \bibnamefont
  {Turner}}, \bibinfo {author} {\bibfnamefont {Ashvin}\ \bibnamefont
  {Vishwanath}}, \ and\ \bibinfo {author} {\bibfnamefont {Sergey~Y.}\
  \bibnamefont {Savrasov}},\ }\bibfield  {title} {\enquote {\bibinfo {title}
  {Topological semimetal and {F}ermi-arc surface states in the electronic
  structure of pyrochlore iridates},}\ }\href {\doibase
  10.1103/PhysRevB.83.205101} {\bibfield  {journal} {\bibinfo  {journal} {Phys.
  Rev. B}\ }\textbf {\bibinfo {volume} {83}},\ \bibinfo {pages} {205101}
  (\bibinfo {year} {2011})}\BibitemShut {NoStop}%
\bibitem [{\citenamefont {Zyuzin}\ and\ \citenamefont
  {Burkov}(2012)}]{Zyuzin:2012}%
  \BibitemOpen
  \bibfield  {author} {\bibinfo {author} {\bibfnamefont {A.~A.}\ \bibnamefont
  {Zyuzin}}\ and\ \bibinfo {author} {\bibfnamefont {A.~A.}\ \bibnamefont
  {Burkov}},\ }\bibfield  {title} {\enquote {\bibinfo {title} {Topological
  response in {W}eyl semimetals and the chiral anomaly},}\ }\href {\doibase
  10.1103/PhysRevB.86.115133} {\bibfield  {journal} {\bibinfo  {journal} {Phys.
  Rev. B}\ }\textbf {\bibinfo {volume} {86}},\ \bibinfo {pages} {115133}
  (\bibinfo {year} {2012})}\BibitemShut {NoStop}%
\bibitem [{\citenamefont {Hosur}\ and\ \citenamefont {Qi}(2013)}]{Hosur:2013}%
  \BibitemOpen
  \bibfield  {author} {\bibinfo {author} {\bibfnamefont {Pavan}\ \bibnamefont
  {Hosur}}\ and\ \bibinfo {author} {\bibfnamefont {Xiaoliang}\ \bibnamefont
  {Qi}},\ }\bibfield  {title} {\enquote {\bibinfo {title} {Recent developments
  in transport phenomena in {W}eyl semimetals},}\ }\href {\doibase
  https://doi.org/10.1016/j.crhy.2013.10.010} {\bibfield  {journal} {\bibinfo
  {journal} {C. R. Phys.}\ }\textbf {\bibinfo {volume} {14}},\ \bibinfo {pages}
  {857--870} (\bibinfo {year} {2013})}\BibitemShut {NoStop}%
\bibitem [{\citenamefont {Burkov}(2014)}]{Burkov:2014}%
  \BibitemOpen
  \bibfield  {author} {\bibinfo {author} {\bibfnamefont {A.~A.}\ \bibnamefont
  {Burkov}},\ }\bibfield  {title} {\enquote {\bibinfo {title} {Anomalous {H}all
  {E}ffect in {W}eyl {M}etals},}\ }\href {\doibase
  10.1103/PhysRevLett.113.187202} {\bibfield  {journal} {\bibinfo  {journal}
  {Phys. Rev. Lett.}\ }\textbf {\bibinfo {volume} {113}},\ \bibinfo {pages}
  {187202} (\bibinfo {year} {2014})}\BibitemShut {NoStop}%
\bibitem [{\citenamefont {Vazifeh}\ and\ \citenamefont
  {Franz}(2013)}]{Vazifeh:2013}%
  \BibitemOpen
  \bibfield  {author} {\bibinfo {author} {\bibfnamefont {M.~M.}\ \bibnamefont
  {Vazifeh}}\ and\ \bibinfo {author} {\bibfnamefont {M.}~\bibnamefont
  {Franz}},\ }\bibfield  {title} {\enquote {\bibinfo {title} {Electromagnetic
  {R}esponse of {W}eyl {S}emimetals},}\ }\href {\doibase
  10.1103/PhysRevLett.111.027201} {\bibfield  {journal} {\bibinfo  {journal}
  {Phys. Rev. Lett.}\ }\textbf {\bibinfo {volume} {111}},\ \bibinfo {pages}
  {027201} (\bibinfo {year} {2013})}\BibitemShut {NoStop}%
\bibitem [{\citenamefont {Vafek}\ and\ \citenamefont
  {Vishwanath}(2014)}]{vafek2014}%
  \BibitemOpen
  \bibfield  {author} {\bibinfo {author} {\bibfnamefont {O.}~\bibnamefont
  {Vafek}}\ and\ \bibinfo {author} {\bibfnamefont {A.}~\bibnamefont
  {Vishwanath}},\ }\bibfield  {title} {\enquote {\bibinfo {title} {Dirac
  {Fermions} in {Solids}: From {High}-{$T_c$} {Cuprates} and {Graphene} to
  {Topological} {Insulators} and {Weyl} {Semimetals}},}\ }\href {\doibase
  10.1146/annurev-conmatphys-031113-133841} {\bibfield  {journal} {\bibinfo
  {journal} {Annu. Rev. Condens. Matter Phys.}\ }\textbf {\bibinfo {volume}
  {5}},\ \bibinfo {pages} {83--112} (\bibinfo {year} {2014})}\BibitemShut
  {NoStop}%
\bibitem [{\citenamefont {Armitage}\ \emph {et~al.}(2018)\citenamefont
  {Armitage}, \citenamefont {Mele},\ and\ \citenamefont
  {Vishwanath}}]{Armitage2018}%
  \BibitemOpen
  \bibfield  {author} {\bibinfo {author} {\bibfnamefont {N.~P.}\ \bibnamefont
  {Armitage}}, \bibinfo {author} {\bibfnamefont {E.~J.}\ \bibnamefont {Mele}},
  \ and\ \bibinfo {author} {\bibfnamefont {Ashvin}\ \bibnamefont
  {Vishwanath}},\ }\bibfield  {title} {\enquote {\bibinfo {title} {Weyl and
  {Dirac} semimetals in three-dimensional solids},}\ }\href {\doibase
  10.1103/RevModPhys.90.015001} {\bibfield  {journal} {\bibinfo  {journal}
  {Rev. Mod. Phys.}\ }\textbf {\bibinfo {volume} {90}},\ \bibinfo {pages}
  {015001} (\bibinfo {year} {2018})}\BibitemShut {NoStop}%
\bibitem [{\citenamefont {Ong}\ and\ \citenamefont {Liang}(2021)}]{ong2021}%
  \BibitemOpen
  \bibfield  {author} {\bibinfo {author} {\bibfnamefont {N.~P.}\ \bibnamefont
  {Ong}}\ and\ \bibinfo {author} {\bibfnamefont {S.}~\bibnamefont {Liang}},\
  }\bibfield  {title} {\enquote {\bibinfo {title} {Experimental signatures of
  the chiral anomaly in {Dirac}–{Weyl} semimetals},}\ }\href {\doibase
  10.1038/s42254-021-00310-9} {\bibfield  {journal} {\bibinfo  {journal} {Nat.
  Rev. Phys.}\ }\textbf {\bibinfo {volume} {3}},\ \bibinfo {pages} {394--404}
  (\bibinfo {year} {2021})}\BibitemShut {NoStop}%
\bibitem [{\citenamefont {Gao}\ \emph {et~al.}(2019)\citenamefont {Gao},
  \citenamefont {Venderbos}, \citenamefont {Kim},\ and\ \citenamefont
  {Rappe}}]{Gao2019}%
  \BibitemOpen
  \bibfield  {author} {\bibinfo {author} {\bibfnamefont {Heng}\ \bibnamefont
  {Gao}}, \bibinfo {author} {\bibfnamefont {J\"{o}rn~W.F.}\ \bibnamefont
  {Venderbos}}, \bibinfo {author} {\bibfnamefont {Youngkuk}\ \bibnamefont
  {Kim}}, \ and\ \bibinfo {author} {\bibfnamefont {Andrew~M.}\ \bibnamefont
  {Rappe}},\ }\bibfield  {title} {\enquote {\bibinfo {title} {Topological
  {S}emimetals from {F}irst {P}rinciples},}\ }\href {\doibase
  10.1146/annurev-matsci-070218-010049} {\bibfield  {journal} {\bibinfo
  {journal} {Annu. Rev. Mater. Sci.}\ }\textbf {\bibinfo {volume} {49}},\
  \bibinfo {pages} {153--183} (\bibinfo {year} {2019})}\BibitemShut {NoStop}%
\bibitem [{\citenamefont {{Coxeter, H. S. M.}}(1957)}]{Coxeter:1957}%
  \BibitemOpen
  \bibfield  {author} {\bibinfo {author} {\bibnamefont {{Coxeter, H. S. M.}}},\
  }\bibfield  {title} {\enquote {\bibinfo {title} {{Crystal symmetry and its
  generalizations}},}\ }\href@noop {} {\bibfield  {journal} {\bibinfo
  {journal} {{Proc. Trans. R. Soc. Can.}}\ }\textbf {\bibinfo {volume} {51}},\
  \bibinfo {pages} {1--13} (\bibinfo {year} {1957})}\BibitemShut {NoStop}%
\bibitem [{\citenamefont {Koll{\'{a}}r}\ \emph {et~al.}(2019)\citenamefont
  {Koll{\'{a}}r}, \citenamefont {Fitzpatrick},\ and\ \citenamefont
  {Houck}}]{Kollr2019}%
  \BibitemOpen
  \bibfield  {author} {\bibinfo {author} {\bibfnamefont {Alicia~J.}\
  \bibnamefont {Koll{\'{a}}r}}, \bibinfo {author} {\bibfnamefont {Mattias}\
  \bibnamefont {Fitzpatrick}}, \ and\ \bibinfo {author} {\bibfnamefont
  {Andrew~A.}\ \bibnamefont {Houck}},\ }\bibfield  {title} {\enquote {\bibinfo
  {title} {Hyperbolic lattices in circuit quantum electrodynamics},}\ }\href
  {\doibase 10.1038/s41586-019-1348-3} {\bibfield  {journal} {\bibinfo
  {journal} {Nature}\ }\textbf {\bibinfo {volume} {571}},\ \bibinfo {pages}
  {45--50} (\bibinfo {year} {2019})}\BibitemShut {NoStop}%
\bibitem [{\citenamefont {Chen}\ \emph
  {et~al.}(2023{\natexlab{a}})\citenamefont {Chen}, \citenamefont {Brand},
  \citenamefont {Helbig}, \citenamefont {Hofmann}, \citenamefont {Imhof},
  \citenamefont {Fritzsche}, \citenamefont {Kie{\ss}ling}, \citenamefont
  {Stegmaier}, \citenamefont {Upreti}, \citenamefont {Neupert}, \citenamefont
  {Bzdu\v{s}ek}, \citenamefont {Greiter}, \citenamefont {Thomale},\ and\
  \citenamefont {Boettcher}}]{Chen2023a}%
  \BibitemOpen
  \bibfield  {author} {\bibinfo {author} {\bibfnamefont {A.}~\bibnamefont
  {Chen}}, \bibinfo {author} {\bibfnamefont {H.}~\bibnamefont {Brand}},
  \bibinfo {author} {\bibfnamefont {T.}~\bibnamefont {Helbig}}, \bibinfo
  {author} {\bibfnamefont {T.}~\bibnamefont {Hofmann}}, \bibinfo {author}
  {\bibfnamefont {S.}~\bibnamefont {Imhof}}, \bibinfo {author} {\bibfnamefont
  {A.}~\bibnamefont {Fritzsche}}, \bibinfo {author} {\bibfnamefont
  {T.}~\bibnamefont {Kie{\ss}ling}}, \bibinfo {author} {\bibfnamefont
  {A.}~\bibnamefont {Stegmaier}}, \bibinfo {author} {\bibfnamefont {L.~K.}\
  \bibnamefont {Upreti}}, \bibinfo {author} {\bibfnamefont {T.}~\bibnamefont
  {Neupert}}, \bibinfo {author} {\bibfnamefont {T.}~\bibnamefont
  {Bzdu\v{s}ek}}, \bibinfo {author} {\bibfnamefont {M.}~\bibnamefont
  {Greiter}}, \bibinfo {author} {\bibfnamefont {R.}~\bibnamefont {Thomale}}, \
  and\ \bibinfo {author} {\bibfnamefont {I.}~\bibnamefont {Boettcher}},\
  }\bibfield  {title} {\enquote {\bibinfo {title} {{Hyperbolic matter in
  electrical circuits with tunable complex phases}},}\ }\href {\doibase
  https://doi.org/10.1038/s41467-023-36359-6} {\bibfield  {journal} {\bibinfo
  {journal} {Nat. Commun.}\ }\textbf {\bibinfo {volume} {14}},\ \bibinfo
  {pages} {622} (\bibinfo {year} {2023}{\natexlab{a}})}\BibitemShut {NoStop}%
\bibitem [{\citenamefont {Lenggenhager}\ \emph {et~al.}(2022)\citenamefont
  {Lenggenhager}, \citenamefont {Stegmaier}, \citenamefont {Upreti},
  \citenamefont {Hofmann}, \citenamefont {Helbig}, \citenamefont {Vollhardt},
  \citenamefont {Greiter}, \citenamefont {Lee}, \citenamefont {Imhof},
  \citenamefont {Brand}, \citenamefont {Kie{\ss}ling}, \citenamefont
  {Boettcher}, \citenamefont {Neupert}, \citenamefont {Thomale},\ and\
  \citenamefont {Bzdu\v{s}ek}}]{Lenggenhager2022}%
  \BibitemOpen
  \bibfield  {author} {\bibinfo {author} {\bibfnamefont {P.~M.}\ \bibnamefont
  {Lenggenhager}}, \bibinfo {author} {\bibfnamefont {A.}~\bibnamefont
  {Stegmaier}}, \bibinfo {author} {\bibfnamefont {L.~K.}\ \bibnamefont
  {Upreti}}, \bibinfo {author} {\bibfnamefont {T.}~\bibnamefont {Hofmann}},
  \bibinfo {author} {\bibfnamefont {T.}~\bibnamefont {Helbig}}, \bibinfo
  {author} {\bibfnamefont {A.}~\bibnamefont {Vollhardt}}, \bibinfo {author}
  {\bibfnamefont {M.}~\bibnamefont {Greiter}}, \bibinfo {author} {\bibfnamefont
  {C.~H.}\ \bibnamefont {Lee}}, \bibinfo {author} {\bibfnamefont
  {S.}~\bibnamefont {Imhof}}, \bibinfo {author} {\bibfnamefont
  {H.}~\bibnamefont {Brand}}, \bibinfo {author} {\bibfnamefont
  {T.}~\bibnamefont {Kie{\ss}ling}}, \bibinfo {author} {\bibfnamefont
  {I.}~\bibnamefont {Boettcher}}, \bibinfo {author} {\bibfnamefont
  {T.}~\bibnamefont {Neupert}}, \bibinfo {author} {\bibfnamefont
  {R.}~\bibnamefont {Thomale}}, \ and\ \bibinfo {author} {\bibfnamefont
  {T.}~\bibnamefont {Bzdu\v{s}ek}},\ }\bibfield  {title} {\enquote {\bibinfo
  {title} {Simulating hyperbolic space on a circuit board},}\ }\href {\doibase
  10.1038/s41467-022-32042-4} {\bibfield  {journal} {\bibinfo  {journal} {Nat.
  Commun.}\ }\textbf {\bibinfo {volume} {13}},\ \bibinfo {pages} {4373}
  (\bibinfo {year} {2022})}\BibitemShut {NoStop}%
\bibitem [{\citenamefont {Zhang}\ \emph {et~al.}(2022)\citenamefont {Zhang},
  \citenamefont {Yuan}, \citenamefont {Sun}, \citenamefont {Sun},\ and\
  \citenamefont {Zhang}}]{Zhang:2022}%
  \BibitemOpen
  \bibfield  {author} {\bibinfo {author} {\bibfnamefont {Weixuan}\ \bibnamefont
  {Zhang}}, \bibinfo {author} {\bibfnamefont {Hao}\ \bibnamefont {Yuan}},
  \bibinfo {author} {\bibfnamefont {Na}~\bibnamefont {Sun}}, \bibinfo {author}
  {\bibfnamefont {Houjun}\ \bibnamefont {Sun}}, \ and\ \bibinfo {author}
  {\bibfnamefont {Xiangdong}\ \bibnamefont {Zhang}},\ }\bibfield  {title}
  {\enquote {\bibinfo {title} {Observation of novel topological states in
  hyperbolic lattices},}\ }\href {\doibase 10.1038/s41467-022-30631-x}
  {\bibfield  {journal} {\bibinfo  {journal} {Nat. Commun.}\ }\textbf {\bibinfo
  {volume} {13}},\ \bibinfo {pages} {2937} (\bibinfo {year}
  {2022})}\BibitemShut {NoStop}%
\bibitem [{\citenamefont {Zhang}\ \emph {et~al.}(2023)\citenamefont {Zhang},
  \citenamefont {Di}, \citenamefont {Zheng}, \citenamefont {Sun},\ and\
  \citenamefont {Zhang}}]{Zhang_2023}%
  \BibitemOpen
  \bibfield  {author} {\bibinfo {author} {\bibfnamefont {Weixuan}\ \bibnamefont
  {Zhang}}, \bibinfo {author} {\bibfnamefont {Fengxiao}\ \bibnamefont {Di}},
  \bibinfo {author} {\bibfnamefont {Xingen}\ \bibnamefont {Zheng}}, \bibinfo
  {author} {\bibfnamefont {Houjun}\ \bibnamefont {Sun}}, \ and\ \bibinfo
  {author} {\bibfnamefont {Xiangdong}\ \bibnamefont {Zhang}},\ }\bibfield
  {title} {\enquote {\bibinfo {title} {{Hyperbolic band topology with
  non-trivial second Chern numbers}},}\ }\href {\doibase
  10.1038/s41467-023-36767-8} {\bibfield  {journal} {\bibinfo  {journal} {Nat.
  Commun.}\ }\textbf {\bibinfo {volume} {14}},\ \bibinfo {pages} {1083}
  (\bibinfo {year} {2023})}\BibitemShut {NoStop}%
\bibitem [{\citenamefont {Yu}\ \emph {et~al.}(2020)\citenamefont {Yu},
  \citenamefont {Piao},\ and\ \citenamefont {Park}}]{Yu2020}%
  \BibitemOpen
  \bibfield  {author} {\bibinfo {author} {\bibfnamefont {S.}~\bibnamefont
  {Yu}}, \bibinfo {author} {\bibfnamefont {X.}~\bibnamefont {Piao}}, \ and\
  \bibinfo {author} {\bibfnamefont {N.}~\bibnamefont {Park}},\ }\bibfield
  {title} {\enquote {\bibinfo {title} {Topological {H}yperbolic {L}attices},}\
  }\href {\doibase 10.1103/PhysRevLett.125.053901} {\bibfield  {journal}
  {\bibinfo  {journal} {Phys. Rev. Lett.}\ }\textbf {\bibinfo {volume} {125}},\
  \bibinfo {pages} {053901} (\bibinfo {year} {2020})}\BibitemShut {NoStop}%
\bibitem [{\citenamefont {Urwyler}\ \emph {et~al.}(2022)\citenamefont
  {Urwyler}, \citenamefont {Lenggenhager}, \citenamefont {Boettcher},
  \citenamefont {Thomale}, \citenamefont {Neupert},\ and\ \citenamefont
  {Bzdu\v{s}ek}}]{urwyler2022hyperbolic}%
  \BibitemOpen
  \bibfield  {author} {\bibinfo {author} {\bibfnamefont {David~M.}\
  \bibnamefont {Urwyler}}, \bibinfo {author} {\bibfnamefont {Patrick~M.}\
  \bibnamefont {Lenggenhager}}, \bibinfo {author} {\bibfnamefont {Igor}\
  \bibnamefont {Boettcher}}, \bibinfo {author} {\bibfnamefont {Ronny}\
  \bibnamefont {Thomale}}, \bibinfo {author} {\bibfnamefont {Titus}\
  \bibnamefont {Neupert}}, \ and\ \bibinfo {author} {\bibfnamefont
  {Tom\'a\v{s}}\ \bibnamefont {Bzdu\v{s}ek}},\ }\bibfield  {title} {\enquote
  {\bibinfo {title} {Hyperbolic {T}opological {B}and {I}nsulators},}\ }\href
  {\doibase 10.1103/PhysRevLett.129.246402} {\bibfield  {journal} {\bibinfo
  {journal} {Phys. Rev. Lett.}\ }\textbf {\bibinfo {volume} {129}},\ \bibinfo
  {pages} {246402} (\bibinfo {year} {2022})}\BibitemShut {NoStop}%
\bibitem [{\citenamefont {Liu}\ \emph {et~al.}(2022)\citenamefont {Liu},
  \citenamefont {Hua}, \citenamefont {Peng},\ and\ \citenamefont
  {Zhou}}]{Liu2022}%
  \BibitemOpen
  \bibfield  {author} {\bibinfo {author} {\bibfnamefont {Zheng-Rong}\
  \bibnamefont {Liu}}, \bibinfo {author} {\bibfnamefont {Chun-Bo}\ \bibnamefont
  {Hua}}, \bibinfo {author} {\bibfnamefont {Tan}\ \bibnamefont {Peng}}, \ and\
  \bibinfo {author} {\bibfnamefont {Bin}\ \bibnamefont {Zhou}},\ }\bibfield
  {title} {\enquote {\bibinfo {title} {Chern insulator in a hyperbolic
  lattice},}\ }\href {\doibase 10.1103/PhysRevB.105.245301} {\bibfield
  {journal} {\bibinfo  {journal} {Phys. Rev. B}\ }\textbf {\bibinfo {volume}
  {105}},\ \bibinfo {pages} {245301} (\bibinfo {year} {2022})}\BibitemShut
  {NoStop}%
\bibitem [{\citenamefont {Pei}\ \emph {et~al.}(2023)\citenamefont {Pei},
  \citenamefont {Yuan}, \citenamefont {Zhang},\ and\ \citenamefont
  {Zhang}}]{Pei:2023}%
  \BibitemOpen
  \bibfield  {author} {\bibinfo {author} {\bibfnamefont {Qingsong}\
  \bibnamefont {Pei}}, \bibinfo {author} {\bibfnamefont {Hao}\ \bibnamefont
  {Yuan}}, \bibinfo {author} {\bibfnamefont {Weixuan}\ \bibnamefont {Zhang}}, \
  and\ \bibinfo {author} {\bibfnamefont {Xiangdong}\ \bibnamefont {Zhang}},\
  }\bibfield  {title} {\enquote {\bibinfo {title} {Engineering
  boundary-dominated topological states in defective hyperbolic lattices},}\
  }\href {\doibase 10.1103/PhysRevB.107.165145} {\bibfield  {journal} {\bibinfo
   {journal} {Phys. Rev. B}\ }\textbf {\bibinfo {volume} {107}},\ \bibinfo
  {pages} {165145} (\bibinfo {year} {2023})}\BibitemShut {NoStop}%
\bibitem [{\citenamefont {Chen}\ \emph
  {et~al.}(2023{\natexlab{b}})\citenamefont {Chen}, \citenamefont {Guan},
  \citenamefont {Lenggenhager}, \citenamefont {Maciejko}, \citenamefont
  {Boettcher},\ and\ \citenamefont {Bzdu\v{s}ek}}]{Chen2023b}%
  \BibitemOpen
  \bibfield  {author} {\bibinfo {author} {\bibfnamefont {Anffany}\ \bibnamefont
  {Chen}}, \bibinfo {author} {\bibfnamefont {Yifei}\ \bibnamefont {Guan}},
  \bibinfo {author} {\bibfnamefont {Patrick~M.}\ \bibnamefont {Lenggenhager}},
  \bibinfo {author} {\bibfnamefont {Joseph}\ \bibnamefont {Maciejko}}, \bibinfo
  {author} {\bibfnamefont {Igor}\ \bibnamefont {Boettcher}}, \ and\ \bibinfo
  {author} {\bibfnamefont {Tom\'a\v{s}}\ \bibnamefont {Bzdu\v{s}ek}},\
  }\bibfield  {title} {\enquote {\bibinfo {title} {Symmetry and topology of
  hyperbolic {H}aldane models},}\ }\href {\doibase 10.1103/PhysRevB.108.085114}
  {\bibfield  {journal} {\bibinfo  {journal} {Phys. Rev. B}\ }\textbf {\bibinfo
  {volume} {108}},\ \bibinfo {pages} {085114} (\bibinfo {year}
  {2023}{\natexlab{b}})}\BibitemShut {NoStop}%
\bibitem [{\citenamefont {Ikeda}\ \emph {et~al.}(2021)\citenamefont {Ikeda},
  \citenamefont {Aoki},\ and\ \citenamefont {Matsuki}}]{ikeda2021}%
  \BibitemOpen
  \bibfield  {author} {\bibinfo {author} {\bibfnamefont {Kazuki}\ \bibnamefont
  {Ikeda}}, \bibinfo {author} {\bibfnamefont {Shoto}\ \bibnamefont {Aoki}}, \
  and\ \bibinfo {author} {\bibfnamefont {Yoshiyuki}\ \bibnamefont {Matsuki}},\
  }\bibfield  {title} {\enquote {\bibinfo {title} {Hyperbolic band theory under
  magnetic field and {Dirac} cones on a higher genus surface},}\ }\href
  {\doibase 10.1088/1361-648X/ac24c4} {\bibfield  {journal} {\bibinfo
  {journal} {J. Phys.: Condens. Matter}\ }\textbf {\bibinfo {volume} {33}},\
  \bibinfo {pages} {485602} (\bibinfo {year} {2021})}\BibitemShut {NoStop}%
\bibitem [{\citenamefont {Stegmaier}\ \emph {et~al.}(2022)\citenamefont
  {Stegmaier}, \citenamefont {Upreti}, \citenamefont {Thomale},\ and\
  \citenamefont {Boettcher}}]{stegmaier2022}%
  \BibitemOpen
  \bibfield  {author} {\bibinfo {author} {\bibfnamefont {Alexander}\
  \bibnamefont {Stegmaier}}, \bibinfo {author} {\bibfnamefont {Lavi~K.}\
  \bibnamefont {Upreti}}, \bibinfo {author} {\bibfnamefont {Ronny}\
  \bibnamefont {Thomale}}, \ and\ \bibinfo {author} {\bibfnamefont {Igor}\
  \bibnamefont {Boettcher}},\ }\bibfield  {title} {\enquote {\bibinfo {title}
  {Universality of {Hofstadter} {Butterflies} on {Hyperbolic} {Lattices}},}\
  }\href {\doibase 10.1103/PhysRevLett.128.166402} {\bibfield  {journal}
  {\bibinfo  {journal} {Phys. Rev. Lett.}\ }\textbf {\bibinfo {volume} {128}},\
  \bibinfo {pages} {166402} (\bibinfo {year} {2022})}\BibitemShut {NoStop}%
\bibitem [{\citenamefont {Koll\'ar}\ \emph {et~al.}(2020)\citenamefont
  {Koll\'ar}, \citenamefont {Fitzpatrick}, \citenamefont {Sarnak},\ and\
  \citenamefont {Houck}}]{kollar2020}%
  \BibitemOpen
  \bibfield  {author} {\bibinfo {author} {\bibfnamefont {Alicia~J.}\
  \bibnamefont {Koll\'ar}}, \bibinfo {author} {\bibfnamefont {Mattias}\
  \bibnamefont {Fitzpatrick}}, \bibinfo {author} {\bibfnamefont {Peter}\
  \bibnamefont {Sarnak}}, \ and\ \bibinfo {author} {\bibfnamefont {Andrew~A.}\
  \bibnamefont {Houck}},\ }\bibfield  {title} {\enquote {\bibinfo {title}
  {Line-{Graph} {Lattices}: {Euclidean} and {Non}-{Euclidean} {Flat} {Bands},
  and {Implementations} in {Circuit} {Quantum} {Electrodynamics}},}\ }\href
  {\doibase 10.1007/s00220-019-03645-8} {\bibfield  {journal} {\bibinfo
  {journal} {Commun. Math. Phys.}\ }\textbf {\bibinfo {volume} {376}},\
  \bibinfo {pages} {1909--1956} (\bibinfo {year} {2020})}\BibitemShut {NoStop}%
\bibitem [{\citenamefont {Saa}\ \emph {et~al.}(2021)\citenamefont {Saa},
  \citenamefont {Miranda},\ and\ \citenamefont {Rouxinol}}]{saa2021}%
  \BibitemOpen
  \bibfield  {author} {\bibinfo {author} {\bibfnamefont {Alberto}\ \bibnamefont
  {Saa}}, \bibinfo {author} {\bibfnamefont {Eduardo}\ \bibnamefont {Miranda}},
  \ and\ \bibinfo {author} {\bibfnamefont {Francisco}\ \bibnamefont
  {Rouxinol}},\ }\bibfield  {title} {\enquote {\bibinfo {title}
  {Higher-dimensional {Euclidean} and non-{Euclidean} structures in planar
  circuit quantum electrodynamics},}\ }\href {http://arxiv.org/abs/2108.08854}
  {\bibfield  {journal} {\bibinfo  {journal} {arXiv:2108.08854}\ } (\bibinfo
  {year} {2021})}\BibitemShut {NoStop}%
\bibitem [{\citenamefont {Bzdu\v{s}ek}\ and\ \citenamefont
  {Maciejko}(2022)}]{Bzdusek2022}%
  \BibitemOpen
  \bibfield  {author} {\bibinfo {author} {\bibfnamefont {Tom\'a\v{s}}\
  \bibnamefont {Bzdu\v{s}ek}}\ and\ \bibinfo {author} {\bibfnamefont {Joseph}\
  \bibnamefont {Maciejko}},\ }\bibfield  {title} {\enquote {\bibinfo {title}
  {Flat bands and band-touching from real-space topology in hyperbolic
  lattices},}\ }\href {\doibase 10.1103/PhysRevB.106.155146} {\bibfield
  {journal} {\bibinfo  {journal} {Phys. Rev. B}\ }\textbf {\bibinfo {volume}
  {106}},\ \bibinfo {pages} {155146} (\bibinfo {year} {2022})}\BibitemShut
  {NoStop}%
\bibitem [{\citenamefont {Mosseri}\ \emph {et~al.}(2022)\citenamefont
  {Mosseri}, \citenamefont {Vogeler},\ and\ \citenamefont
  {Vidal}}]{mosseri2022}%
  \BibitemOpen
  \bibfield  {author} {\bibinfo {author} {\bibfnamefont {Rémy}\ \bibnamefont
  {Mosseri}}, \bibinfo {author} {\bibfnamefont {Roger}\ \bibnamefont
  {Vogeler}}, \ and\ \bibinfo {author} {\bibfnamefont {Julien}\ \bibnamefont
  {Vidal}},\ }\bibfield  {title} {\enquote {\bibinfo {title} {Aharonov-{Bohm}
  cages, flat bands, and gap labeling in hyperbolic tilings},}\ }\href
  {\doibase 10.1103/PhysRevB.106.155120} {\bibfield  {journal} {\bibinfo
  {journal} {Phys. Rev. B}\ }\textbf {\bibinfo {volume} {106}},\ \bibinfo
  {pages} {155120} (\bibinfo {year} {2022})}\BibitemShut {NoStop}%
\bibitem [{\citenamefont {Liu}\ \emph {et~al.}(2023)\citenamefont {Liu},
  \citenamefont {Hua}, \citenamefont {Peng}, \citenamefont {Chen},\ and\
  \citenamefont {Zhou}}]{liu2023}%
  \BibitemOpen
  \bibfield  {author} {\bibinfo {author} {\bibfnamefont {Zheng-Rong}\
  \bibnamefont {Liu}}, \bibinfo {author} {\bibfnamefont {Chun-Bo}\ \bibnamefont
  {Hua}}, \bibinfo {author} {\bibfnamefont {Tan}\ \bibnamefont {Peng}},
  \bibinfo {author} {\bibfnamefont {Rui}\ \bibnamefont {Chen}}, \ and\ \bibinfo
  {author} {\bibfnamefont {Bin}\ \bibnamefont {Zhou}},\ }\bibfield  {title}
  {\enquote {\bibinfo {title} {Higher-order topological insulators in
  hyperbolic lattices},}\ }\href {\doibase 10.1103/PhysRevB.107.125302}
  {\bibfield  {journal} {\bibinfo  {journal} {Phys. Rev. B}\ }\textbf {\bibinfo
  {volume} {107}},\ \bibinfo {pages} {125302} (\bibinfo {year}
  {2023})}\BibitemShut {NoStop}%
\bibitem [{\citenamefont {Tao}\ and\ \citenamefont {Xu}(2023)}]{tao2023}%
  \BibitemOpen
  \bibfield  {author} {\bibinfo {author} {\bibfnamefont {Yu-Liang}\
  \bibnamefont {Tao}}\ and\ \bibinfo {author} {\bibfnamefont {Yong}\
  \bibnamefont {Xu}},\ }\bibfield  {title} {\enquote {\bibinfo {title}
  {Higher-order topological hyperbolic lattices},}\ }\href {\doibase
  10.1103/PhysRevB.107.184201} {\bibfield  {journal} {\bibinfo  {journal}
  {Phys. Rev. B}\ }\textbf {\bibinfo {volume} {107}},\ \bibinfo {pages}
  {184201} (\bibinfo {year} {2023})}\BibitemShut {NoStop}%
\bibitem [{\citenamefont {Sun}\ \emph {et~al.}(2023)\citenamefont {Sun},
  \citenamefont {Li}, \citenamefont {Feng},\ and\ \citenamefont
  {Guo}}]{sun2023}%
  \BibitemOpen
  \bibfield  {author} {\bibinfo {author} {\bibfnamefont {Junsong}\ \bibnamefont
  {Sun}}, \bibinfo {author} {\bibfnamefont {Chang-An}\ \bibnamefont {Li}},
  \bibinfo {author} {\bibfnamefont {Shiping}\ \bibnamefont {Feng}}, \ and\
  \bibinfo {author} {\bibfnamefont {Huaiming}\ \bibnamefont {Guo}},\ }\bibfield
   {title} {\enquote {\bibinfo {title} {Hybrid higher-order skin-topological
  effect in hyperbolic lattices},}\ }\href {\doibase
  10.1103/PhysRevB.108.075122} {\bibfield  {journal} {\bibinfo  {journal}
  {Phys. Rev. B}\ }\textbf {\bibinfo {volume} {108}},\ \bibinfo {pages}
  {075122} (\bibinfo {year} {2023})}\BibitemShut {NoStop}%
\bibitem [{\citenamefont {Zhu}\ \emph {et~al.}(2021)\citenamefont {Zhu},
  \citenamefont {Guo}, \citenamefont {Breuckmann}, \citenamefont {Guo},\ and\
  \citenamefont {Feng}}]{zhu2021}%
  \BibitemOpen
  \bibfield  {author} {\bibinfo {author} {\bibfnamefont {Xingchuan}\
  \bibnamefont {Zhu}}, \bibinfo {author} {\bibfnamefont {Jiaojiao}\
  \bibnamefont {Guo}}, \bibinfo {author} {\bibfnamefont {Nikolas~P.}\
  \bibnamefont {Breuckmann}}, \bibinfo {author} {\bibfnamefont {Huaiming}\
  \bibnamefont {Guo}}, \ and\ \bibinfo {author} {\bibfnamefont {Shiping}\
  \bibnamefont {Feng}},\ }\bibfield  {title} {\enquote {\bibinfo {title}
  {Quantum phase transitions of interacting bosons on hyperbolic lattices},}\
  }\href {\doibase 10.1088/1361-648X/ac0a1a} {\bibfield  {journal} {\bibinfo
  {journal} {J. Phys.: Condens. Matter}\ }\textbf {\bibinfo {volume} {33}},\
  \bibinfo {pages} {335602} (\bibinfo {year} {2021})}\BibitemShut {NoStop}%
\bibitem [{\citenamefont {Bienias}\ \emph {et~al.}(2022)\citenamefont
  {Bienias}, \citenamefont {Boettcher}, \citenamefont {Belyansky},
  \citenamefont {Kollár},\ and\ \citenamefont {Gorshkov}}]{bienias2022}%
  \BibitemOpen
  \bibfield  {author} {\bibinfo {author} {\bibfnamefont {Przemyslaw}\
  \bibnamefont {Bienias}}, \bibinfo {author} {\bibfnamefont {Igor}\
  \bibnamefont {Boettcher}}, \bibinfo {author} {\bibfnamefont {Ron}\
  \bibnamefont {Belyansky}}, \bibinfo {author} {\bibfnamefont {Alicia~J.}\
  \bibnamefont {Kollár}}, \ and\ \bibinfo {author} {\bibfnamefont {Alexey~V.}\
  \bibnamefont {Gorshkov}},\ }\bibfield  {title} {\enquote {\bibinfo {title}
  {Circuit {Quantum} {Electrodynamics} in {Hyperbolic} {Space}: {From} {Photon}
  {Bound} {States} to {Frustrated} {Spin} {Models}},}\ }\href {\doibase
  10.1103/PhysRevLett.128.013601} {\bibfield  {journal} {\bibinfo  {journal}
  {Phys. Rev. Lett.}\ }\textbf {\bibinfo {volume} {128}},\ \bibinfo {pages}
  {013601} (\bibinfo {year} {2022})}\BibitemShut {NoStop}%
\bibitem [{\citenamefont {Gluscevich}\ \emph {et~al.}(2023)\citenamefont
  {Gluscevich}, \citenamefont {Samanta}, \citenamefont {Manna},\ and\
  \citenamefont {Roy}}]{gluscevich2023}%
  \BibitemOpen
  \bibfield  {author} {\bibinfo {author} {\bibfnamefont {Noble}\ \bibnamefont
  {Gluscevich}}, \bibinfo {author} {\bibfnamefont {Abhisek}\ \bibnamefont
  {Samanta}}, \bibinfo {author} {\bibfnamefont {Sourav}\ \bibnamefont {Manna}},
  \ and\ \bibinfo {author} {\bibfnamefont {Bitan}\ \bibnamefont {Roy}},\
  }\bibfield  {title} {\enquote {\bibinfo {title} {Dynamic mass generation on
  two-dimensional electronic hyperbolic lattices},}\ }\href
  {http://arxiv.org/abs/2302.04864} {\bibfield  {journal} {\bibinfo  {journal}
  {arXiv:2302.04864}\ } (\bibinfo {year} {2023})}\BibitemShut {NoStop}%
\bibitem [{\citenamefont {Gluscevich}\ and\ \citenamefont
  {Roy}(2023)}]{gluscevich2023b}%
  \BibitemOpen
  \bibfield  {author} {\bibinfo {author} {\bibfnamefont {Noble}\ \bibnamefont
  {Gluscevich}}\ and\ \bibinfo {author} {\bibfnamefont {Btan}\ \bibnamefont
  {Roy}},\ }\bibfield  {title} {\enquote {\bibinfo {title} {Magnetic catalysis
  in weakly interacting hyperbolic {Dirac} materials},}\ }\href
  {http://arxiv.org/abs/2305.11174} {\bibfield  {journal} {\bibinfo  {journal}
  {arXiv:2305.11174}\ } (\bibinfo {year} {2023})}\BibitemShut {NoStop}%
\bibitem [{\citenamefont {Yan}(2019)}]{yan2019}%
  \BibitemOpen
  \bibfield  {author} {\bibinfo {author} {\bibfnamefont {Han}\ \bibnamefont
  {Yan}},\ }\bibfield  {title} {\enquote {\bibinfo {title} {Hyperbolic fracton
  model, subsystem symmetry, and holography},}\ }\href {\doibase
  10.1103/PhysRevB.99.155126} {\bibfield  {journal} {\bibinfo  {journal} {Phys.
  Rev. B}\ }\textbf {\bibinfo {volume} {99}},\ \bibinfo {pages} {155126}
  (\bibinfo {year} {2019})}\BibitemShut {NoStop}%
\bibitem [{\citenamefont {Yan}\ \emph {et~al.}(2022)\citenamefont {Yan},
  \citenamefont {Slagle},\ and\ \citenamefont {Nevidomskyy}}]{yan2022}%
  \BibitemOpen
  \bibfield  {author} {\bibinfo {author} {\bibfnamefont {Han}\ \bibnamefont
  {Yan}}, \bibinfo {author} {\bibfnamefont {Kevin}\ \bibnamefont {Slagle}}, \
  and\ \bibinfo {author} {\bibfnamefont {Andriy~H.}\ \bibnamefont
  {Nevidomskyy}},\ }\bibfield  {title} {\enquote {\bibinfo {title} {Y-cube
  model and fractal structure of subdimensional particles on hyperbolic
  lattices},}\ }\href {http://arxiv.org/abs/2211.15829} {\bibfield  {journal}
  {\bibinfo  {journal} {arXiv:2211.15829}\ } (\bibinfo {year}
  {2022})}\BibitemShut {NoStop}%
\bibitem [{\citenamefont {Boettcher}\ \emph {et~al.}(2022)\citenamefont
  {Boettcher}, \citenamefont {Gorshkov}, \citenamefont {Koll\'ar},
  \citenamefont {Maciejko}, \citenamefont {Rayan},\ and\ \citenamefont
  {Thomale}}]{Boettcher:2022}%
  \BibitemOpen
  \bibfield  {author} {\bibinfo {author} {\bibfnamefont {Igor}\ \bibnamefont
  {Boettcher}}, \bibinfo {author} {\bibfnamefont {Alexey~V.}\ \bibnamefont
  {Gorshkov}}, \bibinfo {author} {\bibfnamefont {Alicia~J.}\ \bibnamefont
  {Koll\'ar}}, \bibinfo {author} {\bibfnamefont {Joseph}\ \bibnamefont
  {Maciejko}}, \bibinfo {author} {\bibfnamefont {Steven}\ \bibnamefont
  {Rayan}}, \ and\ \bibinfo {author} {\bibfnamefont {Ronny}\ \bibnamefont
  {Thomale}},\ }\bibfield  {title} {\enquote {\bibinfo {title} {Crystallography
  of hyperbolic lattices},}\ }\href {\doibase 10.1103/PhysRevB.105.125118}
  {\bibfield  {journal} {\bibinfo  {journal} {Phys. Rev. B}\ }\textbf {\bibinfo
  {volume} {105}},\ \bibinfo {pages} {125118} (\bibinfo {year}
  {2022})}\BibitemShut {NoStop}%
\bibitem [{\citenamefont {Maciejko}\ and\ \citenamefont
  {Rayan}(2021)}]{maciejko2021}%
  \BibitemOpen
  \bibfield  {author} {\bibinfo {author} {\bibfnamefont {Joseph}\ \bibnamefont
  {Maciejko}}\ and\ \bibinfo {author} {\bibfnamefont {Steven}\ \bibnamefont
  {Rayan}},\ }\bibfield  {title} {\enquote {\bibinfo {title} {Hyperbolic band
  theory},}\ }\href {\doibase 10.1126/sciadv.abe9170} {\bibfield  {journal}
  {\bibinfo  {journal} {Sci. Adv.}\ }\textbf {\bibinfo {volume} {7}},\ \bibinfo
  {pages} {abe9170} (\bibinfo {year} {2021})}\BibitemShut {NoStop}%
\bibitem [{\citenamefont {Maciejko}\ and\ \citenamefont
  {Rayan}(2022)}]{Maciejko2022}%
  \BibitemOpen
  \bibfield  {author} {\bibinfo {author} {\bibfnamefont {Joseph}\ \bibnamefont
  {Maciejko}}\ and\ \bibinfo {author} {\bibfnamefont {Steven}\ \bibnamefont
  {Rayan}},\ }\bibfield  {title} {\enquote {\bibinfo {title} {Automorphic
  {Bloch} theorems for hyperbolic lattices},}\ }\href
  {https://doi.org/10.1073/pnas.2116869119} {\bibfield  {journal} {\bibinfo
  {journal} {Proc. Natl. Acad. Sci. U.S.A.}\ }\textbf {\bibinfo {volume} {119}}
  (\bibinfo {year} {2022})}\BibitemShut {NoStop}%
\bibitem [{\citenamefont {Cheng}\ \emph {et~al.}(2022)\citenamefont {Cheng},
  \citenamefont {Serafin}, \citenamefont {McInerney}, \citenamefont {Rocklin},
  \citenamefont {Sun},\ and\ \citenamefont {Mao}}]{Cheng2022}%
  \BibitemOpen
  \bibfield  {author} {\bibinfo {author} {\bibfnamefont {Nan}\ \bibnamefont
  {Cheng}}, \bibinfo {author} {\bibfnamefont {Francesco}\ \bibnamefont
  {Serafin}}, \bibinfo {author} {\bibfnamefont {James}\ \bibnamefont
  {McInerney}}, \bibinfo {author} {\bibfnamefont {Zeb}\ \bibnamefont
  {Rocklin}}, \bibinfo {author} {\bibfnamefont {Kai}\ \bibnamefont {Sun}}, \
  and\ \bibinfo {author} {\bibfnamefont {Xiaoming}\ \bibnamefont {Mao}},\
  }\bibfield  {title} {\enquote {\bibinfo {title} {Band {T}heory and {B}oundary
  {M}odes of {H}igh-{D}imensional {R}epresentations of {I}nfinite {H}yperbolic
  {L}attices},}\ }\href {\doibase 10.1103/PhysRevLett.129.088002} {\bibfield
  {journal} {\bibinfo  {journal} {Phys. Rev. Lett.}\ }\textbf {\bibinfo
  {volume} {129}},\ \bibinfo {pages} {088002} (\bibinfo {year}
  {2022})}\BibitemShut {NoStop}%
\bibitem [{\citenamefont {Lux}\ and\ \citenamefont
  {Prodan}(2023{\natexlab{a}})}]{Lux:2022}%
  \BibitemOpen
  \bibfield  {author} {\bibinfo {author} {\bibfnamefont {Fabian~R.}\
  \bibnamefont {Lux}}\ and\ \bibinfo {author} {\bibfnamefont {Emil}\
  \bibnamefont {Prodan}},\ }\bibfield  {title} {\enquote {\bibinfo {title}
  {Spectral and combinatorial aspects of cayley-crystals},}\ }\href {\doibase
  10.1007/s00023-023-01373-3} {\bibfield  {journal} {\bibinfo  {journal}
  {Annales Henri Poincaré}\ } (\bibinfo {year} {2023}{\natexlab{a}}),\
  10.1007/s00023-023-01373-3}\BibitemShut {NoStop}%
\bibitem [{\citenamefont {Lux}\ and\ \citenamefont
  {Prodan}(2023{\natexlab{b}})}]{Lux2023}%
  \BibitemOpen
  \bibfield  {author} {\bibinfo {author} {\bibfnamefont {Fabian~R.}\
  \bibnamefont {Lux}}\ and\ \bibinfo {author} {\bibfnamefont {Emil}\
  \bibnamefont {Prodan}},\ }\bibfield  {title} {\enquote {\bibinfo {title}
  {{Converging Periodic Boundary Conditions and Detection of Topological Gaps
  on Regular Hyperbolic Tessellations}},}\ }\href {\doibase
  10.1103/PhysRevLett.131.176603} {\bibfield  {journal} {\bibinfo  {journal}
  {Phys. Rev. Lett.}\ }\textbf {\bibinfo {volume} {131}},\ \bibinfo {pages}
  {176603} (\bibinfo {year} {2023}{\natexlab{b}})}\BibitemShut {NoStop}%
\bibitem [{\citenamefont {Mosseri}\ and\ \citenamefont
  {Vidal}(2023)}]{Mosseri2023}%
  \BibitemOpen
  \bibfield  {author} {\bibinfo {author} {\bibfnamefont {R\'emy}\ \bibnamefont
  {Mosseri}}\ and\ \bibinfo {author} {\bibfnamefont {Julien}\ \bibnamefont
  {Vidal}},\ }\bibfield  {title} {\enquote {\bibinfo {title} {Density of states
  of tight-binding models in the hyperbolic plane},}\ }\href {\doibase
  10.1103/PhysRevB.108.035154} {\bibfield  {journal} {\bibinfo  {journal}
  {Phys. Rev. B}\ }\textbf {\bibinfo {volume} {108}},\ \bibinfo {pages}
  {035154} (\bibinfo {year} {2023})}\BibitemShut {NoStop}%
\bibitem [{\citenamefont {Lenggenhager}\ \emph
  {et~al.}(2023{\natexlab{a}})\citenamefont {Lenggenhager}, \citenamefont
  {Maciejko},\ and\ \citenamefont {Bzdu\v{s}ek}}]{Lenggenhager2023}%
  \BibitemOpen
  \bibfield  {author} {\bibinfo {author} {\bibfnamefont {Patrick~M.}\
  \bibnamefont {Lenggenhager}}, \bibinfo {author} {\bibfnamefont {Joseph}\
  \bibnamefont {Maciejko}}, \ and\ \bibinfo {author} {\bibfnamefont
  {Tom\'a\v{s}}\ \bibnamefont {Bzdu\v{s}ek}},\ }\bibfield  {title} {\enquote
  {\bibinfo {title} {{Non-Abelian Hyperbolic Band Theory from Supercells}},}\
  }\href {\doibase 10.1103/PhysRevLett.131.226401} {\bibfield  {journal}
  {\bibinfo  {journal} {Phys. Rev. Lett.}\ }\textbf {\bibinfo {volume} {131}},\
  \bibinfo {pages} {226401} (\bibinfo {year} {2023}{\natexlab{a}})}\BibitemShut
  {NoStop}%
\bibitem [{\citenamefont {Zhang}\ and\ \citenamefont {Hu}(2001)}]{zhang2001}%
  \BibitemOpen
  \bibfield  {author} {\bibinfo {author} {\bibfnamefont {Shou-Cheng}\
  \bibnamefont {Zhang}}\ and\ \bibinfo {author} {\bibfnamefont {Jiangping}\
  \bibnamefont {Hu}},\ }\bibfield  {title} {\enquote {\bibinfo {title} {A
  {F}our-{D}imensional {G}eneralization of the {Q}uantum {Hall} {E}ffect},}\
  }\href {\doibase 10.1126/science.294.5543.823} {\bibfield  {journal}
  {\bibinfo  {journal} {Science}\ }\textbf {\bibinfo {volume} {294}},\ \bibinfo
  {pages} {823--828} (\bibinfo {year} {2001})}\BibitemShut {NoStop}%
\bibitem [{\citenamefont {Qi}\ \emph {et~al.}(2008)\citenamefont {Qi},
  \citenamefont {Hughes},\ and\ \citenamefont {Zhang}}]{Qi2008}%
  \BibitemOpen
  \bibfield  {author} {\bibinfo {author} {\bibfnamefont {Xiao-Liang}\
  \bibnamefont {Qi}}, \bibinfo {author} {\bibfnamefont {Taylor~L.}\
  \bibnamefont {Hughes}}, \ and\ \bibinfo {author} {\bibfnamefont {Shou-Cheng}\
  \bibnamefont {Zhang}},\ }\bibfield  {title} {\enquote {\bibinfo {title}
  {Topological field theory of time-reversal invariant insulators},}\ }\href
  {\doibase 10.1103/PhysRevB.78.195424} {\bibfield  {journal} {\bibinfo
  {journal} {Phys. Rev. B}\ }\textbf {\bibinfo {volume} {78}},\ \bibinfo
  {pages} {195424} (\bibinfo {year} {2008})}\BibitemShut {NoStop}%
\bibitem [{\citenamefont {Sausset}\ and\ \citenamefont
  {Tarjus}(2007)}]{sausset2007}%
  \BibitemOpen
  \bibfield  {author} {\bibinfo {author} {\bibfnamefont {F.}~\bibnamefont
  {Sausset}}\ and\ \bibinfo {author} {\bibfnamefont {G.}~\bibnamefont
  {Tarjus}},\ }\bibfield  {title} {\enquote {\bibinfo {title} {Periodic
  boundary conditions on the pseudosphere},}\ }\href {\doibase
  10.1088/1751-8113/40/43/004} {\bibfield  {journal} {\bibinfo  {journal} {J.
  Phys. A: Math. Theor.}\ }\textbf {\bibinfo {volume} {40}},\ \bibinfo {pages}
  {12873--12899} (\bibinfo {year} {2007})}\BibitemShut {NoStop}%
\bibitem [{SM()}]{SM}%
  \BibitemOpen
  \href@noop {} {}\bibinfo {note} {See Supplementary Material, which cites
  additional
  Refs.~\onlinecite{LINS,FirthThesis,Robinson,Fukui:2005,nagy2022,Bronzan1988,Miranda:1995,Ryu:2010},
  for further details on PBC clusters (construction and flux threading),
  U($d$)-HBT, and supercell method (2-supercell Hamiltonian, symmetry analysis,
  mass dependence of the nodal manifold, larger $n$-supercell
  spectra).}\BibitemShut {Stop}%
\bibitem [{\citenamefont {Shankar}\ and\ \citenamefont
  {Maciejko}(2023)}]{shankar2023}%
  \BibitemOpen
  \bibfield  {author} {\bibinfo {author} {\bibfnamefont {G.}~\bibnamefont
  {Shankar}}\ and\ \bibinfo {author} {\bibfnamefont {Joseph}\ \bibnamefont
  {Maciejko}},\ }\href@noop {} {\enquote {\bibinfo {title} {Hyperbolic lattices
  and two-dimensional {Y}ang-{M}ills theory},}\ } (\bibinfo {year} {2023}),\
  \Eprint {http://arxiv.org/abs/2309.03857} {arXiv:2309.03857} \BibitemShut
  {NoStop}%
\bibitem [{\citenamefont {Lenggenhager}\ \emph
  {et~al.}(2023{\natexlab{b}})\citenamefont {Lenggenhager}, \citenamefont
  {Maciejko},\ and\ \citenamefont {Bzdu\v{s}ek}}]{HyperCells2023}%
  \BibitemOpen
  \bibfield  {author} {\bibinfo {author} {\bibfnamefont {Patrick~M.}\
  \bibnamefont {Lenggenhager}}, \bibinfo {author} {\bibfnamefont {Joseph}\
  \bibnamefont {Maciejko}}, \ and\ \bibinfo {author} {\bibfnamefont
  {Tom\'{a}\v{s}}\ \bibnamefont {Bzdu\v{s}ek}},\ }\href {\doibase
  10.5281/zenodo.10222599} {\enquote {\bibinfo {title} {{H}yper{C}ells: {A}
  {GAP} package for constructing primitive cells and supercells of hyperbolic
  lattices},}\ }\bibinfo {howpublished}
  {\url{https://github.com/patrick-lenggenhager/HyperCells}} (\bibinfo {year}
  {2023}{\natexlab{b}})\BibitemShut {NoStop}%
\bibitem [{\citenamefont {Conder}(2007)}]{Conder2007}%
  \BibitemOpen
  \bibfield  {author} {\bibinfo {author} {\bibfnamefont {Marston}\ \bibnamefont
  {Conder}},\ }\href
  {https://www.math.auckland.ac.nz/~conder/TriangleGroupQuotients101.txt}
  {\enquote {\bibinfo {title} {Quotients of triangle groups acting on surfaces
  of genus 2 to 101},}\ } (\bibinfo {year} {2007})\BibitemShut {NoStop}%
\bibitem [{GAP()}]{GAP4}%
  \BibitemOpen
  GAP,\ \href {https://www.gap-system.org} {\emph {\bibinfo {title} {{GAP --
  Groups, Algorithms, and Programming, Version 4.11.1}}}},\ \bibinfo
  {organization} {The GAP~Group} (\bibinfo {year} {2021})\BibitemShut {NoStop}%
\bibitem [{\citenamefont {Lenggenhager}\ \emph
  {et~al.}(2023{\natexlab{c}})\citenamefont {Lenggenhager}, \citenamefont
  {Maciejko},\ and\ \citenamefont {Bzdu\v{s}ek}}]{HyperBloch}%
  \BibitemOpen
  \bibfield  {author} {\bibinfo {author} {\bibfnamefont {Patrick~M.}\
  \bibnamefont {Lenggenhager}}, \bibinfo {author} {\bibfnamefont {Joseph}\
  \bibnamefont {Maciejko}}, \ and\ \bibinfo {author} {\bibfnamefont
  {Tom\'{a}\v{s}}\ \bibnamefont {Bzdu\v{s}ek}},\ }\href {\doibase
  10.5281/zenodo.10222866} {\enquote {\bibinfo {title} {{HyperBloch}: {A}
  {M}athematica package for hyperbolic tight-binding models and the supercell
  method},}\ }\bibinfo {howpublished}
  {\url{https://github.com/patrick-lenggenhager/HyperBloch}} (\bibinfo {year}
  {2023}{\natexlab{c}})\BibitemShut {NoStop}%
\bibitem [{\citenamefont {Mochol-Grzelak}\ \emph {et~al.}(2018)\citenamefont
  {Mochol-Grzelak}, \citenamefont {Dauphin}, \citenamefont {Celi},\ and\
  \citenamefont {Lewenstein}}]{MocholGrzelak2018}%
  \BibitemOpen
  \bibfield  {author} {\bibinfo {author} {\bibfnamefont {M}~\bibnamefont
  {Mochol-Grzelak}}, \bibinfo {author} {\bibfnamefont {A}~\bibnamefont
  {Dauphin}}, \bibinfo {author} {\bibfnamefont {A}~\bibnamefont {Celi}}, \ and\
  \bibinfo {author} {\bibfnamefont {M}~\bibnamefont {Lewenstein}},\ }\bibfield
  {title} {\enquote {\bibinfo {title} {Efficient algorithm to compute the
  second {Chern} number in four dimensional systems},}\ }\href {\doibase
  10.1088/2058-9565/aae93b} {\bibfield  {journal} {\bibinfo  {journal} {Quantum
  Sci. Technol.}\ }\textbf {\bibinfo {volume} {4}},\ \bibinfo {pages} {014009}
  (\bibinfo {year} {2018})}\BibitemShut {NoStop}%
\bibitem [{\citenamefont {von Neumann}\ and\ \citenamefont
  {Wigner}(1929)}]{vonNeumann:1929}%
  \BibitemOpen
  \bibfield  {author} {\bibinfo {author} {\bibfnamefont {J.}~\bibnamefont {von
  Neumann}}\ and\ \bibinfo {author} {\bibfnamefont {E.}~\bibnamefont
  {Wigner}},\ }\bibfield  {title} {\enquote {\bibinfo {title} {\"{U}ber das
  {V}erhalten von {E}igenwerten bei adiabatischen {P}rozessen},}\ }\href
  {\doibase 10.1007/978-3-662-02781-3_20} {\bibfield  {journal} {\bibinfo
  {journal} {Phys. Z.}\ }\textbf {\bibinfo {volume} {30}},\ \bibinfo {pages}
  {465} (\bibinfo {year} {1929})}\BibitemShut {NoStop}%
\bibitem [{\citenamefont {Bzdu\v{s}ek}\ and\ \citenamefont
  {Sigrist}(2017)}]{Bzdusek:2017}%
  \BibitemOpen
  \bibfield  {author} {\bibinfo {author} {\bibfnamefont {Tom\'a\v{s}}\
  \bibnamefont {Bzdu\v{s}ek}}\ and\ \bibinfo {author} {\bibfnamefont {Manfred}\
  \bibnamefont {Sigrist}},\ }\bibfield  {title} {\enquote {\bibinfo {title}
  {Robust doubly charged nodal lines and nodal surfaces in centrosymmetric
  systems},}\ }\href {\doibase 10.1103/PhysRevB.96.155105} {\bibfield
  {journal} {\bibinfo  {journal} {Phys. Rev. B}\ }\textbf {\bibinfo {volume}
  {96}},\ \bibinfo {pages} {155105} (\bibinfo {year} {2017})}\BibitemShut
  {NoStop}%
\bibitem [{\citenamefont {Price}\ \emph {et~al.}(2015)\citenamefont {Price},
  \citenamefont {Zilberberg}, \citenamefont {Ozawa}, \citenamefont
  {Carusotto},\ and\ \citenamefont {Goldman}}]{Price:2015}%
  \BibitemOpen
  \bibfield  {author} {\bibinfo {author} {\bibfnamefont {H.~M.}\ \bibnamefont
  {Price}}, \bibinfo {author} {\bibfnamefont {O.}~\bibnamefont {Zilberberg}},
  \bibinfo {author} {\bibfnamefont {T.}~\bibnamefont {Ozawa}}, \bibinfo
  {author} {\bibfnamefont {I.}~\bibnamefont {Carusotto}}, \ and\ \bibinfo
  {author} {\bibfnamefont {N.}~\bibnamefont {Goldman}},\ }\bibfield  {title}
  {\enquote {\bibinfo {title} {Four-{D}imensional {Q}uantum {H}all {E}ffect
  with {U}ltracold {A}toms},}\ }\href {\doibase 10.1103/PhysRevLett.115.195303}
  {\bibfield  {journal} {\bibinfo  {journal} {Phys. Rev. Lett.}\ }\textbf
  {\bibinfo {volume} {115}},\ \bibinfo {pages} {195303} (\bibinfo {year}
  {2015})}\BibitemShut {NoStop}%
\bibitem [{\citenamefont {Huang}\ \emph {et~al.}(2024)\citenamefont {Huang},
  \citenamefont {He}, \citenamefont {Zhang}, \citenamefont {Zhang},
  \citenamefont {Liu}, \citenamefont {Feng}, \citenamefont {Liu}, \citenamefont
  {Cui}, \citenamefont {Huang}, \citenamefont {Zhang},\ and\ \citenamefont
  {Zhang}}]{Huang:2024}%
  \BibitemOpen
  \bibfield  {author} {\bibinfo {author} {\bibfnamefont {Lei}\ \bibnamefont
  {Huang}}, \bibinfo {author} {\bibfnamefont {Lu}~\bibnamefont {He}}, \bibinfo
  {author} {\bibfnamefont {Weixuan}\ \bibnamefont {Zhang}}, \bibinfo {author}
  {\bibfnamefont {Huizhen}\ \bibnamefont {Zhang}}, \bibinfo {author}
  {\bibfnamefont {Dongning}\ \bibnamefont {Liu}}, \bibinfo {author}
  {\bibfnamefont {Xue}\ \bibnamefont {Feng}}, \bibinfo {author} {\bibfnamefont
  {Fang}\ \bibnamefont {Liu}}, \bibinfo {author} {\bibfnamefont {Kaiyu}\
  \bibnamefont {Cui}}, \bibinfo {author} {\bibfnamefont {Yidong}\ \bibnamefont
  {Huang}}, \bibinfo {author} {\bibfnamefont {Wei}\ \bibnamefont {Zhang}}, \
  and\ \bibinfo {author} {\bibfnamefont {Xiangdong}\ \bibnamefont {Zhang}},\
  }\href@noop {} {\enquote {\bibinfo {title} {Hyperbolic photonic topological
  insulators},}\ } (\bibinfo {year} {2024}),\ \Eprint
  {http://arxiv.org/abs/2401.16006} {arXiv:2401.16006 [physics.optics]}
  \BibitemShut {NoStop}%
\bibitem [{\citenamefont {Kaufman}\ and\ \citenamefont
  {Ni}(2021)}]{Kaufman2021}%
  \BibitemOpen
  \bibfield  {author} {\bibinfo {author} {\bibfnamefont {Adam~M.}\ \bibnamefont
  {Kaufman}}\ and\ \bibinfo {author} {\bibfnamefont {Kang-Kuen}\ \bibnamefont
  {Ni}},\ }\bibfield  {title} {\enquote {\bibinfo {title} {{Quantum science
  with optical tweezer arrays of ultracold atoms and molecules}},}\ }\href
  {\doibase 10.1038/s41567-021-01357-2} {\bibfield  {journal} {\bibinfo
  {journal} {Nat. Phys.}\ }\textbf {\bibinfo {volume} {17}},\ \bibinfo {pages}
  {1324} (\bibinfo {year} {2021})}\BibitemShut {NoStop}%
\bibitem [{\citenamefont {Spar}\ \emph {et~al.}(2022)\citenamefont {Spar},
  \citenamefont {Guardado-Sanchez}, \citenamefont {Chi}, \citenamefont {Yan},\
  and\ \citenamefont {Bakr}}]{Spar2022}%
  \BibitemOpen
  \bibfield  {author} {\bibinfo {author} {\bibfnamefont {Benjamin~M.}\
  \bibnamefont {Spar}}, \bibinfo {author} {\bibfnamefont {Elmer}\ \bibnamefont
  {Guardado-Sanchez}}, \bibinfo {author} {\bibfnamefont {Sungjae}\ \bibnamefont
  {Chi}}, \bibinfo {author} {\bibfnamefont {Zoe~Z.}\ \bibnamefont {Yan}}, \
  and\ \bibinfo {author} {\bibfnamefont {Waseem~S.}\ \bibnamefont {Bakr}},\
  }\bibfield  {title} {\enquote {\bibinfo {title} {Realization of a
  {Fermi-Hubbard} optical tweezer array},}\ }\href {\doibase
  10.1103/PhysRevLett.128.223202} {\bibfield  {journal} {\bibinfo  {journal}
  {Phys. Rev. Lett.}\ }\textbf {\bibinfo {volume} {128}},\ \bibinfo {pages}
  {223202} (\bibinfo {year} {2022})}\BibitemShut {NoStop}%
\bibitem [{\citenamefont {Monroe}\ \emph {et~al.}(2021)\citenamefont {Monroe},
  \citenamefont {Campbell}, \citenamefont {Duan}, \citenamefont {Gong},
  \citenamefont {Gorshkov}, \citenamefont {Hess}, \citenamefont {Islam},
  \citenamefont {Kim}, \citenamefont {Linke}, \citenamefont {Pagano},
  \citenamefont {Richerme}, \citenamefont {Senko},\ and\ \citenamefont
  {Yao}}]{Monroe:2021}%
  \BibitemOpen
  \bibfield  {author} {\bibinfo {author} {\bibfnamefont {C.}~\bibnamefont
  {Monroe}}, \bibinfo {author} {\bibfnamefont {W.~C.}\ \bibnamefont
  {Campbell}}, \bibinfo {author} {\bibfnamefont {L.-M.}\ \bibnamefont {Duan}},
  \bibinfo {author} {\bibfnamefont {Z.-X.}\ \bibnamefont {Gong}}, \bibinfo
  {author} {\bibfnamefont {A.~V.}\ \bibnamefont {Gorshkov}}, \bibinfo {author}
  {\bibfnamefont {P.~W.}\ \bibnamefont {Hess}}, \bibinfo {author}
  {\bibfnamefont {R.}~\bibnamefont {Islam}}, \bibinfo {author} {\bibfnamefont
  {K.}~\bibnamefont {Kim}}, \bibinfo {author} {\bibfnamefont {N.~M.}\
  \bibnamefont {Linke}}, \bibinfo {author} {\bibfnamefont {G.}~\bibnamefont
  {Pagano}}, \bibinfo {author} {\bibfnamefont {P.}~\bibnamefont {Richerme}},
  \bibinfo {author} {\bibfnamefont {C.}~\bibnamefont {Senko}}, \ and\ \bibinfo
  {author} {\bibfnamefont {N.~Y.}\ \bibnamefont {Yao}},\ }\bibfield  {title}
  {\enquote {\bibinfo {title} {Programmable quantum simulations of spin systems
  with trapped ions},}\ }\href {\doibase 10.1103/RevModPhys.93.025001}
  {\bibfield  {journal} {\bibinfo  {journal} {Rev. Mod. Phys.}\ }\textbf
  {\bibinfo {volume} {93}},\ \bibinfo {pages} {025001} (\bibinfo {year}
  {2021})}\BibitemShut {NoStop}%
\bibitem [{\citenamefont {Witten}(1998)}]{Witten:1998}%
  \BibitemOpen
  \bibfield  {author} {\bibinfo {author} {\bibfnamefont {Edward}\ \bibnamefont
  {Witten}},\ }\bibfield  {title} {\enquote {\bibinfo {title} {{Anti-de
  {S}itter space and holography}},}\ }\href {\doibase
  10.4310/atmp.1998.v2.n2.a2} {\bibfield  {journal} {\bibinfo  {journal} {Adv.
  Theor. Math. Phys.}\ }\textbf {\bibinfo {volume} {2}},\ \bibinfo {pages}
  {253} (\bibinfo {year} {1998})}\BibitemShut {NoStop}%
\bibitem [{\citenamefont {Zaanen}\ \emph {et~al.}(2015)\citenamefont {Zaanen},
  \citenamefont {Sun}, \citenamefont {Liu},\ and\ \citenamefont
  {Schalm}}]{Zaanen:2015}%
  \BibitemOpen
  \bibfield  {author} {\bibinfo {author} {\bibfnamefont {Jan}\ \bibnamefont
  {Zaanen}}, \bibinfo {author} {\bibfnamefont {Ya-Wen}\ \bibnamefont {Sun}},
  \bibinfo {author} {\bibfnamefont {Yan}\ \bibnamefont {Liu}}, \ and\ \bibinfo
  {author} {\bibfnamefont {Koenraad}\ \bibnamefont {Schalm}},\ }\href@noop {}
  {\emph {\bibinfo {title} {Holographic {D}uality in {C}ondensed {M}atter
  {P}hysics}}}\ (\bibinfo  {publisher} {Cambridge University Press},\ \bibinfo
  {address} {Cambridge, England},\ \bibinfo {year} {2015})\BibitemShut
  {NoStop}%
\bibitem [{\citenamefont {Asaduzzaman}\ \emph {et~al.}(2020)\citenamefont
  {Asaduzzaman}, \citenamefont {Catterall}, \citenamefont {Hubisz},
  \citenamefont {Nelson},\ and\ \citenamefont
  {Unmuth-Yockey}}]{Asaduzzaman:2020}%
  \BibitemOpen
  \bibfield  {author} {\bibinfo {author} {\bibfnamefont {Muhammad}\
  \bibnamefont {Asaduzzaman}}, \bibinfo {author} {\bibfnamefont {Simon}\
  \bibnamefont {Catterall}}, \bibinfo {author} {\bibfnamefont {Jay}\
  \bibnamefont {Hubisz}}, \bibinfo {author} {\bibfnamefont {Roice}\
  \bibnamefont {Nelson}}, \ and\ \bibinfo {author} {\bibfnamefont {Judah}\
  \bibnamefont {Unmuth-Yockey}},\ }\bibfield  {title} {\enquote {\bibinfo
  {title} {{Holography on tessellations of hyperbolic space}},}\ }\href
  {\doibase 10.1103/PhysRevD.102.034511} {\bibfield  {journal} {\bibinfo
  {journal} {Phys. Rev. D}\ }\textbf {\bibinfo {volume} {102}},\ \bibinfo
  {pages} {034511} (\bibinfo {year} {2020})}\BibitemShut {NoStop}%
\bibitem [{\citenamefont {Chen}\ \emph
  {et~al.}(2023{\natexlab{c}})\citenamefont {Chen}, \citenamefont {Chen},
  \citenamefont {Yang}, \citenamefont {Yang}, \citenamefont {Chen},
  \citenamefont {Meng}, \citenamefont {Yan}, \citenamefont {Xi}, \citenamefont
  {Zhu}, \citenamefont {Liu}, \citenamefont {Shum}, \citenamefont {Chen},
  \citenamefont {Cai}, \citenamefont {Yang}, \citenamefont {Yang},\ and\
  \citenamefont {Gao}}]{Chen:2023:Ads/CFT}%
  \BibitemOpen
  \bibfield  {author} {\bibinfo {author} {\bibfnamefont {Jingming}\
  \bibnamefont {Chen}}, \bibinfo {author} {\bibfnamefont {Feiyu}\ \bibnamefont
  {Chen}}, \bibinfo {author} {\bibfnamefont {Yuting}\ \bibnamefont {Yang}},
  \bibinfo {author} {\bibfnamefont {Linyun}\ \bibnamefont {Yang}}, \bibinfo
  {author} {\bibfnamefont {Zihan}\ \bibnamefont {Chen}}, \bibinfo {author}
  {\bibfnamefont {Yan}\ \bibnamefont {Meng}}, \bibinfo {author} {\bibfnamefont
  {Bei}\ \bibnamefont {Yan}}, \bibinfo {author} {\bibfnamefont {Xiang}\
  \bibnamefont {Xi}}, \bibinfo {author} {\bibfnamefont {Zhenxiao}\ \bibnamefont
  {Zhu}}, \bibinfo {author} {\bibfnamefont {Gui-Geng}\ \bibnamefont {Liu}},
  \bibinfo {author} {\bibfnamefont {Perry~Ping}\ \bibnamefont {Shum}}, \bibinfo
  {author} {\bibfnamefont {Hongsheng}\ \bibnamefont {Chen}}, \bibinfo {author}
  {\bibfnamefont {Rong-Gen}\ \bibnamefont {Cai}}, \bibinfo {author}
  {\bibfnamefont {Run-Qiu}\ \bibnamefont {Yang}}, \bibinfo {author}
  {\bibfnamefont {Yihao}\ \bibnamefont {Yang}}, \ and\ \bibinfo {author}
  {\bibfnamefont {Zhen}\ \bibnamefont {Gao}},\ }\href@noop {} {\enquote
  {\bibinfo {title} {{A}ds/{C}{F}{T} {C}orrespondence in {H}yperbolic
  {L}attices},}\ } (\bibinfo {year} {2023}{\natexlab{c}}),\ \Eprint
  {http://arxiv.org/abs/2305.04862} {arXiv:2305.04862} \BibitemShut {NoStop}%
\bibitem [{\citenamefont {Tummuru}\ \emph {et~al.}(2023)\citenamefont
  {Tummuru}, \citenamefont {Chen}, \citenamefont {Lenggenhager}, \citenamefont
  {Neupert}, \citenamefont {Maciejko},\ and\ \citenamefont
  {Bzdu\v{s}ek}}]{Tummuru:2023:SDC}%
  \BibitemOpen
  \bibfield  {author} {\bibinfo {author} {\bibfnamefont {Tarun}\ \bibnamefont
  {Tummuru}}, \bibinfo {author} {\bibfnamefont {Anffany}\ \bibnamefont {Chen}},
  \bibinfo {author} {\bibfnamefont {Patrick~M.}\ \bibnamefont {Lenggenhager}},
  \bibinfo {author} {\bibfnamefont {Titus}\ \bibnamefont {Neupert}}, \bibinfo
  {author} {\bibfnamefont {Joseph}\ \bibnamefont {Maciejko}}, \ and\ \bibinfo
  {author} {\bibfnamefont {Tom\'{a}\v{s}}\ \bibnamefont {Bzdu\v{s}ek}},\ }\href
  {\doibase 10.5281/zenodo.10729119} {\enquote {\bibinfo {title}
  {{Supplementary Data and Code for Hyperbolic non-Abelian semimetal}},}\ }
  (\bibinfo {year} {2023})\BibitemShut {NoStop}%
\bibitem [{\citenamefont {Rober}(2020)}]{LINS}%
  \BibitemOpen
  \bibfield  {author} {\bibinfo {author} {\bibfnamefont {F.}~\bibnamefont
  {Rober}},\ }\href@noop {} {\enquote {\bibinfo {title} {The {GAP} package
  {LINS}},}\ }\bibinfo {howpublished}
  {\url{https://github.com/FriedrichRober/LINS}} (\bibinfo {year}
  {2020})\BibitemShut {NoStop}%
\bibitem [{\citenamefont {Firth}(2004)}]{FirthThesis}%
  \BibitemOpen
  \bibfield  {author} {\bibinfo {author} {\bibfnamefont {D.}~\bibnamefont
  {Firth}},\ }\emph {\bibinfo {title} {An Algorithm to Find Normal Subgroups of
  a Finitely Presented Group, up to a Given Finite Index}},\ \href@noop {}
  {Ph.D. thesis},\ \bibinfo  {school} {University of Warwick} (\bibinfo {year}
  {2004})\BibitemShut {NoStop}%
\bibitem [{\citenamefont {Robinson}(1996)}]{Robinson}%
  \BibitemOpen
  \bibfield  {author} {\bibinfo {author} {\bibfnamefont {D.~J.~S.}\
  \bibnamefont {Robinson}},\ }\href@noop {} {\emph {\bibinfo {title} {A Course
  in the Theory of Groups}}},\ \bibinfo {edition} {2nd}\ ed.\ (\bibinfo
  {publisher} {Springer},\ \bibinfo {address} {New York},\ \bibinfo {year}
  {1996})\BibitemShut {NoStop}%
\bibitem [{\citenamefont {Fukui}\ \emph {et~al.}(2005)\citenamefont {Fukui},
  \citenamefont {Hatsugai},\ and\ \citenamefont {Suzuki}}]{Fukui:2005}%
  \BibitemOpen
  \bibfield  {author} {\bibinfo {author} {\bibfnamefont {Takahiro}\
  \bibnamefont {Fukui}}, \bibinfo {author} {\bibfnamefont {Yasuhiro}\
  \bibnamefont {Hatsugai}}, \ and\ \bibinfo {author} {\bibfnamefont {Hiroshi}\
  \bibnamefont {Suzuki}},\ }\bibfield  {title} {\enquote {\bibinfo {title}
  {Chern numbers in discretized {B}rillouin zone: efficient method of computing
  (spin) {H}all conductances},}\ }\href {\doibase 10.1143/JPSJ.74.1674}
  {\bibfield  {journal} {\bibinfo  {journal} {J. Phys. Soc. Japan}\ }\textbf
  {\bibinfo {volume} {74}},\ \bibinfo {pages} {1674--1677} (\bibinfo {year}
  {2005})}\BibitemShut {NoStop}%
\bibitem [{\citenamefont {Nagy}\ and\ \citenamefont {Rayan}(2023)}]{nagy2022}%
  \BibitemOpen
  \bibfield  {author} {\bibinfo {author} {\bibfnamefont {{\'{A}}kos}\
  \bibnamefont {Nagy}}\ and\ \bibinfo {author} {\bibfnamefont {Steven}\
  \bibnamefont {Rayan}},\ }\bibfield  {title} {\enquote {\bibinfo {title} {On
  the hyperbolic bloch transform},}\ }\href {\doibase
  10.1007/s00023-023-01336-8} {\bibfield  {journal} {\bibinfo  {journal}
  {Annales Henri Poincaré}\ }\textbf {\bibinfo {volume} {25}},\ \bibinfo
  {pages} {1713–1732} (\bibinfo {year} {2023})}\BibitemShut {NoStop}%
\bibitem [{\citenamefont {Bronzan}(1988)}]{Bronzan1988}%
  \BibitemOpen
  \bibfield  {author} {\bibinfo {author} {\bibfnamefont {J.~B.}\ \bibnamefont
  {Bronzan}},\ }\bibfield  {title} {\enquote {\bibinfo {title} {Parametrization
  of {$SU(3)$}},}\ }\href {\doibase 10.1103/physrevd.38.1994} {\bibfield
  {journal} {\bibinfo  {journal} {Phys. Rev. D}\ }\textbf {\bibinfo {volume}
  {38}},\ \bibinfo {pages} {1994--1999} (\bibinfo {year} {1988})}\BibitemShut
  {NoStop}%
\bibitem [{\citenamefont {Miranda}(1995)}]{Miranda:1995}%
  \BibitemOpen
  \bibfield  {author} {\bibinfo {author} {\bibfnamefont {R.}~\bibnamefont
  {Miranda}},\ }\href@noop {} {\emph {\bibinfo {title} {Algebraic Curves and
  Riemann Surfaces}}}\ (\bibinfo  {publisher} {American Mathematical Society},\
  \bibinfo {address} {Providence},\ \bibinfo {year} {1995})\BibitemShut
  {NoStop}%
\bibitem [{\citenamefont {Ryu}\ \emph {et~al.}(2010)\citenamefont {Ryu},
  \citenamefont {Schnyder}, \citenamefont {Furusaki},\ and\ \citenamefont
  {Ludwig}}]{Ryu:2010}%
  \BibitemOpen
  \bibfield  {author} {\bibinfo {author} {\bibfnamefont {Shinsei}\ \bibnamefont
  {Ryu}}, \bibinfo {author} {\bibfnamefont {Andreas~P.}\ \bibnamefont
  {Schnyder}}, \bibinfo {author} {\bibfnamefont {Akira}\ \bibnamefont
  {Furusaki}}, \ and\ \bibinfo {author} {\bibfnamefont {Andreas W.~W.}\
  \bibnamefont {Ludwig}},\ }\bibfield  {title} {\enquote {\bibinfo {title}
  {Topological insulators and superconductors: tenfold way and dimensional
  hierarchy},}\ }\href {\doibase 10.1088/1367-2630/12/6/065010} {\bibfield
  {journal} {\bibinfo  {journal} {New J. Phys.}\ }\textbf {\bibinfo {volume}
  {12}},\ \bibinfo {pages} {065010} (\bibinfo {year} {2010})}\BibitemShut
  {NoStop}%
\end{thebibliography}%


\let\addcontentsline\oldaddcontentsline     

\clearpage

\renewcommand{\theequation}{S\arabic{equation}}
\renewcommand{\thefigure}{S\arabic{figure}}
\renewcommand{\theHfigure}{S\arabic{figure}}

\setcounter{page}{1}
\setcounter{equation}{0}
\setcounter{section}{0}
\setcounter{figure}{0}

\title{Supplementary Material for: Hyperbolic non-Abelian semimetal}
\maketitle

\begingroup
\hypersetup{linkcolor=black}
\tableofcontents
\endgroup


\begin{figure*}[t]
\includegraphics[width=\textwidth]{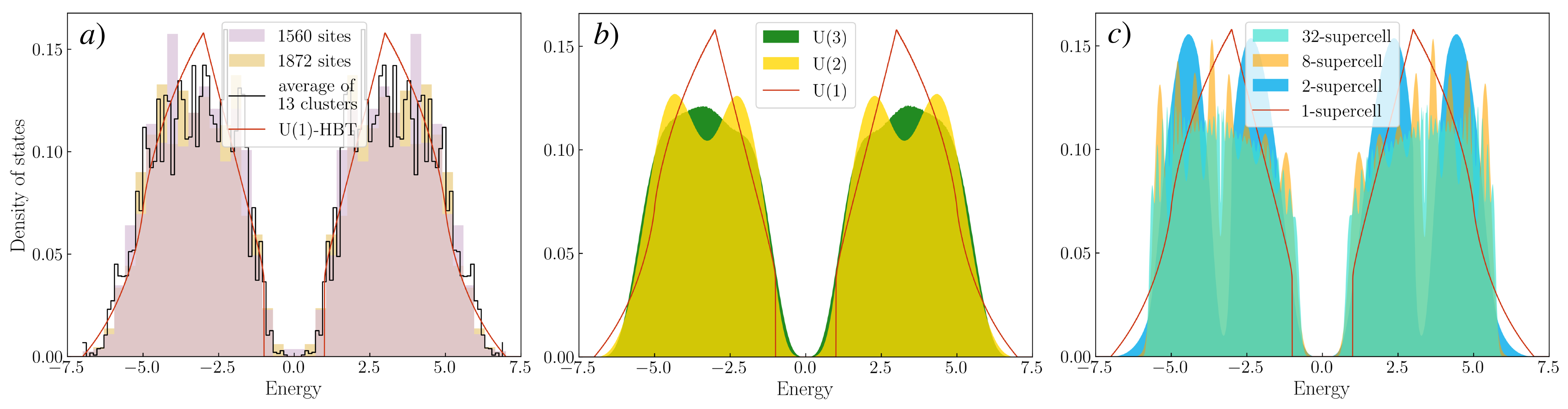}
    \caption{A comparison of different techniques to access non-Abelian Bloch states
    in the hyperbolic QHI model at $m=3$. 
    (a)~Sample DOS of two PBC clusters, alongside the average of 13 clusters (120 bins), per the discussion in Sec.~\ref{app:clusters}. 
    (b)~$\textrm{U}(d)$-HBT for $d=1,2$ and $3$, as described in Sec.~\ref{app:sud}. 
    While $\textrm{U}(1)$-HBT is gapped, spectra of higher-dimensional IRs are semimetallic. 
    (c)~Modelling with the supercell approach dicussed in Sec.~\ref{app:supercell}.
    Going beyond the 2-supercell, we observe a gradual convergence to a smooth DOS curve in the limit of large supercells. 
    The converged DOS curve appears to be insulating at $m=3$, see Sec.~\ref{sec:larger-sc}.}
    \label{fig:comp} 
\end{figure*}

\section{PBC clusters}\label{app:clusters}
PBC clusters are finite-sized hyperbolic lattices with periodic boundary conditions (PBC). Owing to the spatial curvature, applying PBC to a finite $\{p,q\}$
hyperbolic lattice is a nontrivial task relying on geometric group theory~\cite{sausset2007,Maciejko2022,Lux:2022,Lux2023}. In this section, we describe our method for constructing PBC clusters, which applies to $\{p,q\}$ lattices with known hyperbolic Bravais lattices and corresponding unit cells \cite{Boettcher:2022,Chen2023a,Chen2023b}. 
The hyperbolic translation group $\HTG$ of a hyperbolic Bravais lattice is generated by elements $\htg{j}$ which relate pairs of edges of the primitive cell. 
A PBC cluster with $N$ primitive cells is defined by a choice of normal subgroup $\nsg\triangleleft\HTG$ of index $N$ corresponding to a set of translations under which the wave function is invariant~\cite{Maciejko2022}. 
Given such a normal subgroup, one can compute the factor group $\HTG/\nsg$, which is a finite group of order $N$ and corresponds to a group of translations modulo PBC on the finite cluster. 
The structure of the factor group allows us to construct an adjacency matrix $A$ on the cluster, in which each lattice site corresponds to a coset $[\htgel]\in \HTG/\nsg$. We use $[\htgel]$ to denote
the right coset of $\htgel\in \HTG$, i.e., $[\htgel]=\nsg \htgel=\{\tilde{\htgel} \htgel:\tilde{\htgel}\in\nsg\}$. More complicated tight-binding models---with multiple on-site degrees of freedom (e.g., orbitals or sublattices)---can then be built by tensoring the adjacency matrix $A$ with finite-dimensional matrices in orbital/sublattice space. 
PBC clusters with Abelian factor group $\HTG/\nsg$ are termed Abelian clusters, and such clusters exclude all higher-dimensional IRs of the translation group. 
To study the gapless spectrum of the hyperbolic non-Abelian semimetal, we focus on non-Abelian clusters, whose factor group is non-Abelian and allows us to discretely sample the spectrum of non-Abelian Bloch states.

In the following, we outline our methodology for constructing the tight-binding Dirac model on \textit{arbitrarily large} non-Abelian PBC clusters of the \{$8,8$\} hyperbolic (Bravais) lattice, with hyperbolic translation group $\HTG$ defined in \cref{eq:constraint}. \cref{app:PBC-NSG} presents a method based on subgroup intersection to overcome the computational limitations of Ref.~\cite{Maciejko2022}, in which only modest system sizes $N\leq 25$ had been reached. The algorithm to construct the adjacency matrix $A$ on a PBC cluster is explained in \cref{app:PBC-adj}. Finally, in \cref{app:PBC-dirac}, we construct the Hamiltonian for the hyperbolized QHI model.

\subsection{Large-index normal subgroups from intersection}
\label{app:PBC-NSG}

In Ref.~\cite{Maciejko2022}, a GAP implementation~\cite{LINS,GAP4} of the low-index normal subgroups procedure~\cite{FirthThesis} was used to enumerate all normal subgroups $\nsg$ up to a given finite index $N_\text{max}$. 
Such a procedure is computationally expensive, and Ref.~\cite{Maciejko2022} only presented results for PBC clusters of maximum size $N_\text{max}=25$.
To overcome these limitations, we propose an algorithm based on a few mathematical results. 
First, given two normal subgroups $\nsg^{(1)},\nsg^{(2)}$ of the Fuchsian translation group $\HTG$, we show that their intersection $\nsg^{(1)}\cap\nsg^{(2)}$ is also normal in $\HTG$ and thus defines a valid PBC cluster. 
In the following, we denote set inclusion by $\subseteq$, subgroup inclusion by $\leq$, and normal subgroup inclusion by $\triangleleft$.
\begin{lemma}\label{lemma1}
Let $H,K\triangleleft G$. Then $H\cap K\triangleleft G$.
\end{lemma}
\begin{proof}
First, we prove that $H\cap K\leq G$. If $g,h\in H\cap K$, they are also in $H$. 
Since $H$ is a group, then $gh\in H$. 
On the other hand, $g,h$ are also in $K$, and since $K$ is also a group, we have $gh\in K$. 
Since $gh$ is both in $H$ and in $K$, we have $gh\in H\cap K$. 
But this holds for every pair of elements $g,h$ in $H\cap K$, thus $H\cap K$ is a group. Since $H\cap K$ is contained in $G$ as a set, it is a subgroup of $G$.

Next, we prove that $H\cap K$ is normal in $G$. Consider an element $b$ of the set $g(H\cap K)$. Since $H$ contains $H\cap K$, it is also true that $b\in gH$, and since $K$ also contains $H\cap K$, it is also true that $b\in gK$. Therefore, $b$ must belong to the intersection $gH\cap gK$. Since this is true for any $b\in g(H\cap K)$, we have $g(H\cap K)\subseteq gH\cap gK$. However, we can also show that the converse is true: $gH\cap gK\subseteq g(H\cap K)$. Let $b\in gH\cap gK$. Then, $b\in gH$ and $b\in gK$, and there exist $h\in H$ and $k\in K$ such that $b=gh=gk$. Multiplying by $g^{-1}$, we find that $h=k\in H\cap K$. Thus, $b\in g(H\cap K)$, and since $b$ was otherwise arbitrary, we find that $gH\cap gK\subseteq g(H\cap K)$. Since both $g(H\cap K)\subseteq gH\cap gK$ and $gH\cap gK\subseteq g(H\cap K)$ are true, we necessarily have that $g(H\cap K)=gH\cap gK$. We leave it as an exercise to the reader to also prove using exactly the same arguments that $(H\cap K)g=Hg\cap Kg$.

We finally show that the latter two equalities imply normality of $H\cap K$ in $G$. Let $g\in G$, then $H\triangleleft G$ implies $gH=Hg$, and $K\triangleleft G$ implies $gK=Kg$. Therefore, $g(H\cap K)=gH\cap gK=Hg\cap Kg=(H\cap K)g$, thus $H\cap K\triangleleft G$.
\end{proof}
Next, we show that the index of $\nsg^{(1)}\cap\nsg^{(2)}$ is bounded from below by the least common multiple of the indices of $\nsg^{(1)}$ and $\nsg^{(2)}$, and is thus typically larger than either individual index.
\begin{lemma}
Let $H,K\triangleleft G$ with indices $|G:H|=m$ and $|G:K|=n$. Then
\begin{align}\label{eqn:index-bounds}
\lcm(m,n)\leq|G:H\cap K|\leq mn,
\end{align}
with the equality (i.e., $|G:H\cap K|=mn$) realized for $m,n$ coprime.
\end{lemma}
\begin{proof}
First, observe that $H\cap K$ is a subgroup of both $H$ and $K$, since it is a group and is contained in both $H$ and $K$. Since $H\cap K\leq H\leq G$, we have (see 1.3.5 in Ref.~\cite{Robinson})
\begin{align}
|G:H\cap K|=|G:H||H:H\cap K|=m|H:H\cap K|.
\end{align}
Also, $H\cap K\leq K\leq G$, thus
\begin{align}
|G:H\cap K|=|G:K||K:H\cap K|=n|K:H\cap K|.
\end{align}
Hence $|G:H\cap K|$ is a common multiple of $m$ and $n$, thus $|G:H\cap K|\geq\lcm(m,n)$. For the second inequality in Eq.~(\ref{eqn:index-bounds}), 
see 1.3.11(ii) in Ref.~\cite{Robinson}; it becomes an equality for $m,n$ coprime (and $\lcm(m,n)=mn$ in that case).
\end{proof}
Finally, we show that the intersection of normal subgroups $\nsg^{(1)}$ and $\nsg^{(2)}$ corresponding to Abelian clusters can only produce another Abelian cluster. To accomplish this, we rely on the notion of commutator subgroup. The commutator subgroup of a group is generated by all commutators of the group elements. 
Roughly speaking, the commutator subgroup is a measure of non-commutativity: a large commutator subgroup indicates a `less Abelian' group.

\begin{lemma}\label{commutator}
Let $H\triangleleft G$. Then $G/H$ is Abelian if and only if $H$ contains the commutator subgroup $[G,G]=\langle [g_1,g_2]|g_1,g_2\in G\rangle$.
\end{lemma}
\begin{proof}
Denote the (right) coset of $g_1\in G$ by $[g_1]=Hg_1$. Using the normality of $H$ in $G$, we have:
\begin{align}
[g_1][g_2]&=(Hg_1)(Hg_2)=H(g_1Hg_1^{-1})g_1g_2=Hg_1g_2.
\end{align}
Likewise, $[g_2][g_1]=Hg_2g_1$. Assume that $G/H$ is Abelian, then $[g_1][g_2]=[g_2][g_1]$, which implies $Hg_1g_2=Hg_2g_1$. Multiplying by $(g_2g_1)^{-1}=g_1^{-1}g_2^{-1}$ from the right, we obtain $H[g_1,g_2]=H$ where $[g_1,g_2]=g_1g_2g_1^{-1}g_2^{-1}$ is the commutator of $g_1$ and $g_2$. This implies that $[g_1,g_2]\in H$ for all $g_1,g_2\in G$, thus $[G,G]\subseteq H$.

Now instead, assume $[G,G]\subseteq H$. For any $g_1,g_2\in G$, we have:
\begin{align}
[g_1][g_2]&=Hg_1g_2=Hg_1g_2(g_1^{-1}g_2^{-1}g_2g_1)=H[g_1,g_2]g_2g_1\nn\\
&=Hg_2g_1=[g_2][g_1],
\end{align}
hence $G/H$ is Abelian. We have used $H[g_1,g_2]=H$ which holds since $[g_1,g_2]\in [G,G]\subseteq H$.
\end{proof}

\begin{theorem}
Let $H,K\triangleleft G$. Then $G/H$ and $G/K$ are both Abelian if and only if $G/(H\cap K)$ is Abelian.
\end{theorem}
\begin{proof}
First assume that $G/H$ and $G/K$ are Abelian. By Lemma~\ref{commutator}, $[G,G]$ is contained in both $H$ and $K$. Thus $[G,G]\subseteq H\cap K$, and by the same Lemma, $G/(H\cap K)$ is Abelian.
Now assume instead that $G/(H\cap K)$ is Abelian. By the same Lemma, $[G,G]$ is contained in $H\cap K$, and thus is contained in both $H$ and $K$. As a result, $G/H$ and $G/K$ are both Abelian.
\end{proof}

Thus, for a cluster constructed from subgroup intersection to be non-Abelian, it is a necessary and sufficient condition that at least one of the two clusters in the intersection is non-Abelian.

We now describe the computational procedure for generating
normal subgroups of $\HTG$. We begin with the list of small-index normal subgroups generated by the methods of Ref.~\cite{Maciejko2022}.
We then randomly select $n_p$ parent subgroups $\{\nsg^{(p)}\}_{p=1}^{n_p}$ among this list
with indices between $N_{\rm{min}}$ and $N_{\rm{max}}$, keeping only the ones with non-Abelian factor groups. 
We then compute in GAP the intersection 
\begin{equation}
    \nsg^{\cap}\equiv\bigcap_{p=1}^{n_p}\nsg^{(p)}
\end{equation} 
of the parent subgroups. (Using Lemma~\ref{lemma1} repeatedly, it is clear that $\nsg^{\cap}$ thus constructed is normal in $\HTG$ for any~$n_p$.)

If the factor group $\HTG/\nsg^{\cap}$
is of sufficiently large order (i.e., system size), we construct the adjacency matrix for the corresponding non-Abelian PBC cluster. To obtain normal subgroups
within some target index range, we have experimented with the input parameters
$(n_{p},N_{\rm{min}},N_{\rm{max}})$. We find that input parameters $(n_{p},N_{\rm{min}},N_{\rm{max}})=(2,15,20)$
result in $\nsg^{\cap} $ with indices $N\sim 300$, while $(n_{p},N_{\rm{min}},N_{\rm{max}})=(3,10,15)$ gives indices $N\sim1500\text{--}2000$. For our analysis, we randomly chose 13 normal subgroups $\nsg^{\cap} $ with indices 
$N=$ \{304, 304, 320, 320, 306, 1560, 1584, 1716, 1800, 1800, 1872, 1980, 1980\}, which correspond to the sizes of PBC clusters [see \cref{fig:comp}(a)]. The 1D IRs make up 1/2 of all eigenstates in the first four clusters, while the fraction is 1/3 in the other clusters.
Note that such indices are completely out of reach of brute-force applications of the low-index normal subgroups procedure, as done e.g. in Refs.~\cite{Maciejko2022,Bzdusek2022}. By contrast, the method presented here can, in principle, produce normal subgroups of arbitrarily large indices.

\subsection{Adjacency matrix}\label{app:PBC-adj}

We employ GAP~\cite{GAP4} to construct an adjacency matrix $A$ from a previously 
obtained normal subgroup $\nsg^{\cap}$ with a large index $N$. If two elements $[\htgel_{n}],$ $[\htgel_{m}]\in \HTG/\nsg^{\cap} $
correspond to nearest neighbors, they must be related by one of the generators $\gamma_{j}$ of $\HTG$~\cite{Maciejko2022}:
\begin{equation}
[\htgel_{n}]=[\htgel_{m}][\gamma_{j}].
\end{equation}
Here, we include $\{\gamma_{1},\gamma_{2},\gamma_{3},\gamma_{4}\}$ and their inverses in the generating set:
\begin{equation}
    \{\gamma_{j}\}_{j=1}^8\equiv\{\gamma_{1},\gamma_{2},\gamma_{3},\gamma_{4},\gamma_{1}^{-1},\gamma_{2}^{-1},\gamma_{3}^{-1},\gamma_{4}^{-1}\}.
\end{equation}
The algorithm proceeds through the following steps:

\begin{enumerate}
\item Compute the factor group $\HTG/\nsg^{\cap} $.
\item Compute the right cosets of generators \{$\gamma_{j}$\}$_{j=1}^8$:
\begin{enumerate}
\item Construct the homomorphism $\Phi:\HTG\rightarrow \HTG/\nsg^{\cap}$ which sends each element
$\htgel\in \HTG$ to the right coset $[\htgel]=\nsg^{\cap} \htgel$. 
\item Compute the right coset $[\gamma_{j}]=\Phi(\gamma_{j})$ for $j=1,\ldots,8$.
\end{enumerate}
\item Construct the adjacency matrix $A$:
\begin{enumerate}
\item Initialize $A$ as an $N\times N$ matrix with zero entries.
\item For each pair of elements $[\htgel_{n}],$ $[\htgel_{m}]\in \HTG/\nsg^{\cap}$, if $[\htgel_{n}]=[\htgel_{m}][\gamma_{j}]$
for some integer $j\in [1,8]$, then $[\htgel_{n}]$ and
$[\htgel_{m}]$ are nearest neighbors and we let $A_{nm}\rightarrow A_{nm}+1$. Note that it is possible to have $A_{nm}\ge 1$, implying that there are multiple generators relating $[\htgel_{n}]$ and
$[\htgel_{m}]$. This is uncommon for large clusters.
\end{enumerate}
\item Construct a matrix $B$ to record the indices of all the generators that relate
a given pair of neighbors:
\begin{enumerate}
\item Initialize $B$ as a three-dimensional $N\times N\times 8$ array with zero entries.
\item For each pair of elements $[\htgel_{n}],$ $[\htgel_{m}]\in \HTG/\nsg^{\cap}$ and for each index  $j\in [1,8]$, if $[\htgel_{n}]=[\htgel_{m}][\gamma_{j}]$,  then let $B_{n,m,j}=1$;  otherwise let $B_{n,m,j}=0$.
\end{enumerate}
\end{enumerate}
Step 4 is required for implementing additional on-site degrees
of freedom (e.g., orbitals or sublattices), as in the Dirac model (see \cref{app:PBC-dirac}).

\subsection{Hamiltonian construction}\label{app:PBC-dirac}

The Dirac model on a ${\{8,8\}}$ PBC cluster with $N$ sites is described by a $4N\times4N$
Hamiltonian matrix $H$. A 4-component spinor $\psi$ lives on each
site. Nearest-neighboring spinors are coupled by inter-site matrices
\begin{equation}
T_{j}=\frac{\GM{5}-\mathrm{i}\GM{j}}{2}\text{, for }j=1,\ldots,4,
\label{eq:inter-site-mat}
\end{equation}
and their Hermitian conjugates,
\begin{equation}
T_{j}=T_{j-4}^{\dag}=\frac{\GM{5}+\mathrm{i}\GM{j-4}}{2}\text{, for }j=5,\ldots,8,
\end{equation}
where $\GM{\mu}=(\GM\mu)^\dag$ with $\mu=1,\ldots,5$ are the Dirac matrices.
Given the adjacency matrix $A$ and the generator-label matrix $B$
of a PBC cluster, the procedure for defining the Dirac model is as
follows:
\begin{enumerate}
\item Initialize $H$ as a $4N\times4N$ matrix with zero entries.
\item For each pair of neighboring sites $n$ and $m$ such that $A_{nm}>0$, use matrix $B$ to recall the indices of the generators that relate sites $n$ and $m$. For each recalled index $j'$,
add to $H$ the tensor product of $M$, which is a $N\times N$
matrix with zeros everywhere except at $M_{nm}=1$, and $T_{j'}$:
\begin{equation}
H\rightarrow H+M\otimes T_{j'}.
\end{equation}
\item To add a nonzero mass $m$, add to $H$ the tensor product of the identity
matrix $I_{N\times N}$ and $\Gamma_{5}$:
\begin{equation}
H\rightarrow H+mI_{N\times N}\otimes\GM{5}.
\end{equation}
\end{enumerate}
The DOS for $H$ thus constructed is shown for various PBC clusters and compared with $U(1)$-HBT in Fig.~\ref{fig:comp}(a).


\subsection{Flux insertion}
\label{app:flux_PBC}

\begin{figure}[t]
\includegraphics[width=\columnwidth]{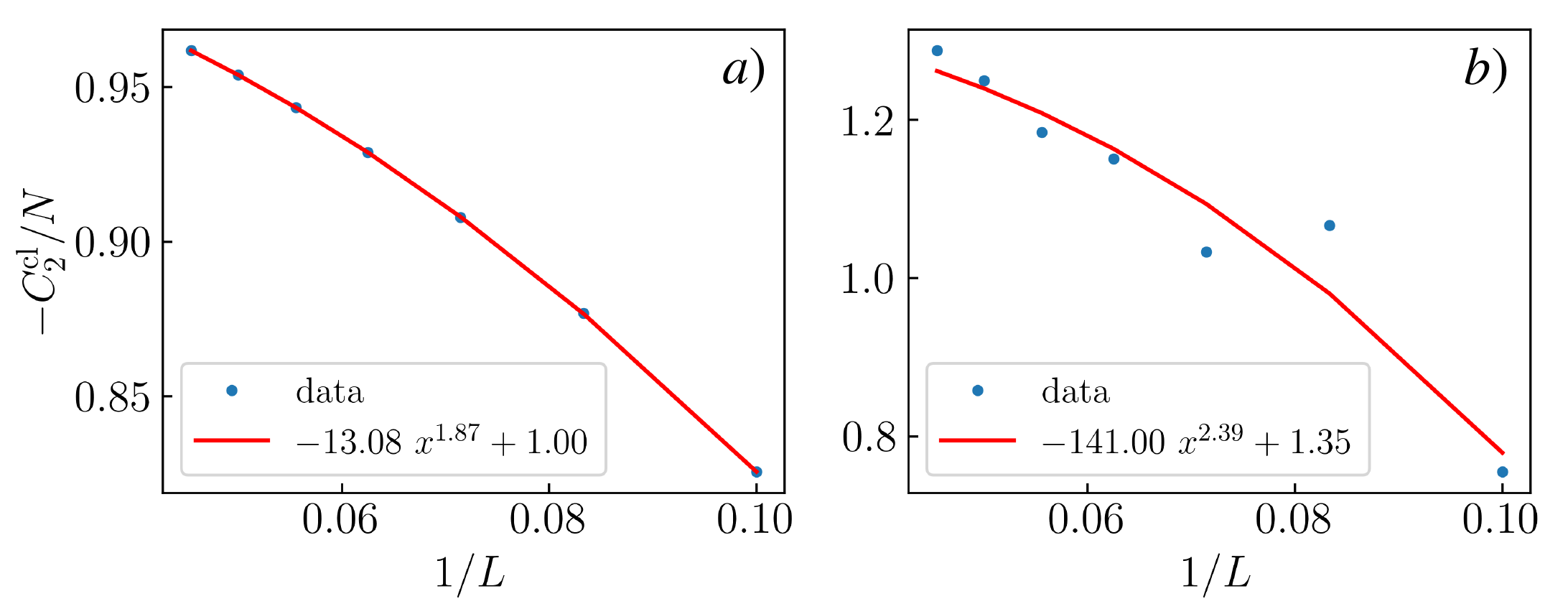}
    \caption{
    Scaling of the computed cluster-specific second Chern number $C_2^{\textrm{cl}}/N$ as a function of the linear mesh size $L$ for (a)~an Abelian and (b)~a non-Abelian cluster with $N=16$ primitive cells. 
    The considerable noise in non-Abelian clusters may be attributed to the change in $C_2$ contribution from different Bloch states that are included at larger $L$ (smaller $1/L=:x$). 
    Nonetheless, the best-fit curves suggest that, for large $L$, $-C_2^{\textrm{cl}}/N$ converges to $1.0$ and $1.35$, respectively.}
    \label{fig:C2fit} 
\end{figure}

To explore the topological response of finite systems to external fields, we focus on small PBC clusters and compute the cluster-specific second Chern number $C_2^{\textrm{cl}}$ by threading fluxes along the four translation directions. Namely, we append a phase to the inter-site transition matrices $T_j$ in \cref{eq:inter-site-mat} as $T_j \to T_j \mathrm{e}^{\mathrm{i}\phi_j}$ with $\boldsymbol{\phi} = \{\phi_j\}_{j=1}^4 \in [0,2\pi)^{4}$. We consider Abelian and non-Abelian clusters with $N=16, 18, 20$ primitive cells. For the non-Abelian clusters, the non-Abelian states constitute 1/2, 2/3, and 4/5 of the total eigenstates, respectively. 

In the flux space, the total second Chern number of a multiplet of $n_f$ filled bands (which may have mutual degeneracies) is~\cite{Qi2008}
\begin{equation}
    C_2 =
    \frac{1}{32 \pi^2} \int_{\mathbb{T}^4} d^4\phi ~ \epsilon^{jklm} ~ \text{Tr}(F_{jk}F_{lm}),
\end{equation}
where 
\begin{equation}
F_{jk} = \frac{\partial A_k}{\partial \phi_j}  - \frac{\partial A_j}{\partial \phi_k} +\mathrm{i}[A_j,A_k]
\end{equation}
is the non-Abelian Berry curvature in the plane spanned by directions $j$ and $k$, and \begin{equation}A_j^{\alpha\beta} = \langle \psi_\alpha(\boldsymbol{\phi}) \vert \frac{\partial}{\partial \phi_j} \vert \psi_\beta(\boldsymbol{\phi}) \rangle\end{equation}
is the non-Abelian Berry connection computed from eigenstates $\ket{\psi_\alpha(\boldsymbol{\phi})}$ with $\alpha=1,...,n_f$. Here $\epsilon^{jklm}$ is the Levi-Civita symbol in 4D with an implicit sum over repeated indices and the integration is over the 4-torus $\mathbb{T}^4$ defined by $\boldsymbol{\phi}$. 
We follow Ref.~\cite{MocholGrzelak2018} for  the numerical computation of $C_2$. 

Let us remark that while for the computation of first Chern number $C_1$ on a 2D manifold there exists an elegant algorithm of Ref.~\cite{Fukui:2005} that guarantees the correct quantized result for a rather coarse choice of a momentum-space grid (namely, the only requirement to find the correct integer is that the Berry phase on each square of the grid is less than $\pi$), the algorithm of Ref.~\cite{MocholGrzelak2018} to compute second Chern number $C_2$ exhibits deviations from the correct results for any finite momentum-space grid. 
Therefore, we study the dependence of the numerically computed $C_2^{\rm cl}$ on the grid size $L$ in \cref{fig:C2fit} to find the result extrapolated at $L\to\infty$. 
Our results show that while the Abelian clusters have $C_2^{\rm cl}/N = C_2$ in agreement with the band theory, positioning of the Bloch states with respect to the nodal manifold influences $C_2^{\rm cl}$ in non-Abelian clusters. Our implementation of the numerical algorithm is available in the Supplementary Data and Code~\cite{Tummuru:2023:SDC}. 


\section{U\texorpdfstring{$(d)$}{d}-HBT} \label{app:sud} 

Here we provide details about the algorithm for sampling BZ$^{(1,d)}$. We focus on the simplest class of non-Abelian Bloch states, i.e., $(d=2)$-dimensional irreducible representations, where the matrices $U_j$, $j=1,\ldots,4$ belong to the $\textrm{SU}(2)$ group. 
In Fig.~\ref{fig:comp}(b), the data for $\textrm{SU}(3)$ are obtained by nearly identical steps, with the only difference being the matrix dimension and parameterization. 

In general, an $\textrm{SU}(2)$ matrix may be parameterized as
\begin{equation}\label{eqn:SU2-decomp}
    U_j = 
    \begin{pmatrix}
        a_j & b_j \\
        -b_j^\ast & a_j^\ast
    \end{pmatrix},
\end{equation}
where $a_j,b_j \in \mathbb{C}$ and $|a_j|^2+|b_j|^2=1$. The latter condition reduces the number of free real parameters to three. 
To be a valid representation of the Fuchsian group $\Gamma$, one has to ensure that the choice of four matrices $\bU$ obey 
\begin{equation}
    U_1^{\phantom{1}} U_2^{-1}U_3^{\phantom{1}} U_4^{-1} U_1^{-1} U_2^{\phantom{1}} U_3^{-1} U_4^{\phantom{1}} - \mathbbm{1}_2 = 0.
    \label{eq:cons2}
\end{equation}

The procedure to find solutions $\{U_j\}_{j=1}^4 \equiv \bU$ is as follows. We pick random $U_{1,2}$ according to the circular unitary ensemble, which represents a uniform distribution over the unitary $d\times d$ matrices, and corresponds to the Haar measure on the unitary group. 
Subsequently, we choose a random initial choice of $U_{3,4}$ according to the same uniform distribution, which we decompose into $a_{3,4}$ and $b_{3,4}$ per \cref{eqn:SU2-decomp}, and we perform a gradient descent with respect to these parameters to minimize the Frobenius norm of \cref{eq:cons2}. 
Sometimes, this procedure gets trapped in a local minimum that does not solve Eq.~(\ref{eq:cons2}); in that case the matrices are discarded and we repeat the steps for a different choice of fixed $U_{1,2}$ and initial $U_{3,4}$.
The algorithm is iterated until we reach a specified number of solutions $\bU$.

To motivate why our described approach for finding random representations $\bU$ works, let us discuss the dimensions of the mathematical spaces at play. 
First, the space of \emph{distinct} $d$-dimensional (unitary) representations $D_\lambda$ of $\Gamma$, i.e., where equivalent representations are treated as a single point, is called the $\textrm{U}(d)$-character variety $X(\Sigma_2,\textrm{U}(d))$; here, $\Sigma_2$ is the genus-$2$ surface obtained from compactifying the edges of the primitive cell of the $\{8,8\}$ lattice. 
This variety is known to be ten-dimensional~\cite{Maciejko2022,nagy2022}. 
One can further decompose the representation matrices on the group generators as $D_\lambda(\gamma_j) = U_j \mathrm{e}^{\mathrm{i}k_j}$, where $\mathrm{e}^{\mathrm{i}k_j}\in\textrm{U}(1)$ and $U_j\in\textrm{SU}(2)$. 
The U(1) factors clearly absorb four of the ten dimensions of $X(\Sigma_2,\textrm{U}(2))$, implying that the space of non-equivalent $\textrm{SU}(2)$ representation matrices that obey Eq.~(\ref{eq:cons2}) is six-dimensional.

Crucially, in our numerical search for random $\textrm{SU}(2)$ matrices obeying Eq.~(\ref{eq:cons2}) we also need to account for distinct choices of representation matrices $\bU$ that fall into the same equivalence class.
Recall that given a unitary matrix $M$, the choice $M \bU M^{\dagger}$ constitutes a representation of $\Gamma$ that is \emph{equivalent} to $\bU$ as it merely corresponds to a unitary rotation of the eigenvectors of the Hamiltonian $H_{\lambda}^{(1,d)}$ that span the representation. 
However, our numerical search is not restricted from potentially finding equivalent representations. 
Since the $\textrm{SU}(2)$ group is three-dimensional, the orbit $\cup_{M\in\mathrm{SU}(2)}M\bU M^\dagger$ of the set $\{U_j\}_{j=1}^4$ is generically three-dimensional; consequently, such unitary 
transformations increase the extent of $\textrm{SU}(2)$ solutions $\bU$ to Eq.~(\ref{eq:cons2}) to a nine-dimensional manifold.
By selecting a randomly chosen $U_{1,2}$, we fix six of the nine parameters; in other words, the space of solutions for $U_{3,4}$ with given $U_{1,2}$ (if such solutions exist) is generically three-dimensional. 
We verified that by performing minimization of the Frobenius norm of Eq.~(\ref{eq:cons2}) for different initial choices of $U_{3,4}$ and fixed $U_{1,2}$ we obtain different optimized solutions for $U_{3,4}$. 
Note that finding solutions for $U_4$ given randomly chosen $U_{1,2,3}$, which corresponds to fixing nine coordinates, generically does \emph{not} yield any result. 
This is because the orbit of a single matrix $\cup_{M\in\mathrm{SU}(2)}M U_j M^\dagger$ can be shown to be only two-dimensional. As a consequence, fixing three random matrices $U_{1,2,3}$ already overdetermines the problem of finding $U_4$.

Before concluding, let us make three final remarks.
First, to make sure that we investigate all \emph{dimensions} of $\textrm{BZ}^{(1,d)}$ uniformly, we compute the Hamiltonian spectrum for $L^6\times L^4 = L^{10}$ unitary representations; here, $L^6$ points correspond to the random choice of $\textrm{SU}(2)$ matrices $\bU$ inside the six-dimensional space of distinct representations, and $L^4$ corresponds to choices of $\textrm{U}(1)$ factors $\{\mathrm{e}^{\mathrm{i}k_j}\}_{j=1}^4$ inside the four-dimensional space of momenta. 
For the presented data, we chose $L=8$.
Second, owing to the lack of explicit parametrizations of these higher BZs, it is presently unknown to us whether they are equipped with a canonical measure of volume. 
Therefore, it is possible that our random sampling is non-homogeneous, i.e., it explores some sectors more densely than others. 
Despite these subtleties, the good agreement of $\textrm{U}(2)$-HBT and 2-supercell DOS lends support to the outlined scheme. 
Finally, when generalizing the algorithm to $\textrm{SU}(3)$ matrices, the decomposition in Eq.~(\ref{eqn:SU2-decomp}) has to be replaced by an eight-parameter decomposition, as shown in Ref.~\cite{Bronzan1988}. A comparison of DOS computed from randomly sampled 2D and 3D representations is shown in \cref{fig:comp}(b).

We note in passing that while $\textrm{U}(d)$-HBT on the primitive cell and $\textrm{U}(1)$-HBT on $n$-supercells provide complementary ways to access the non-Abelian states, an investigation of $\textrm{U}(d)$ Bloch theory on $n$-supercells could conceivably open up a larger share of the hyperbolic reciprocal space.
We did not pursue such a generalization in this work.


\section{Supercell method}\label{app:supercell}

\subsection{Hyperbolic translation groups}
As described in the main text, the supercell method~\cite{Lenggenhager2023} applies $\textrm{U}(1)$-HBT to supercells in order to gain access to the non-Abelian BZs.
By considering sequences of supercells with increasing number $n$ of primitive cells, an increasing number of non-Abelian Bloch states is generated and convergence is achieved for $n\to\infty$.
Supercells can be defined~\cite{Lenggenhager2023} in terms of factor groups $\HPG$ of the underlying space group (more precisely its proper subgroup) with a translation group $\HTG$.
The proper subgroup of the space group of the $\{8,8\}$ lattice is the proper triangle group with presentation
\begin{equation}
    \TG^+(2,8,8) = \gpres{x,y,z}{x^2,y^8,z^8,xyz},
\end{equation}
generated by rotations $x$, $y$, and $z$ around the three vertices of a hyperbolic triangle by angles $\tfrac{2\pi}{2}$, $\tfrac{2\pi}{8}$, and $\tfrac{2\pi}{8}$, respectively, which satisfy the constraints that the products listed after the vertical line are equal to the identity (\cref{fig:2-supercell}).

\begin{figure}[t]
    \centering
    \includegraphics[width=0.7\linewidth]{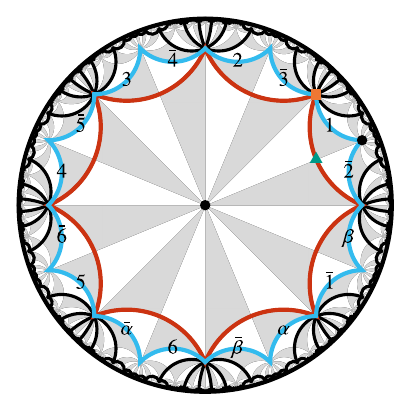}
    \caption{
        Symmetries and unit cells of the $\lee$ lattice.
        The space group symmetries of the lattice are illustrated by grey and white triangles which are generated by applying those symmetry operations to a single one of them.
        Both the symmetric primitive cell (red polygon) as well as the symmetric $2$-supercell (blue polygon), consisting of two copies of the primitive cell, are shown.
        The boundary identification of the supercell is indicated as follows: the edge labeled $\bar{j}$ is related to the edge labeled $j$ by the translation generator $\htg{j}\in\HTG^{(2)}$.
        Edges related by composite translations are labeled by $\bar{\alpha}$, $\alpha=\htg{3}^{-1}\htg{1}^{-1}$, and $\bar{\beta}$, $\beta=\htg{4}^{-1}\htg{2}^{-1}$.
        For the purpose of Appendix~\ref{app:symmetry-analysis}, we indicate the following high-symmetry positions. 
        Black dots mark inversion centers with respect to sites of the $\{8,8\}$ lattice, green triangle with respect to mid-edge, and orange square with respect to faces of the lattice.
    }
    \label{fig:2-supercell}
\end{figure}

To achieve convergence to the thermodynamic limit~\cite{Lux:2022,Mosseri2023}, the sequence of translation groups $\HTG^{(\ell)}\nsubg\Delta^+$ with $\ell\in\{1,2,3,\ldots\}$ needs to obey the normal subgroup relations
\begin{equation}
    \label{eqn:sequence-of-gammas}\HTGuc\gbusn\HTGsc\gbusn\dotsb\gbusn\HTG^{(\ell)}\gbusn\dotsb,
\end{equation}
such that $\bigcap_{\ell\geq 1}\HTG^{(\ell)}$ is the trivial group containing only the identity element~\cite{Lux:2022}.
Note that in the main text we used a slightly different notation, where the superscript `$^{(n)}$' of a translation subgroup indicates its subgroup index \begin{equation}
n=|\HTG:\HTG^{(\ell)}|,
\end{equation} 
whereas here the superscript `$^{(\ell)}$' denotes the position of the subgroup in the coherent sequence (\ref{eqn:sequence-of-gammas}).
This distinction is of little importance in the main text, since we take the second element ($\ell = 2$) in Eq.~(\ref{eqn:sequence-of-gammas}) to be a 2-supercell ($n=2$).

The $n$-supercell is compactified into a surface of genus $\genus{}^{(n)}$ that grows linearly with $n$, as expressed by the Riemann-Hurwitz formula~\cite{Miranda:1995}.
Specifically, if $\genus_{\rm pc}$ is the genus of the compactified primitive cell, then $\genus{}^{(n)} = n(\genus_{\rm pc}-1)+1$.
Here, we adopt the same sequence that was considered in Ref.~\cite{Lenggenhager2023} for the nearest-neighbor hopping model on the $\{8,8\}$ lattice.
This sequence was obtained from the factor groups given in Ref.~\cite{Conder2007} using the \textsc{HyperBloch} package~\cite{HyperBloch} and they are denoted by $\tgquot{2}{6}$, $\tgquot{3}{11}$, $\tgquot{5}{13}$, $\tgquot{9}{20}$, $\tgquot{17}{29}$, $\tgquot{33}{44}$, 
where ``$T\genus{}.j$'' labels the $j^\mathrm{th}$ quotient of any triangle group $\TG^+(r,q,p)$ where the quotient acts on a surface of genus $\genus{}$.
Their explicit presentations are 
\begin{equation}\label{eqn:supercell-PGs}
    \begin{split}
        \HPG^{(1)} &= \gpres{x,y,z}{x^2,y^8,z^8,xyz,xzy,y^3z^{-1}},\\
        \HPG^{(2)} &= \gpres{x,y,z}{x^2,y^8,z^8,xyz,xzy},\\
        \HPG^{(3)} &= \gpres{x,y,z}{x^2,y^8,z^8,xyz,xy^{-2}z^{-1}y,xzy^{-1}z^{-2}},\\
        \HPG^{(4)} &= \gpres{x,y,z}{x^2,y^8,z^8,xyz,xy^{-2}z^{-1}y},\\
        \HPG^{(5)} &= \gpres{x,y,z}{x^2,y^8,z^8,xyz,xzy^{-2}z^{-2}y,(yz^{-1}y^2)^2},\\
        \HPG^{(6)} &= \gpres{x,y,z}{x^2,y^8,z^8,xyz,xzy^{-2}z^{-2}y}.
    \end{split}
\end{equation}
The $2$-supercell derived from $\HPG^{(2)}$ including boundary identifications is illustrated in \cref{fig:2-supercell}.
For each $n$-supercell, we construct the $\textrm{U}(1)$ Bloch Hamiltonian $H^{(n,1)}_{\bk}$, parameterized by the momentum vector $\bk\in\mathrm{BZ}^{(n,1)}\cong\mathsf{T}^{2\genus{}^{(n)}}$.

\begin{widetext}

\subsection{Hyperbolic Brillouin zones}

Because of the subgroup relations in \cref{eqn:sequence-of-gammas}, one can always express supercell translation generators in terms of  $\{\gamma_{1},\gamma_{2},\gamma_{3},\gamma_{4}\}$ and their inverses.
For the $2$-supercell, specifically, we find
\begin{alignat}{2}
& \tilde{\gamma}_1 = \gamma^{}_{4} \gamma^{}_{1}, ~ 
&& \tilde{\gamma}_2 = \gamma^{-1}_{1} \gamma^{}_{2}, \nonumber \\
& \tilde{\gamma}_3 = \gamma^{-1}_{2} \gamma^{}_{3}, ~
&& \tilde{\gamma}_4 = \gamma^{-1}_{3} \gamma^{}_{4}, \nonumber \\
& \tilde{\gamma}_5 = \gamma^{-1}_{4} \gamma^{-1}_{1}, ~
&& \tilde{\gamma}_6 = \gamma^{}_{1} \gamma^{-1}_{2}.
\end{alignat}
These relations among the translation generators imply the following immersion of the original 4D BZ$^{(1,1)}$ inside the 6D BZ$^{(2,1)}$:
    \begin{eqnarray}
    \iota:  \qquad \quad \textrm{  BZ}^{(1,1)} &\to& \textrm{BZ}^{(2,1)} \nonumber \\
    (k_1,k_2,k_3,k_4) &\mapsto& (k_4+k_1,k_2-k_1,k_3-k_2,k_4-k_3,-k_4-k_1,k_1-k_2) \label{eqn:immerse-4D-in-6D}
    \end{eqnarray}
Note that the immersion $\iota$ is not one-to-one (i.e., it is non-injective): two four-momenta that differ by $\pi$ in all components $\{k_j\}_{j=1}^4$ are mapped onto the same six-momentum in BZ$^{(2,1)}$. This two-to-one mapping (associated with doubling of the number of bands) is analogous to the Brillouin-zone folding that accompanies supercell constructions in Euclidean lattices. 
We also remark that the immersion $\iota\big(\textrm{BZ}^{(1,1)}\big)$ passes through the point $(0,0,0,0,0,0)$ of BZ$^{(2,1)}$, which corresponds to the center of the nodal-line ring derived in Sec.~\ref{sec:k.p}. For this reason, the second Chern number computed inside BZ$^{(1,1)}$ relates to the second Chern number associated with the nodal ring inside BZ$^{(2,1)}$.

By randomly sampling $\mathrm{BZ}^{(n,1)}$ with $10^9$ $\bk$-points and diagonalizing $H^{(n,1)}_{\bk}$ for each of them, we compute the energy spectrum.
The DOS is then obtained with an energy resolution of $0.005$ [\cref{fig:comp}(c)].
The appearance of peaks in the DOS may be understood as a consequence of ``band folding'', which increases the number of energy bands. 
We observe convergence with increasing supercell size $n$.
Crucially, the DOS near zero energy converges rapidly, such that already the $2$-supercell ($\tgquot{3}{11}$) gives a good approximation.
Therefore, we next study the corresponding Bloch Hamiltonian $H^{(2,1)}_{\bk}$ on the 6D $\mathrm{BZ}^{(2,1)}$.

Before narrowing our attention to the $2$-supercell Bloch Hamiltonian in Sec.~\ref{sec:k.p}, let us briefly comment on the existence of a general subset relation \begin{equation}
    \textrm{BZ}^{(n_\ell,d)}\subset \textrm{BZ}^{(n_{\ell+1},d)}
\end{equation}
which generalizes the particular immersion in Eq.~(\ref{eqn:immerse-4D-in-6D}). 
This sought relation follows directly from the subgroup relation $\Gamma^{(\ell)}<\Gamma^{(\ell+1)}$ between the translation groups. 
Namely, if $\lambda$ is a $d$-dimensional representation of $\Gamma^{(\ell)}$ that for $g\in\Gamma^{(\ell)}$ assigns $g\mapsto \lambda(g)$, then a restriction of this mapping to $g\in \Gamma^{(\ell+1)}$ gives a \emph{subduced representation} of $\Gamma^{(\ell+1)}$~\cite{Lenggenhager2023}. 
This subduction defines the immersion (and therefore a subset relation) among the spaces $\textrm{BZ}^{(n_\ell,d)}$ and $\textrm{BZ}^{(n_{\ell+1},d)}$ of $d$-dimensional representations. 
Note that this immersion is not injective, because different representions $\lambda_{1} \neq \lambda_{2}$ of $\Gamma^{(\ell)}$ may subduce the same representation $\lambda$ when restricted to elements of the subgroup $\Gamma^{(\ell+1)}$. For example, we already showed that the immersion in Eq.~(\ref{eqn:immerse-4D-in-6D}) is two-to-one.


\subsection{$2$-supercell approximation}

\subsubsection{Abelian Bloch Hamiltonian for $2$-supercells}
\label{sec:k.p}

The Bloch Hamiltonian corresponding to the 2-supercell can be expressed as a linear combination of $8 \times 8$ matrices as
\begin{equation}
    H^{\rm (2,1)}_{\bk} = \frac{1}{4} \sum_{ij\ell} d_{ij\ell} \sigma_{ij\ell},
    \label{eq:bloch_h}
\end{equation}
where $\sigma_{ij\ell} = \sigma_i \otimes \sigma_j \otimes \sigma_\ell$, with $\{\sigma_i\}_{i=0}^3$ denoting the Pauli matrices and the $2\times 2$ identity matrix. The non-zero summands in Eq.~(\ref{eq:bloch_h}) are
\begin{subequations}
\label{eq:bloch_h-detail}
\begin{align}
d_{0, 3, 3} &= 4 m, \\ d_{1, 1, 0} &= -\mathrm{i} \mathrm{e}^{-\mathrm{i} \left(k_2 + k_3 + k_4 + k_5\right)} \left(-1 + \mathrm{e}^{2\mathrm{i} \left(k_2 + k_3 + k_4 + k_5\right)}\right), \\ d_{1, 2, 0} &= -\mathrm{i} \mathrm{e}^{-\mathrm{i} \left(k_2 + k_3 + k_4 + k_5 + k_6\right)} \left(-1 + \mathrm{e}^{\mathrm{i} \left(k_3 + k_4 + k_5 + k_6\right)}\right) \left(1 + \mathrm{e}^{\mathrm{i} \left(2 k_2 + k_3 + k_4 + k_5 + k_6\right)}\right), \\ d_{1, 3, 1} &= \mathrm{i} \mathrm{e}^{-\mathrm{i} \left(k_1 + k_2 + k_3 + k_4 + k_6\right)} \left(\mathrm{e}^{\mathrm{i} \left(k_1 + k_3\right)} - \mathrm{e}^{\mathrm{i} \left(k_4 + k_6\right)}\right) \left(\mathrm{e}^{\mathrm{i} k_1} + \mathrm{e}^{\mathrm{i} \left(2 k_2 + k_3 + k_4 + k_6\right)}\right), \\ d_{1, 3, 2} &= \mathrm{i} \mathrm{e}^{-\mathrm{i} \left(k_1 + k_2 + k_3 + k_4\right)} \left(\mathrm{e}^{\mathrm{i} k_1} + \mathrm{e}^{\mathrm{i} \left(k_2 + k_3 + k_4\right)}\right) \left(-1 + \mathrm{e}^{\mathrm{i} \left(k_1 + k_2 + k_3 + k_4\right)}\right), \\
d_{1, 3, 3} &= \mathrm{e}^{-\mathrm{i} \left(k_1 + k_2 + k_3 + k_4 + k_5 + k_6\right)} \Big(\mathrm{e}^{\mathrm{i} k_1} + \mathrm{e}^{\mathrm{i} \left(2 k_1 + k_3 + k_5\right)} + \mathrm{e}^{\mathrm{i} \left(k_1 + k_6\right)} + \mathrm{e}^{\mathrm{i} \left(k_1 + k_5 + k_6\right)} + \mathrm{e}^{\mathrm{i} \left(k_1 + k_4 + k_5 + k_6\right)} \nonumber \\ & + \mathrm{e}^{\mathrm{i} \left(k_1 + k_3 + k_4 + k_5 + k_6\right)} + \mathrm{e}^{\mathrm{i} \left(k_2 + k_3 + k_4 + k_5 + k_6\right)} + 2 \mathrm{e}^{\mathrm{i} \left(k_1 + k_2 + k_3 + k_4 + k_5 + k_6\right)} + \mathrm{e}^{\mathrm{i} \left(2 k_1 + k_2 + k_3 + k_4 + k_5 + k_6\right)} \nonumber \\ & + \mathrm{e}^{\mathrm{i} \left(k_1 + 2 k_2 + k_3 + k_4 + k_5 + k_6\right)} + \mathrm{e}^{\mathrm{i} \left(k_1 + 2 k_2 + 2 k_3 + k_4 + k_5 + k_6\right)} + \mathrm{e}^{\mathrm{i} \left(k_1 + 2 k_2 + 2 k_3 + 2 k_4 + k_5 + k_6\right)} \nonumber \\ & + \mathrm{e}^{\mathrm{i} \left(k_1 + 2 k_2 + 2 k_3 + 2 k_4 + 2 k_5 + k_6\right)} + \mathrm{e}^{\mathrm{i} \left(2 k_2 + k_3 + 2 k_4 + k_5 + 2 k_6\right)} + \mathrm{e}^{\mathrm{i} \left(k_1 + 2 \left(k_2 + k_3 + k_4 + k_5 + k_6\right)\right)}\Big), \\
d_{2, 1, 0} &= -\mathrm{e}^{-\mathrm{i} \left(k_2 + k_3 + k_4 + k_5\right)} \left(-1 + \mathrm{e}^{\mathrm{i} \left(k_2 + k_3 + k_4 + k_5\right)}\right){} ^ 2, \\ d_{2, 2, 0} &= -\mathrm{e}^{-\mathrm{i} \left(k_2 + k_3 + k_4 + k_5 + k_6\right)} \left(-1 + \mathrm{e}^{\mathrm{i} \left(k_3 + k_4 + k_5 + k_6\right)}\right) \left(-1 + \mathrm{e}^{\mathrm{i} \left(2 k_2 + k_3 + k_4 + k_5 + k_6\right)}\right), \\ d_{2, 3, 1} &= \mathrm{e}^{-\mathrm{i} \left(k_1 + k_2 + k_3 + k_4 + k_6\right)} \left(\mathrm{e}^{\mathrm{i} \left(k_1 + k_3\right)} - \mathrm{e}^{\mathrm{i} \left(k_4 + k_6\right)}\right) \left(\mathrm{e}^{\mathrm{i} \left(2 k_2 + k_3 + k_4 + k_6\right)} - \mathrm{e}^{\mathrm{i} k_1}\right), \\ d_{2, 3, 2} &= \mathrm{e}^{-\mathrm{i} \left(k_1 + k_2 + k_3 + k_4\right)} \left(\mathrm{e}^{\mathrm{i} \left(k_2 + k_3 + k_4\right)} - \mathrm{e}^{\mathrm{i} k_1}\right) \left(-1 + \mathrm{e}^{\mathrm{i} \left(k_1 + k_2 + k_3 + k_4\right)}\right), \\
d_{2, 3, 3} &= \mathrm{i} \mathrm{e}^{-\mathrm{i} \left(k_1 + k_2 + k_3 + k_4 + k_5 + k_6\right)} \Big(\mathrm{e}^{\mathrm{i} k_1} + \mathrm{e}^{\mathrm{i} \left(2 k_1 + k_3 + k_5\right)} + \mathrm{e}^{\mathrm{i} \left(k_1 + k_6\right)} + \mathrm{e}^{\mathrm{i} \left(k_1 + k_5 + k_6\right)} + \mathrm{e}^{\mathrm{i} \left(k_1 + k_4 + k_5 + k_6\right)}\nonumber \\ & +\mathrm{e}^{\mathrm{i} \left(k_1 + k_3 + k_4 + k_5 + k_6\right)} - \mathrm{e}^{\mathrm{i} \left(k_2 + k_3 + k_4 + k_5 + k_6\right)} + \mathrm{e}^{\mathrm{i} \left(2 k_1 + k_2 + k_3 + k_4 + k_5 + k_6\right)} - \mathrm{e}^{\mathrm{i} \left(k_1 + 2 k_2 + k_3 + k_4 + k_5 + k_6\right)} \nonumber \\ & - \mathrm{e}^{\mathrm{i} \left(k_1 + 2 k_2 + 2 k_3 + k_4 + k_5 + k_6\right)} - \mathrm{e}^{\mathrm{i} \left(k_1 + 2 k_2 + 2 k_3 + 2 k_4 + k_5 + k_6\right)} - \mathrm{e}^{\mathrm{i} \left(k_1 + 2 k_2 + 2 k_3 + 2 k_4 + 2 k_5 + k_6\right)} \nonumber \\ & - \mathrm{e}^{\mathrm{i} \left(2 k_2 + k_3 + 2 k_4 + k_5 + 2 k_6\right)} - \mathrm{e}^{\mathrm{i} \left(k_1 + 2 \left(k_2 + k_3 + k_4 + k_5 + k_6\right)\right)}\Big).
\end{align}
\end{subequations}
\end{widetext}
As we shall see in the next section, symmetries constrain some of the nodes to lie in the $(k_a,k_b,-k_a,-k_b,k_a,k_b)$ plane for $k_{a,b}\in[-\pi,+\pi)$. 
To obtain the low-energy spectrum in this plane and to simplify the expressions, we proceed in three~steps:

\begin{figure*}[t]
\includegraphics[width=\textwidth]{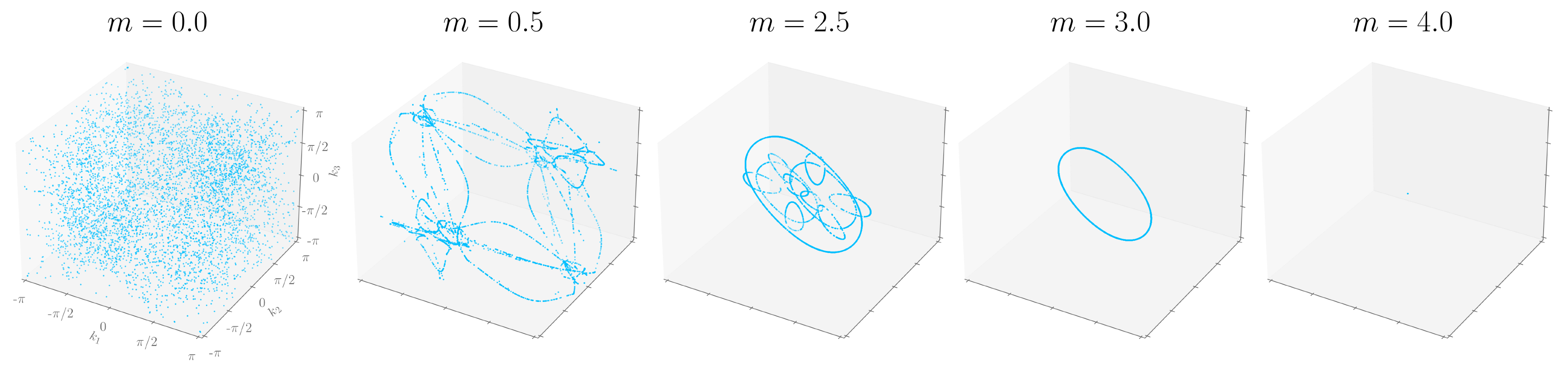}
    \caption{
    Evolution of the nodal manifold within the 2-supercell description as a function of mass $m$. 
    Each point denotes a projected six-momentum coordinate where the spectral gap at half-filling vanishes. 
    The nodal ring emerges close to $m=3$, shrinks to a point at the origin when $m=4$, and finally vanishes as the system enters the trivial insulating phase.
    The behavior is mirrored for negative values of the mass.
    }
    \label{fig:evol} 
\end{figure*}

1.~A basis transformation on the Bloch Hamiltonian makes the structure of the ring manifest in the small momentum expansion: $\tilde{H}_{\bk}^{(2,1)} = V_{H} H_{\bk}^{(2,1)} V_{H}^{\dagger}$, with $V_{H} = {\rm diag}(1,1,1,1,r,r,r,r)$ and $r=\mathrm{e}^{\mathrm{i} (k_1-k_2)/2}$.

2.~For a direct handle on the nodal plane, we perform a basis rotation to obtain the Hamiltonian $\tilde{H}^{\rm (2,1)}_{\bq}$ where $\bq = V_{\bk} \bk$ and 
\begin{equation}
V_{\bk} = \frac{1}{\sqrt{6}}
\left(
\begin{array}{cccccc}
 \sqrt{2} & 0 & \sqrt{3} & 0 & -1 & 0 \\
 0 & \sqrt{2} & 0 & \sqrt{3} & 0 & -1 \\
 -\sqrt{2} & 0 & \sqrt{3} & 0 & 1 & 0 \\
 0 & -\sqrt{2} & 0 & \sqrt{3} & 0 & 1 \\
 \sqrt{2} & 0 & 0 & 0 & 2 & 0 \\
 0 & \sqrt{2} & 0 & 0 & 0 & 2 \\
\end{array}
\right).
\end{equation}
Essentially, $V_{\bk}$ provides an orthonormal basis where the first two columns span the nodal plane and the ring now lies in $(\bq_{1},\bq_{2},0,0,0,0)$. At this stage, to leading order in momenta $\bq$, the Hamiltonian near
the origin is given by $\tilde{H}^{\rm (2,1)}_{\bq} = \sum_{ij\ell} d_{ij\ell} \sigma_{ij\ell}$ and
\begin{subequations}
\label{eq:kp_exp}
\begin{align}
    d_{0,3,3}&=m, \\
    d_{1,3,3}&=-\tfrac{{q}_1^2+{q}_2^2}{6}+4, \\
    d_{1,1,0}&=\tfrac{{q}_3+2 {q}_4+\sqrt{3} {q}_5}{2 \sqrt{2}}, \\
    d_{1,2,0}&=\tfrac{{q}_3+{q}_4+\sqrt{3} \left({q}_5+{q}_6\right)}{2 \sqrt{2}}, \\
    d_{1,3,1}&=\tfrac{-2 {q}_3+{q}_4+\sqrt{3} {q}_6}{2 \sqrt{2}}, \\
    d_{1,3,2}&=-\tfrac{{q}_3+ {q}_4}{\sqrt{2}}, \\
    d_{2,3,3}&=\tfrac{3 ({q}_3+ {q}_4) +\sqrt{3} \left({q}_5+{q}_6\right)}{\sqrt{2}}.
\end{align}    
\end{subequations}
Note that the nodal ring lies in the plane spanned by $(q_1,q_2)$ and its form is apparent in $d_{1,3,3}$. The terms \cref{eq:kp_exp}(b--f) form a set of mutually anti-commuting matrices. The eigenvalues in this low-energy approximation are given by 
\begin{align} \label{eq:lowE}
    E^2 & =  \frac{1}{36} \left(q_1^2+q_2^2-24\right)^2+m^2+X^2 + d_{2,3,3}^2 \\ & \pm 2 \sqrt{m^2 \left(  d_{2,3,3}^2+\left(q_1^2+q_2^2-24\right)^2/36\right)+ d_{2,3,3}^2 X^2} \nonumber,
\end{align}
wherein $X^2 = d_{1,1,0}^2 + d_{1,2,0}^2 + d_{1,3,1}^2 + d_{1,3,2}^2$ and each of the four energy bands is doubly degenerate. 

3.~To proceed, it would be convenient to diagonalize the expression $X$, i.e., identify an orthonormal basis for the coordinates $(q_3,q_4,q_5,q_6)$, while leaving $(q_1,q_2)$ unaltered. By obtaining the eigenvectors of the matrix representation of $X$, it can be shown that with
\begin{equation}
    \begin{pmatrix}
        q_3 \\ q_4 \\ q_5 \\ q_6  
    \end{pmatrix}
    \to \frac{1}{4}
    \left(
    \begin{array}{cccc}
     2 \alpha _+ & 2 \beta _+ & 2 \alpha _-
       & 2 \beta _- \\
     -\kappa _+ & \kappa _- & \kappa _- &
       \kappa _+ \\
     \kappa _- & -\kappa _+ & \kappa _+ &
       \kappa _- \\
     2 \alpha _- & -2 \beta _- & -2 \alpha
       _+ & 2 \beta _+ \\
    \end{array}
    \right)
    \begin{pmatrix}
        q_3 \\ q_4 \\ q_5 \\ q_6  
    \end{pmatrix},
\end{equation}
such that $\alpha_{\pm} = \sqrt{1\pm \tfrac{1}{2 \sqrt{19}}}$, $\beta_{\pm} = \sqrt{1\pm \tfrac{7}{2\sqrt{19}}}$ and $\kappa_{\pm} =\sqrt{3}\pm 1$, one has 
\begin{align}
    4X^2= & \left(5+\sqrt{19}\right) q_3^2+\left(3+\sqrt{3}\right)
   q_4^2 \nonumber \\ & +\left(3-\sqrt{3}\right) q_5^2+\left(5-\sqrt{19}\right) q_6^2.
\end{align}

From \cref{eq:lowE}, the energy vanishes when 
\begin{align}
    4 X^2 = -\left(\left(q_1^2 + q_2^2 \pm 24\right)/6 \pm \sqrt{m^2 - d_{2,3,3}^2} \right)^2,
\end{align}
which in turn implies that $X=0$, $q_3=q_4=q_5=q_6=0$ and $d_{2,3,3}=0$. It is then easy to deduce that the nodal line is described by $q_1^2+q_2^2 = 6(4-|m|)$. 


\subsubsection{Mass dependence of the nodal manifold}
\label{app:diff_m}

The nodal manifold of the 2-supercell in its 6D BZ can be further studied as a function of the mass $m$. As seen in \cref{fig:evol}, the manifold has a complicated structure when projected to the first three momentum coordinates. The simple structure of a nodal ring emerges only close to $m=3$. 
Referring to the case of Weyl semimetals, the degeneracy in the spectrum can only be removed either by breaking translation symmetry and allowing the Weyl points to hybridize or by tuning model parameters to bring them together.
Similarly, here the nodal ring shrinks to a point as one tunes the mass towards $m=4$, which marks a phase transition to the trivial atomic limit. 

The topology characterizing the nodal line can be visualized by integrating the second Chern number for fixed $(k_1,k_2)$ (the horizontal coordinates of plots in Fig.~\ref{fig:evol}) along the four-dimensional torus spanned by the remaining momenta $(k_3,k_4,k_5,k_6)$. 
For $m=3$ the computed $C_2$ data exhibits a clear convergence; therefore, we carry the integration for a single grid size $L=15$ and round the result to the nearest integer. 
The outcome of this computation is plotted in the horizontal plane of Fig.~\ref{fig3}(b) of the main text.
In addition, we attempted to perform an analogous computation for $m=2.5$, where the data in Fig.~\ref{fig:evol} suggest eleven distinct nodal-line rings.
The occurrence of multiple nodal lines implies that the typical gap size on the four-dimensional tori is smaller, making the Berry curvature peaks sharper, thus necessitating a finer momentum-space grid. 
We show in \cref{fig:m2.5_phase}(a) the result of the integration for the grid size $L=15$, where large deviations from quantized values are observed.
Therefore, following the discussion in Sec.~\ref{app:flux_PBC}, we perform the computation for a range of grid sizes $L\in\{5,6,\ldots,12,13,15\}$ and then perform a finite-size scaling with the fitting function \begin{equation}
\label{eqn:C2-extrapol}
C_2 = C_{2,\infty} + \frac{a}{L}+\frac{b}{L^2}.
\end{equation}
Unfortunately, the considered range of grid sizes $L$ turns out to be insufficient for extracting well converged results.
Namely, in \cref{fig:m2.5_phase}(b) we plot the extract values of $C_{2,\infty}$ where significant noise is observed. 
Due to computational limitations we did not perform the analysis with larger values of $L$ and for smaller~$\left|m\right|$.

\begin{figure}[t]
\includegraphics[width=\columnwidth]{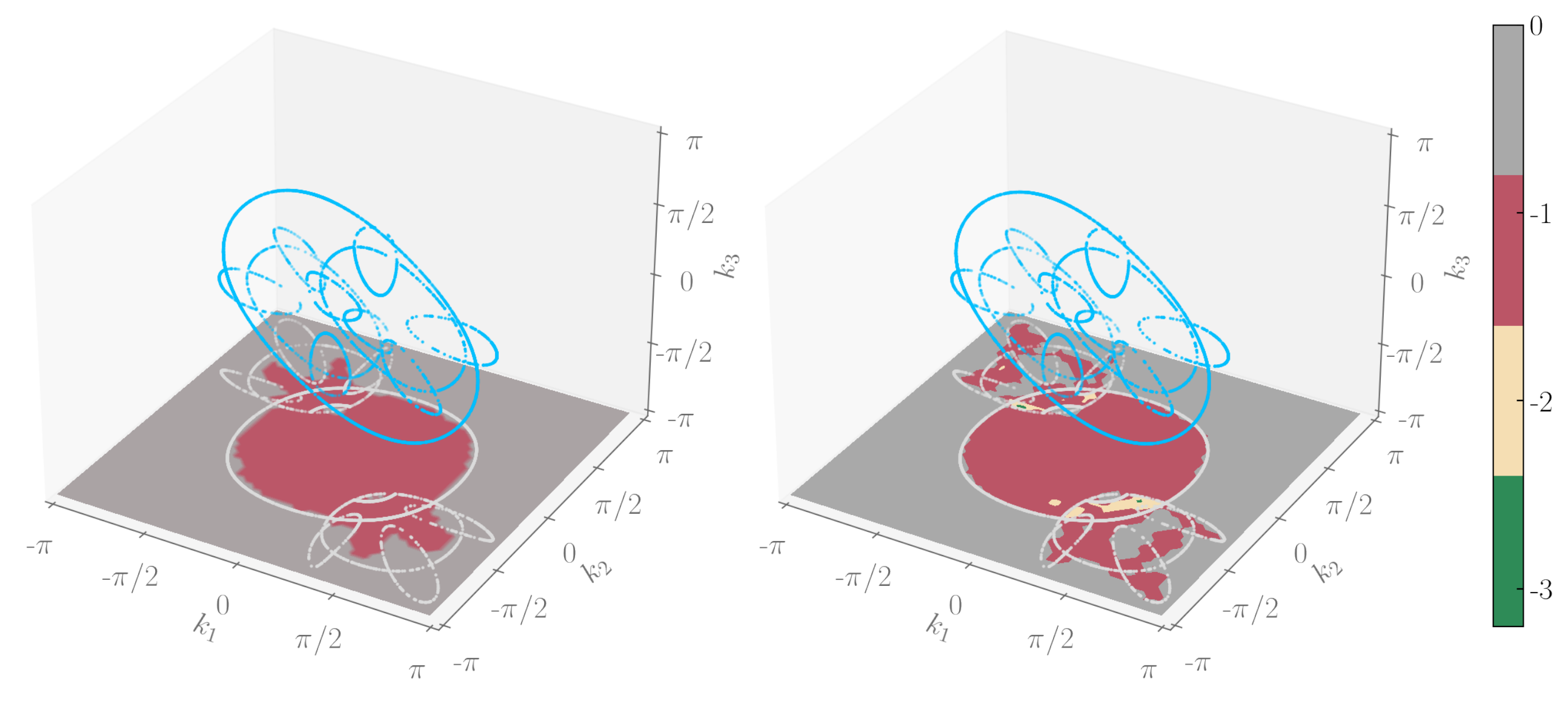}
    \caption{Second Chern number $C_2$ of the $2$-supercell Hamiltonian at mass parameter $m=2.5$, computed by integration inside the four-dimensional tori with fixed momenta $(k_1,k_2)$ for (a)~$L=15$ and (b)~when extrapolated to large $L$ by rounding $C_{2,\infty}$ in the fitted function~(\ref{eqn:C2-extrapol}) to the nearest integer. The data is noisy due to computational limitations.}
    \label{fig:m2.5_phase} 
\end{figure}

We next extract the scaling of the DOS as a function of the energy, with the result displayed in \cref{fig:dos_scaling}.
Assuming a power-law scaling of the DOS close to $E=0$, i.e. $\rho(E\to 0)\propto E^\alpha$, we fit the data using a non-linear least-squares algorithm.
To account for the non-normally distributed errors in the data obtained from random sampling, we perform a weighted least-squares fit with weights given by the inverse square-root of the \emph{observed} DOS values (see \cref{sec:larger-sc} for more details on the fitting procedure).
We find good agreement with the model $\rho(E)\propto E^\alpha$ when fitting in an energy range of width $E_\textrm{max}-E_\textrm{min}=0.1$, starting from the lowest energy $E_\textrm{min}>0$ for which the DOS data does not exhibit any holes (i.e., energy bins with no sampled states).
The extracted scaling exponent $\alpha$ varies with $0<m\leq 4$ within the range~$(3.3,4.1)$. 
For $0<m<3$, a deviation to values somewhat smaller than the theoretically predicted $\alpha=4$ (see \cref{app:symmetry-analysis}) is observed, while it remains closer to $4$ in the range $3\leq m\leq 4$.
Exactly at $m=0$, however, the 2-supercell DOS has a linear scaling, $\alpha=0.976\pm 0.002$. 
On the other hand, for $m>4$, a gap opens in the $2$-supercell spectrum and $\alpha$, as extracted with our fitting procedure, suddenly jumps from roughly $3.8$ at $m=4$ to $30\pm 1$ at $m=4.25$.


\subsubsection{Symmetry protection of the nodal ring}
\label{app:symmetry-analysis}

In this section, we analyze how certain crystalline symmetries of the hyperbolic $\{8,8\}$ lattice impose particular band-structure features observed inside the 6D BZ$^{(2,1)}$ of the $2$-supercell Hamiltonian discussed in Appendix~\ref{sec:k.p}. 
Special attention is given to space inversion ($\mathcal{P}$) and to time reversal ($\mathcal{T}$).
In particular, these symmetries allow us to address the three following aspects: 
(i) identify the 2D plane of the nodal ring, which is enforced by inversion of bands with opposite $\mathcal{P}$-inversion eigenvalues, 
(ii) clarify the two-fold Kramers degeneracy of all bands, 
and (iii) explain the codimension $\mathfrak{d}=5$ for node formation at generic values of $m$ as well as the significantly decreased codimension $\mathfrak{d}=2$ observed for $m=0$.

Let us begin by discussing the role of $\mathcal{PT}$ symmetry in the 4D BZ$^{(1,1)}$ of the hyperbolic Bloch Hamiltonian $H_{\boldsymbol{k}}^{(1,1)}$ constructed on the primitive unit cell, which matches exactly the Bloch Hamiltonian of the 4D QHI on a hypercubic lattice~\cite{zhang2001,Qi2008}.
The topology of the latter Hamiltonian is well understood~\cite{Ryu:2010}.
In particular, it is easily verified that $H_{\boldsymbol{k}}^{(1,1)}$ in \cref{eq:4D-QHI:Bloch-Ham}) commutes with $\mathcal{PT} = \GM{2} \GM{4} \mathcal{K}$:
\begin{equation}
(\mathcal{PT})H_{\boldsymbol{k}}^{(1,1)}(\mathcal{PT})^{-1} = H_{\boldsymbol{k}}^{(1,1)}
\end{equation}
where $\GM{2}$ and $\GM{4}$ are the two imaginary Dirac matrices. 
Since $(\mathcal{PT})^2 = -\mathbbm{1}_4$, it follows that matrices $H_{\boldsymbol{k}}^{(1,1)}$ belong to the symplectic class~\cite{vonNeumann:1929}, which corresponds to nodal class $\mathrm{AII}$ of Ref.~\cite{Bzdusek:2017}. 
Energy bands in this symmetry class exhibit two-fold Kramers degeneracies in the spectrum, while robust band degeneracies occurring at generic $\boldsymbol{k}$-points are of codimension $\mathfrak{d}=5$~\cite{vonNeumann:1929,Bzdusek:2017} and therefore do not occur inside the 4D BZ$^{(1,1)}$. 

On the other hand, robust point nodes may generically occur inside the five-dimensional space spanned by $(\boldsymbol{k},m)$. 
Indeed, as manifested by the phase diagram in \cref{fig1}(b), topological phase transitions at $m\in\{0,\pm 2,\pm4\}$ are facilitated through gap closings, which are pinned by inversion symmetry $\mathcal{P} = \GM{5}$ to high-symmetry momenta with components $\{k_j\}_{j=1}^4\in\{0,\pi\}$:
\begin{equation}
\mathcal{P}H_{\boldsymbol{k}}^{(1,1)}\mathcal{P}^{-1} = H_{-\boldsymbol{k}}^{(1,1)}.\label{eqn:INV-4D}
\end{equation}
Since Dirac matrices $\GM{5}$ have eigenvalues $\pm 1$, each with multiplicity two, it follows from the $m_{\boldsymbol{k}}\GM{5}$ term in the Hamiltonian~(\ref{eq:4D-QHI:Bloch-Ham}) that the consecutive topological phase transitions are driven by inversion of bands of opposite $\mathcal{P}$-eigenvalues at the various high-symmetry $\boldsymbol{k}$-points.

Let us remark that when representing the action of $\mathcal{P}$ on the hyperbolic lattice in the coordinate space, we encounter certain non-uniqueness; namely, there are several distinct points of inversion center that are not related by symmetry of the $\{8,8\}$ lattice.
These inversion points respectively correspond to: vertices/sites (black dots in Fig.~\ref{fig:2-supercell}), mid-points of edges (green triangle in Fig.~\ref{fig:2-supercell}), and centers of faces (orange squares in Fig.~\ref{fig:2-supercell}) of the $\{8,8\}$ lattice. 
Nevertheless, the inversion operations with respect to these distinct points belong to the same conjugacy class in the hyperbolic point group $G^{(1)}:=\Delta^+/\Gamma^{(1)}$, where $\Delta^+ = \Delta^+(2,8,8)$ 
is the triangle group of the $\{8,8\}$ lattice~\cite{Boettcher:2022,Lenggenhager2023} and $\Gamma^{(1)}$ is the translation group associated with the primitive unit cell. 
It follows that all the specified inversion centers transform the hyperbolic momentum in the same way, $\mathcal{P}:\boldsymbol{k}\mapsto -\boldsymbol{k}$, and each can, therefore, be interpreted as the coordinate-space equivalent of the inversion symmetry in Eq.~(\ref{eqn:INV-4D}).

\begin{figure}[t]
\includegraphics[width=\columnwidth]{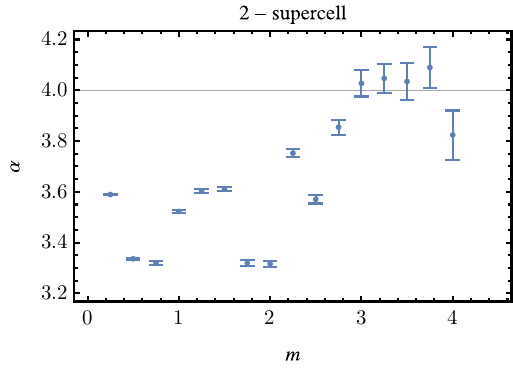}
    \caption{
    Low-energy scaling exponent $\alpha$ as function of the mass parameter $m$ computed from fitting the $2$-supercell density of states with the model $\rho(E\to 0)\propto E^\alpha$. The error bars indicate the standard errors for the estimate of $\alpha$ (see \cref{sec:larger-sc} for details of the fitting procedure).
    The horizontal gray line indicates the scaling of $\alpha=4$ predicted by the theoretical argument discussed in \cref{app:symmetry-analysis}). Data at $m=0$ ($\alpha = 0.976\pm 0.002$) and $m>4$ ($\alpha = 30 \pm 1$) are outside the range shown here, see \cref{fig:dos_scaling_scs}.
    }
    \label{fig:dos_scaling} 
\end{figure}

We next analyze how the notions of $\mathcal{P}$ and $\mathcal{T}$ symmetry, and of their composition $\mathcal{PT}$, generalize to the 6D BZ$^{(2,1)}$ associated with the hyperbolic Bloch Hamiltonian $H_{\boldsymbol{k}}^{(2,1)}$ on the $2$-supercell, specified in Eqs.~(\ref{eq:bloch_h}--\ref{eq:bloch_h-detail}).
Let us remind that a hyperbolic crystalline symmetry $g$ transforms momentum components inside the 6D BZ$^{(2,1)}$ by a linear transformation $\boldsymbol{k}\mapsto \boldsymbol{k'}= M_g \cdot \boldsymbol{k}$, where we call $M_g\in\mathsf{GL}(6,\mathbb{Z})$ the point-group matrix of $g$~\cite{Chen2023b}. 
If $g$ is a symmetry of the studied model, then the hyperbolic Bloch Hamiltonians $H_{\boldsymbol{k}}^{(2,1)}$ and $H_{\boldsymbol{k'}}^{(2,1)}$ are related by a linear (unitary or antiunitary) transformation, and consequently the energy spectra at $\boldsymbol{k}$ and $\boldsymbol{k'}$ are identical.

We first analyze the effect of time-reversal symmetry.
Since $\mathcal{T}$ does not act on spatial coordinates, it does not induce a non-trivial permutation of the boundaries of the $2$-supercell.
However, being an antiunitary operator, $\mathcal{T}$ acts by complex conjugation on the twisted boundary conditions across the $2$-supercell edges. 
The conjugation results in flipping the sign of all momentum components, implying $M_\mathcal{T}=-\mathbbm{1}_6$, and there exists a unitary matrix $U_{\mathcal{T}}$ such that
\begin{equation}
H^{(2,1)}_{M_\mathcal{T}\cdot\boldsymbol{k}} = U_{\mathcal{T}} \cdot \big[H^{(2,1)}_{\boldsymbol{k}}\big]^* \cdot U_{\mathcal{T}}^\dagger 
\end{equation}
with the specific form of $U_{\mathcal{T}} = \sigma_0 \otimes (\GM{1}\GM{3}) \mathcal{K}$.

In contrast to the primitive cell, the discussion of inversion symmetry for the $2$-supercell branches into two cases; this is because the various choices of the inversion center now fall into \emph{two} conjugacy classes in the (enlarged) point group $G^{(2)}:=\Delta/\Gamma^{(2)}$, where $\Gamma^{(2)}$ is the (reduced) translation group of the $2$-supercell.
Specifically, the inversion $\mathcal{P}^\textrm{V}$ with respect to vertices/sites (black dots in Fig.~\ref{fig:2-supercell}) and inversion $\mathcal{P}^\textrm{F}$ with respect to centers of faces (orange square in Fig.~\ref{fig:2-supercell}) of the $\{8,8\}$ lattice are \emph{distinguished} from the inversion $\mathcal{P}^\textrm{E}$ with respect to the mid-points of edges (green triangle in Fig.~\ref{fig:2-supercell}) of the $\{8,8\}$ lattice. 
For this reason, to understand the full implications of inversion symmetry on the spectrum of $H_{\boldsymbol{k}}^{(2,1)}$, one has to analyze the constraints imposed by both $\mathcal{P}^\textrm{V,F}$~and~$\mathcal{P}^\textrm{E}$.

We begin by considering the inversion $\mathcal{P}^\mathrm{V}$ with respect to a vertex (which is equivalent to considering $\mathcal{P}^\textrm{F}$). 
By specifically taking the vertex in the center of the primitive cell, we recognize from the arrangement in Fig.~\ref{fig:2-supercell} that $\mathcal{P}^\textrm{V}$ permutes the subscripts of the translation generators as follows:
\begin{equation}
(1,2,3,4,5,6,\alpha,\beta)\mapsto (5,6,\alpha,\beta,1,2,3,4).    
\end{equation}
For 1D IRs, the relators among the eight generators of $\Gamma^{(2)}$ imply that $k_\alpha = -k_1 - k_3 - k_5$ and $k_\beta = -k_2 - k_4 - k_6$~\cite{Lenggenhager2023}, i.e., they are uniquely specified by the other six momentum components.
Therefore, $\mathcal{P}^\textrm{V}$ transforms momentum components as $\boldsymbol{k} \mapsto  M_{\mathcal{P}^\textrm{V}} \cdot \boldsymbol{k}$ with 
\begin{equation}
M_{\mathcal{P}^\textrm{V}} = \left(\begin{array}{cccccc}
0 & 0 & 0 & 0 & 1 & 0 \\
0 & 0 & 0 & 0 & 0 & 1 \\
-1 & 0 & -1 & 0 & -1 & 0 \\
0 & -1 & 0 & -1 & 0 & -1 \\
1 & 0 & 0 & 0 & 0 & 0 \\
0 & 1 & 0 & 0 & 0 & 0 \\
\end{array}\right).
\end{equation}
We find that $\mathcal{P}^\textrm{V}$ constraints the hyperbolic Bloch Hamiltonian for $2$-supercell as
\begin{equation}
H^{(2,1)}_{M_{\mathcal{P}^\mathrm{V}}\cdot\boldsymbol{k}} = U_{\mathcal{P}^\mathrm{V},\boldsymbol{k}} \cdot H^{(2,1)}_{\boldsymbol{k}} \cdot U_{\mathcal{P}^\mathrm{V},\boldsymbol{k}}^\dagger,
\end{equation}
where
\begin{equation}
U_{\mathcal{P}^\mathrm{V},\boldsymbol{k}} = \left(\begin{array}{cc}
1   &   0   \\
0   &   \mathrm{e}^{-\mathrm{i}(k_2+k_3+k_4+k_5)}
\end{array}\right)\otimes \GM{5}.\label{eqn:UPV}
\end{equation}
We further identify $\mathcal{P}^\mathrm{V}$-invariant momenta by solving for $\boldsymbol{k}\stackrel{!}{=}M_{\mathcal{P}^\textrm{V}}\cdot \boldsymbol{k}$ (mod $2\pi$ in each of the six components), which defines four two-dimensional planes
\begin{eqnarray}
\mathcal{M}_{s_1,s_2} &=& \{(k_1,k_2,-k_1+s_1,-k_2+s_2,k_1,k_2)|\ldots \nonumber \\
&\phantom{=}&\qquad\ldots|
k_{1,2}\in[-\pi,+\pi )\} \subset    
\textrm{BZ}^{(2,1)}
\end{eqnarray}
where $s_{1,2}\in\{0,\pi\}$. 
Within these planes, the exponential in Eq.~(\ref{eqn:UPV}) evaluates to \begin{equation}
\mathrm{e}^{-\mathrm{i}(k_2+k_3+k_4+k_5)}|_{\mathcal{M}_{s_1,s_2}}=\mathrm{e}^{-\mathrm{i}(s_1+s_2)} \in\{\pm 1\},
\end{equation}
and the eigenvalues of $U_{\mathcal{P}^\mathrm{V}}$ within the high-symmetry planes are $\pm 1$, each with multiplicity four. 

If inversion of bands with opposite $\mathcal{P}^\mathrm{V}$ eigenvalue occurs inside $\mathcal{M}_{s_1,s_2}$, then these bands are prevented from hybridization and a $\mathcal{P}^\textrm{V}$-protected nodal line is formed inside the plane $\mathcal{M}_{s_1,s_2}$. 
Crucially, as $|m|$ decreases across the critical value $|m_c|=4$, band inversion occurs inside the 4D BZ of $H_{\boldsymbol{k}}^{(1,1)}$ at  $\boldsymbol{k}^\textrm{(4D)}=(0,0,0,0)$ (for $m_c = -4$) resp.~at $\boldsymbol{k}^\textrm{(4D)}=(\pi,\pi,\pi,\pi)$ (for $m_c = +4$).
Owing to the zone-folding of immersion $\iota$ in Eq.~(\ref{eqn:immerse-4D-in-6D}), both four-momenta are mapped to the \emph{same} six-momentum $\boldsymbol{k}=\boldsymbol{0}$ inside the plane $\mathcal{M}_{0,0}$.
Therefore, we conclude that the topological transitions at $|m_c|=4$ are associated with the formation of a nodal line inside the plane $\mathcal{M}_{0,0}$. 
This prediction, rooted in symmetry analysis, agrees with the analytical calculation presented in Appendix~\ref{sec:k.p}.

Next, we consider implications of the inversion $\mathcal{P}^\textrm{E}$ with respect to the mid-point of an edge of the $\{8,8\}$ lattice (green triangle in Fig.~\ref{fig:2-supercell}).
To understand how $\mathcal{P}^\textrm{E}$ acts on the momenta inside BZ$^{(2,1)}$, we use GAP~\cite{GAP4,LINS} to express conjugations $\mathcal{P}^\textrm{E} \tilde{\gamma}_j \big(\mathcal{P}^\textrm{E}\big)^{-1}$ of the generators $\{\tilde{\gamma}_j\}_{j=1}^6$ of the translation group $\Gamma^{(2)}$ as composite words in $\{\tilde{\gamma}_j\}_{j=1}^6$~\cite{Chen2023b}. We obtain
\begin{subequations}
\begin{eqnarray}
\tilde{\gamma}_1 &\mapsto& \tilde{\gamma}_1^{-1}, \\
\tilde{\gamma}_2 &\mapsto& \tilde{\gamma}_2^{-1}, \\
\tilde{\gamma}_3 &\mapsto& \tilde{\gamma}_2\tilde{\gamma}_3^{-1}\tilde{\gamma}_2^{-1}, \\
\tilde{\gamma}_4 &\mapsto& \tilde{\gamma}_2\tilde{\gamma}_3\tilde{\gamma}_4^{-1}\tilde{\gamma}_3^{-1}\tilde{\gamma}_2^{-1}, \\
\tilde{\gamma}_5 &\mapsto& \tilde{\gamma}_2\tilde{\gamma}_3\tilde{\gamma}_4\tilde{\gamma}_5^{-1}\tilde{\gamma}_4^{-1}\tilde{\gamma}_3^{-1}\tilde{\gamma}_2^{-1}, \\
\tilde{\gamma}_6 &\mapsto& \tilde{\gamma}_2\tilde{\gamma}_3\tilde{\gamma}_4\tilde{\gamma}_3^{-1}\tilde{\gamma}_1^{-1}\tilde{\gamma}_6^{-1}\tilde{\gamma}_4^{-1}\tilde{\gamma}_2^{-1}\tilde{\gamma}_1,
\end{eqnarray}
\end{subequations}
which at the level of Abelian representations $\rho(\tilde{\gamma}_j)=\mathrm{e}^{\mathrm{i}k_j}$ simplifies to a diagonal point-group matrix $M_{\mathcal{P}^\textrm{E}}=-\mathbbm{1}_6$. 
Since $\mathcal{P}^\textrm{E}$ exchanges the two sites within the $2$-supercell, we find that 
\begin{equation}
H^{(2,1)}_{M_{\mathcal{P}^\mathrm{E}}\cdot\boldsymbol{k}} = U_{\mathcal{P}^\mathrm{E}} \cdot H^{(2,1)}_{\boldsymbol{k}} \cdot U_{\mathcal{P}^\mathrm{E}}^\dagger,
\end{equation}
where
\begin{equation}
U_{\mathcal{P}^\mathrm{E}} = \sigma_{1} \otimes \GM{5}.\label{eqn:UPE}
\end{equation}
Importantly, note that the composition $\mathcal{P}^\textrm{E}\mathcal{T}$ acts trivially on BZ$^{(2,1)}$: 
\begin{equation}
M_{\mathcal{P}^\textrm{E}\mathcal{T}}=M_{\mathcal{P}^\textrm{E}}\cdot M_{\mathcal{T}}=\mathbbm{1}_6,
\end{equation}
while inside the Hilbert space of Bloch states at $\boldsymbol{k}$ it acts by an antiunitary transformation:
\begin{equation}
U_{\mathcal{P}^\textrm{E}\mathcal{T}}=U_{\mathcal{P}^\textrm{E}}\cdot U_{\mathcal{T}}=\sigma_{1} \otimes (\GM{1}\GM{3}) \mathcal{K}.
\end{equation}
Since $(U_{\mathcal{P}^\textrm{E}\mathcal{T}})^2=-\mathbbm{1}_6$, we conclude that Hamiltonian matrices $H_{\boldsymbol{k}}^{(2,1)}$ belong to the symplectic class (nodal class AII of Ref.~\cite{Bzdusek:2017}) whose spectrum exhibits Kramers degeneracy and the band nodes are generically of codimension $\mathfrak{d}=5$. 
This is consistent with the observed formation of $(6-5)=1$-dimensional nodal lines inside BZ$^{(2,1)}$, as illustrated with numerically extracted data in Fig.~\ref{fig:evol}. 

As discussed in the \emph{Nodal ring} section of the main text, if one assumes a linear dispersion of the energy bands near the band node in all directions perpendicular to the node, then the expected DOS scaling is $\rho(E)\propto E^{\mathfrak{d}-1}$, i.e., we anticipate $\rho(E)\propto E^{4}$ for the hyperbolic non-Abelian semimetal.
While this theoretical prediction appropriately describes the scaling of numerically obtained DOS data for $m \in (3,4) $, where the $2$-supercell description identifies a single nodal ring, we observe in Fig.~\ref{fig:dos_scaling} a noticeable deviation to somewhat lower values of the DOS exponent for $m\in(0,3)$, where the nodal manifold acquires a complicated structure. 
This observation suggests that the assumption on linear band dispersion fails in some directions near some locations on the nodal manifold: either for symmetry reasons (e.g, nodal line located along a high-symmetry line, where symmetry forbids linear terms), or due to intersections of multiple nodal rings (i.e., nodal chain), or accidentally (which may be expected due to the large extent of the nodal manifold). 

We finally tackle the particular case of $m=0$, where Fig.~\ref{fig:evol} suggests a higher-dimensional nodal manifold, and Fig.~\ref{fig:dos_scaling} indicates an altered scaling of the density of states. 
For this choice of the mass parameter, an additional sublattice symmetry $\mathcal{S}$ arises, which leaves the six-dimensional momentum $\boldsymbol{k}$ invariant ($M_{\mathcal{S}}=\mathbbm{1}_6$). 
It acts on the Hilbert space by a unitary matrix: 
\begin{equation}
U_{\mathcal{S}}=\sigma_{3} \otimes\mathbbm{1}_4,
\end{equation}
and that flips the sign of energy:
\begin{equation}
U_{\mathcal{S}} \cdot H_{\boldsymbol{k}}^{(2,1)} \cdot U_{\mathcal{S}}^{-1} = -H_{\boldsymbol{k}}^{(2,1)}.  
\end{equation}
Since $\mathcal{S}$ anticommutes with $\mathcal{P}^\mathrm{E}\mathcal{T}$ (i.e., the ``sublattice parity is odd'' per the terminology of Ref.~\cite{Bzdusek:2017}), the Hamiltonian $H_{\boldsymbol{k}}^{(2,1)}$ belongs to nodal class $\textrm{DIII}$ at $m=0$. 
The sublattice symmetry thereby reduces the codimension for node formation to $\mathfrak{d}=2$~\cite{Bzdusek:2017}, which implies that the nodal manifold becomes four-dimensional in BZ$^{(2,1)}$ (so that it completely fills the three-dimensional projection in the $m=0$ panel of Fig.~\ref{fig:evol}).
Assuming again the linear dispersion in the two direction perpendicular to the nodal manifold, the expect DOS scaling is $\rho(E) = \int \mathrm{d}^2k_\perp \delta(E - v |k_\perp|) \propto E$.  
This theoretical prediction is consistent with the observed drop in the numerically extracted exponent visible in Fig.~\ref{fig:dos_scaling}.

\subsection{Density of states for larger $n$-supercells}\label{sec:larger-sc}

\subsubsection{Density of states from random sampling}\label{sec:DOS-from-sampling}

\begin{figure*}[t]
\includegraphics{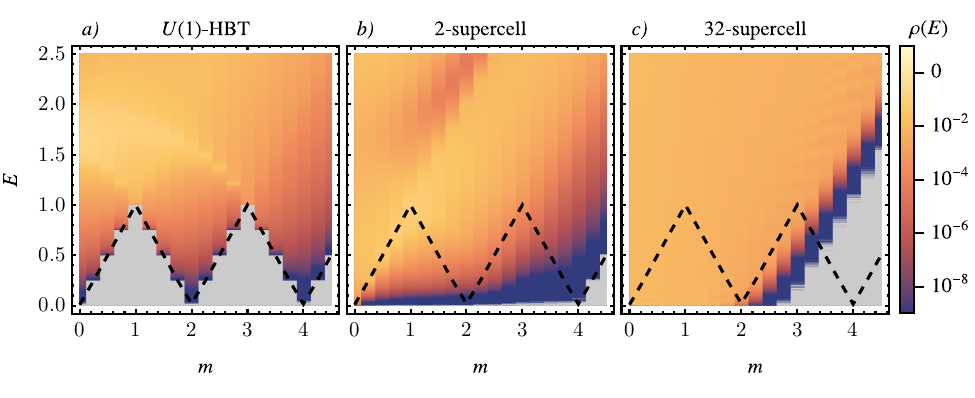}
    \caption{
    Density of Abelian states $\rho(E;m)$ obtained from (a) the primitive cell, corresponding to $\textrm{U}(1)$ hyberbolic band theory (HBT), (b) the $2$-supercell, and (c) the $32$-supercell as a function of energy $0\leq E\leq 2.5$ and mass parameter $0\leq m\leq 4.5$ from random sampling with (a,c) $10^9$ and (b) $10^{10}$ Abelian momenta.
    The density is shown on a logarithmic scale (legend on the right)
    Zero values, indicating a true gap in the spectrum for that particular (super-)cell, are shown in gray.
    The black dashed line indicates the exact gap in the density of $\textrm{U}(1)$ states.
    }
    \label{fig:dos_scs_vs_m}
\end{figure*}

We systematically study the low-energy DOS as a function of supercell size $n$ and mass parameter $m$ to gain some understanding of the possible phase diagrams arising from the inclusion of higher-dimensional irreducible representations and thus at larger system sizes.
To that end, for each of the $n$-supercells defined in \cref{eqn:supercell-PGs}, we randomly sampled $10^9$ ($10^{10}$ for $n=2$) $\bk$-points in their $(2n+2)$-dimensional $\textrm{BZ}^{(n,1)}$ of $\mathrm{U}(1)$ representations, resulting in $N=4n\times 10^9$ ($8\times 10^{10}$ for $n=2$) eigenenergies.
We then count the collected eigenenergies, producing a histogram of counts $N_i$ associated with bins of width $\Delta E = 0.001$ centered at energies in the range $E_i\in[-8.0005,8.0005]$.

Ignoring potential correlations from states at the same Abelian momentum, we assume that the $N$ states characterized by energy $E$ are independent samples from the probability distribution $\rho(E)$, which is precisely the density of states.
Thus, the probability of finding a randomly sampled state in the bin centered at $E_i$ is
\begin{equation}
    p_i = \int_{E_i-\Delta E/2}^{E_i+\Delta E/2}\dd{E}\rho(E) \approx \rho(E_i)\Delta E.
\end{equation}
The probability of observing $N_i$ states in the bin centered at $E_i$ when sampling $N$, thus, is given by the binomial distribution $\mathcal{B}(N,p_i)$:
\begin{equation}
    \mathrm{P}(N_i)=\binom{N}{N_i}p_i^{N_i}(1-p_i)^{N-N_i},
\end{equation}
with mean $Np_i=N\rho(E_i)\Delta E$ and variance $\sigma^2(N_i)=Np_i(1-p_i)$.
This allows us to estimate $p_i$ via the estimator $\hat{p}_i = N_i/N$ and thus $y_i=p_i/\Delta E\approx\rho(E_i)$ via the estimator
\begin{equation}
    \hat{y}_i = \frac{\hat{p}_i}{\Delta E} = \frac{N_i}{N\Delta E}.
\end{equation}
The associated standard deviation is
\begin{equation}
    \begin{split}
        \sigma(y_i) &= \frac{\sigma(N_i)}{N\Delta E} = \frac{\sqrt{Np_i(1-p_i)}}{N\Delta E}\\
        &= \sqrt{\frac{y_i(1/\Delta E-y_i)}{N}},
    \end{split}
\end{equation}
which can be estimated by
\begin{equation}
    \hat{\sigma}(y_i) = \sqrt{\frac{\hat{y}_i(1/\Delta E-\hat{y}_i)}{N}}.
    \label{eq:dos-data_variance}
\end{equation}
The resulting data $\hat{y}_i$ representing $\rho(E)$ is shown in \cref{fig:dos_scs_vs_m} in the energy range $0\leq E\leq 2.5$ as a function of the mass parameter $0\leq m\leq 4.5$ for ({a}) $\textrm{U}(1)$-HBT, ({b}) the $2$-supercell, and ({c}) the $32$-supercell.

For $\textrm{U}(1)$-HBT [\cref{fig:dos_scs_vs_m}(a)] we find perfect agreement with the expected semimetallic gap closing at $m=0,2,4$ and with insulators at the intermediate mass parameters with maximal gaps at $m=1,3$.
\Cref{fig:dos_scs_vs_m}(b) shows that there are non-Abelian in-gap states across the full range $0<m<4$ with a semimetallic (i.e., power-law-scaling) density of states.
In addition, we observe a quantitative depletion of DOS in the immediate neighborhood of $E=0$ with increasing $m>2$ (which could be attributable to shrinking of the nodal manifold, cf.~Fig.~\ref{fig:evol}) and a gap opening after $m=4$ (which corresponds to disappearance of the nodal-line ring), with a band edge that follows the Abelian band edge for $m>4$.
Note that despite the quantitative depletion, we have seen in \cref{fig:dos_scaling} that the \emph{scaling} of the DOS with $E$ remains approximately constant, with exponent $\alpha\in(3.3,4.1)$.

Going to larger supercells and finally to the $32$ supercell, for which the data is shown in \cref{fig:dos_scs_vs_m}(c), we observe a qualitatively different behavior at different values of the mass parameter~$m$.
First, in the range $0<m\lesssim 2$, the Abelian gap is clearly filled up with more and more non-Abelian states seemingly converging towards a metallic DOS in the thermodynamic limit.
On the other hand, the region $3\lesssim m\leq 4$ is depleted of further states (as is the trivially insulating region $m>4$), such that we expect an insulator in the thermodynamic limit.
The transition between the two phases most probably lies in the range $2<m<3$. 
More detailed discussion of our findings appears in \cref{sec:larger-sc}.

Before concluding, two remarks are in order.
First, the non-Abelian states arising from finite-dimensional irreducible representations, in particular those shown in \cref{fig:dos_scs_vs_m}(b), are still present in the thermodynamic limit, even if they occur with a vanishingly small weight.
Second, the density of Abelian states of any of the studied supercells ($2\leq n\leq 32$) in the range $0\leq m\leq 4$ is, strictly speaking, always \emph{semimetallic} with 
\begin{itemize}
    \item[(\emph{i})] a narrow dip at $E=0$ due to inherent numerical artifacts of the supercell technique (see Sec.~II.A.3 of Supplemental Material to Ref.~\cite{Lenggenhager2023}) in the metallic regime,~and 
    \item[(\emph{ii})] with heavily suppressed tails in the insulating~regime.
\end{itemize}
Therefore, special care is necessary to distinguish insulating resp.~metallic spectra from truly semimetallic ones.

\subsubsection{Fitting the low-energy density of states}\label{sec:alpha-fitting}

For a more quantitative characterization of the DOS arising from the different supercells, we repeat the fitting procedure applied in the earlier \cref{app:diff_m} to the $2$-supercell, assuming the model $\rho(E\to 0)\propto E^\alpha$.
Because we are not directly fitting the DOS $\rho(E)$ but a \emph{histogram}, i.e., $\rho(E)$ integrated over disjoint energy ranges, we have to be careful in the range $0<m\lesssim 3$, where the finite bin-width might have smeared out the anticipated DOS dip [cf.~the remark~(\emph{i}) above], such that the histogram nevertheless indicates a finite $\rho(E=0)$.
To account for these subtleties, we apply a multistep fitting strategy as outline below.

In the first step, we determine whether the DOS \emph{data} (for the given bin width) is metallic or semimetallic.
To that end, we consider the \emph{cumulative} density of states (with positive energy),
\begin{equation}
    F_\rho(E) = \int_0^{E}\dd{E'}\rho(E'),
\end{equation}
which can be obtained from the observed histogram as follows (assuming $i=0$ corresponds to the bin centered at $E=0$):
\begin{equation}
    z_i = F_\rho(E_i+\Delta E/2) = \Delta E\left(\frac{y_0}{2} + \sum_{0<j\leq i}y_j\right),
\end{equation}
where we made use of the spectral symmetry $\rho(-E)=\rho(E)$ to split the bin at zero energy.
Going to the cumulative density of states has two main advantages:
First, it averages over the oscillations that are typically observed in the Abelian DOS of finite supercells~\cite{Lenggenhager2023} and second, it removes the bias of the binning, since the $z_i$ are actually evaluations of $F_\rho$ at certain energies, unlike $y_i$ which only correspond to $\rho(E_i)$ in the limit $\Delta E\to 0$.

Next, we perform a linear least-square fit of the cumulative density of states as a function of energy in a log-log-scale, i.e., we set up the following linear model
\begin{equation}
    \log(z_i) = (\alpha+1) \log(E_i+\Delta E/2) + \beta
\end{equation}
obtained from integrating the model $\rho(E\to 0)\propto E^\alpha$.
We constrain the fit to a very small energy window $[E_\text{min},E_\text{max}]$ near $E=0$ chosen such that $F_\rho(E_\text{min})\geq 10^{-9}$ and $E_\text{max}=E_\text{min}+15\Delta E$.
A linearly increasing cumulative DOS indicates a metallic behavior, $\alpha=0$, allowing us to detect those cases; here we set a threshold $\alpha<0.5$.
For the cases that are deemed to be metallic, the extracted values $\alpha$ are shown in \cref{fig:dos_scaling_scs} as orange data points.
Note that parameter errors obtained from the least-square fit have to be interpreted with care.
Since the fitting is done in a log-log-scale and furthermore even the errors in the original data are not normally distributed (which we partially take into account for the second step as discussed below), the statistical assumptions that would allow us to interpret the parameter errors as confidence intervals are not satisfied.
Nevertheless, the errors do give a rough (relative) idea of how close the scaling is to zero and therefore we include error bars for the orange data points in \cref{fig:dos_scaling_scs}.

In the second step of the procedure, we perform a non-linear least-square fit of the DOS data $\hat{y}_i$.
Both the model and the energy window $[E_\text{min},E_\text{max}]$ depend on the previous estimate of $\alpha$.
The lower bound of the energy window is now determined by the condition $\rho(E\geq E_\text{min})>5\times 10^{-8}$, while the width $E_\text{max}-E_\text{min}$ is $15\Delta E$ for $\alpha\leq 1.2$ (with the exception of $m=0$, where it is always $10\Delta$) and $100\Delta E$ for $\alpha>1.2$.
The motivation for these different choices of the fitting range is the significantly larger deviation from a power law behavior in the metallic regime $m\lesssim 3$ due to subleading corrections.
In particular, we explicitly include subleading corrections to the constant term in the fit model for the metallic regime, cf.~\cref{eq:dos-fit_metallic}, despite the significantly smaller energy range.

\begin{figure*}[p]
\includegraphics{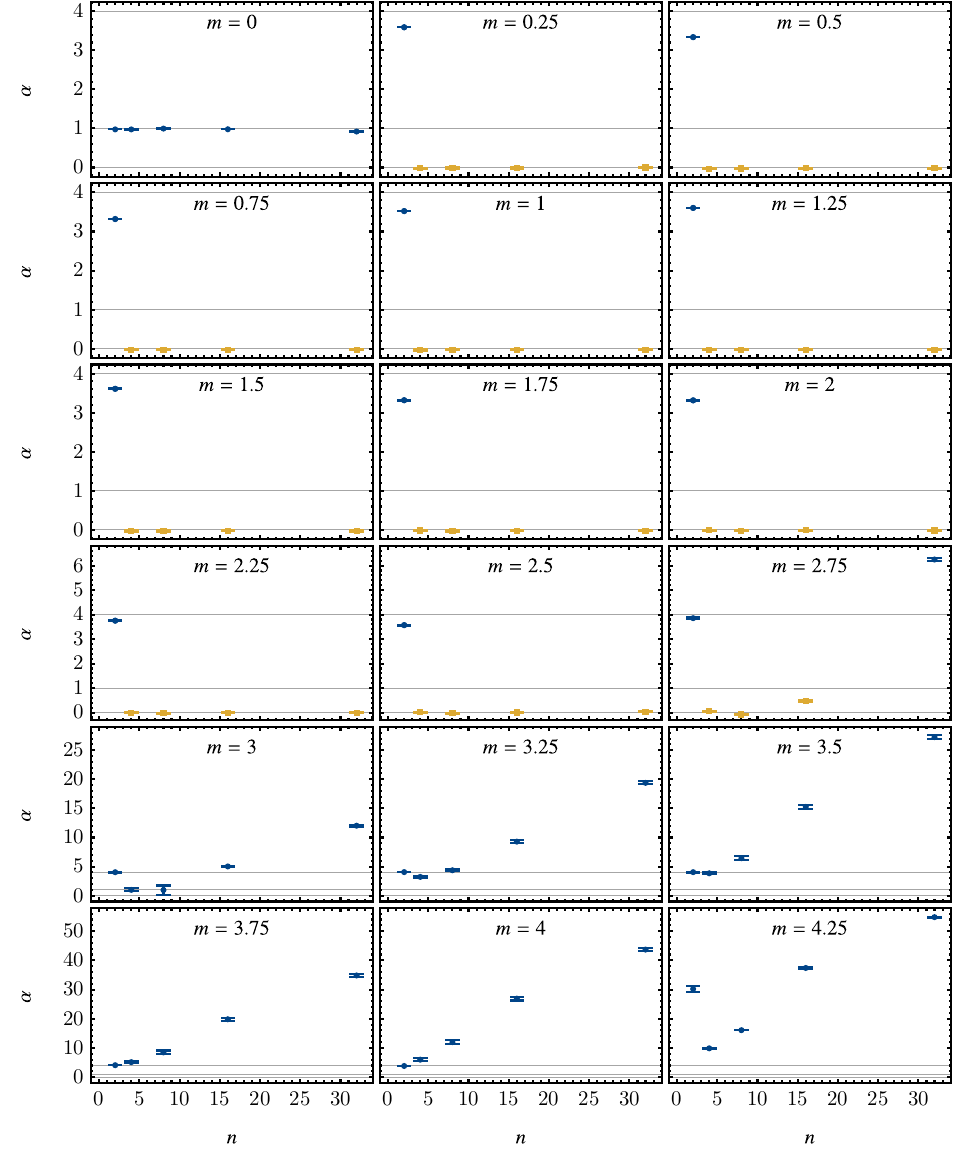}
    \caption{
    Low-energy scaling exponent $\alpha$ as function of supercell size $n$ for different values of the mass parameter $m$, obtained using the fitting procedure described in \cref{sec:alpha-fitting}.
    The blue data points have been obtained from a weighted non-linear least-square fit of the $n$-supercell density of states with model $\rho(E\to 0)\propto E^\alpha$ and the error bars indicate the standard errors but should be interpreted with care.
    The orange data points have been obtained from a linear least-square fit of the cumulative density of states with model $\rho(E\to 0)\propto E^{\alpha+1}$ on a log-log-scale.
    The horizontal gray lines indicates the following values: $0$, $1$, and $4$.
    }
    \label{fig:dos_scaling_scs}
\end{figure*}

\begin{figure*}[t]
\includegraphics{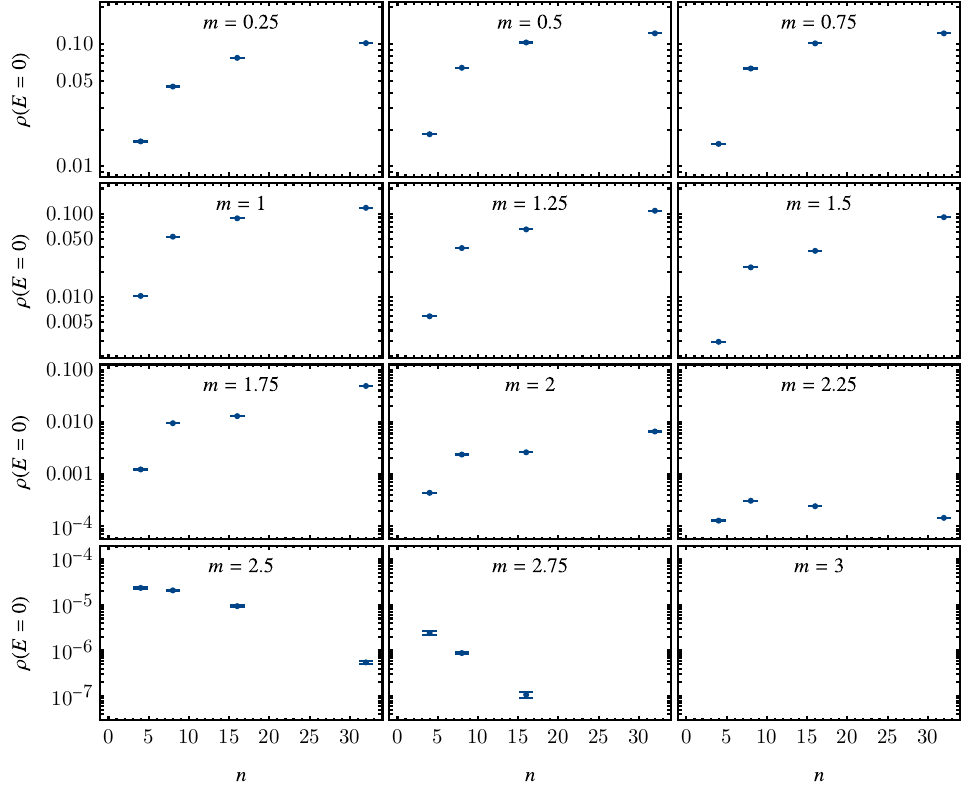}
    \caption{
    Density of states at zero energy $\rho(E=0)$ as function of supercell size $n$ for different values of the mass parameter $m$.
    The data have been obtained from a weighted non-linear least-square fit of the $n$-supercell density of states with model $\rho(E\to 0)=\rho_0 + \rho_1E^{\alpha_1}$ and the error bars indicate the standard errors but should be interpreted with care.
    Note that only data for $(n,m)$ with a metallic density of states are shown.
    In particular, $n=32$ for $m=2.75$ and any $n$ for $m=3$ are semimetallic such that no data points are shown.
    }
    \label{fig:dos_offset_scs}
\end{figure*}

To compensate for the variation of the variance of $y_i$ with $i$, a weighted fitting procedure is used with weights proportional to the inverse of the standard deviation:
\begin{equation}
    w_i \propto \frac{1}{\hat{\sigma}(y_i)} \approx \sqrt{\frac{N\Delta E}{\hat{y}_i}}\propto \frac{1}{\sqrt{\hat{y}_i}},
\end{equation}
where we used \cref{eq:dos-data_variance} and that typically $\hat{y}\Delta E\ll 1$.
Note that this does not compensate for the non-normality of the binomial distribution according to which the $y_i$ are distributed.
While the binomial distribution can be approximated by a normal distribution with same mean and variance for large $N$, such an approximation is typically not very good if $p_i$ is close to $0$ or $1$.
Thus, while the parameter errors are not unreasonable, they have to be interpreted with care.

The non-linear models used in the fitting are as follows.
In the metallic regime, we assume
\begin{equation}
    \hat{y}_i = \rho_0 + \rho_1 E^{\alpha_1},
    \label{eq:dos-fit_metallic}
\end{equation}
while in the semimetallic/insulating regime, we assume
\begin{equation}
    \hat{y}_i = \rho_1 E^\alpha.
\end{equation}
Note that $\alpha_1$ is \emph{not} the leading order scaling but we force a term $\rho_0$ due to our previous estimate of $\alpha$ resulting in a value close to zero, see orange data points in \cref{fig:dos_scaling_scs}.
Thus, we obtain estimates of $\alpha$ for the semimetallic/insulating regime as well; these are shown as blue data points in \cref{fig:dos_scaling_scs}.
Finally, in the metallic regime $0<m<3$, we extract estimates of $\rho_0$, corresponding to $\rho(E=0)$, which are shown in \cref{fig:dos_offset_scs}. 
These can be interpreted as a measure of the `Fermi surface extent'.

\subsubsection{Metal-to-insulator transition for large $n$-supercells}

Let us finally discuss how the trends observed in Fig.~\ref{fig:dos_scs_vs_m}(c) and described in Sec.~\ref{sec:DOS-from-sampling} are correlated with the in-depth analysis of the DOS scaling shown in \cref{fig:dos_scaling_scs,fig:dos_offset_scs}. 

For $0 < m \lesssim 2$ our fitting procedure indicates a metallic DOS (yellow data in \cref{fig:dos_scaling_scs}) for all supercells with $n>2$, with $\rho(E=0)$ (\cref{fig:dos_offset_scs}) that monotonously grows with the supercell size $n$. Therefore, we anticipate the model to be in the metallic phase in the thermodynamic limit for this range of the mass parameter. Furthermore, observe that $\rho(E=0)$ for the 32-supercell peaks for $m\in(0.5, 0.75)$, and shrinks to a smaller value at $m=0.25$. 
This could be interpreted as a shrinking of the Fermi surface extent, vanishing altogether at $m=0$ where the DOS appears semimetallic for all studied $n$-supercells.

The points $2<m\lesssim 3$ appear to be close to the metal-to-insulator transition. 
Specifically, for $m=2.25$ the zero-energy DOS $\rho(E=0)$ seems to converge to a small value below $2\times 10^{-4}$ (potentially to zero). 
For $m=2.5$, although the fits for all supercells with $n>2$ appear metallic in \cref{fig:dos_scaling_scs}, a look at \cref{fig:dos_offset_scs} indicates a rapidly shrinking $\rho(E=0)$. 
Therefore, we expect that the point $m=2.5$ is already on the insulating side of the phase boundary. 
For $m=2.75$, the DOS appears metallic for the intermediate values of the supercell size $n\in\{4,8,16\}$, but ultimately becomes semimetallic with large $\alpha \approx 6.3$ for the $32$-supercell, suggesting an insulating phase in the thermodynamic limit. 

The proximity to the critical point still leaves a mark on the extracted $\alpha(n)$ for $m=3$.
Namely, while the $2$-supercell suggests semimetallic $\alpha \approx 4$, the $4$-supercell and the $8$-supercell are characterized by a significantly smaller DOS scaling exponent $\alpha \approx 1$, reminiscent of the onset of the metallic phase at lower values of $m$. 
This trend at intermediate values of $n$ is overturned for supercells with $n\in\{16,32\}$. 
A similar but much less pronounced dip in $\alpha$ is also observed for the $4$-supercell at $m=3.25$.
Finally, for $m \gtrsim 3.25$, the fitted scaling exponent $\alpha$ grows dramatically with $n$, providing a clear sign of an insulating phase in the thermodynamic limit.

\end{document}